\documentclass{article}
\usepackage[english]{babel}
\usepackage[letterpaper,top=2cm,bottom=2cm,left=3cm,right=3cm,marginparwidth=1.75cm]{geometry}
\usepackage{amsmath}
\usepackage{amsthm}
\usepackage{amssymb}
\usepackage{graphicx}

\usepackage{algorithm}  
\usepackage{algpseudocode}  
\usepackage{algorithmicx} 
\usepackage{makecell}
\usepackage{booktabs}
\usepackage{multirow}
\usepackage{diagbox}
\usepackage{cite}
\theoremstyle{definition}
\newtheorem{theorem}{Theorem}
\newtheorem{definition}{Definition}
\newtheorem{lemma}{lemma}
\usepackage[colorlinks=true, allcolors=blue]{hyperref}
\usepackage{caption}
\usepackage{subcaption}
\usepackage{ulem}
\usepackage{indentfirst}
\usepackage{authblk}
\usepackage{enumerate}
\usepackage{hyperref}
\usepackage[capitalize]{cleveref}

\title{Low Rank Quaternion Matrix Completion Based on Quaternion QR Decomposition and Sparse Regularizer}
\author[a]{Juan Han}
\author[a]{Liqiao Yang}
\author[a]{Kit Ian Kou\thanks{Corresponding author: kikou@umac.mo }}
\author[b]{Jifei Miao}
\author[c]{Lizhi Liu}
\affil[a]{Department of Mathematics, Faculty of
	Science and Technology, University of Macau, Macau 999078, China}
\affil[b]{School of Mathematics and Statistics, Yunnan University, Kunming, Yunnan, 650091, China}	
\affil[c]{Department of Radiology, Sun Yat-sen University Cancer Center, State Key Laboratory of Oncology in South China, Collaborative Innovation Center for Cancer Medicine, Guangdong Key Laboratory of Nasopharyngeal Carcinoma Diagnosis and Therapy, 651 Dongfeng Road East, Guangzhou, Guangdong, 510060, China}

\date{}

\begin{document}

\captionsetup[figure]{labelfont={},labelformat={default},labelsep=period,name={Fig.}}

\maketitle

\begin{abstract}
Matrix completion is one of the most challenging problems in computer vision. Recently, quaternion representations of color images have achieved competitive performance in many fields. Because it treats the color image as a whole, the coupling information between the three channels of the color image is better utilized. Due to this, low-rank quaternion matrix completion (LRQMC) algorithms have gained considerable attention from researchers. In contrast to the traditional quaternion matrix completion algorithms based on quaternion singular value decomposition (QSVD), we propose a novel method based on quaternion Qatar Riyal decomposition (QQR). In the first part of the paper, a novel method for calculating an approximate QSVD based on iterative QQR is proposed (CQSVD-QQR), whose computational complexity is lower than that of QSVD. The largest $r \ (r>0)$ singular values of a given quaternion matrix can be computed by using CQSVD-QQR. Then, we propose a new quaternion matrix completion method based on CQSVD-QQR which combines low-rank and sparse priors of color images. Experimental results on color images and color medical images demonstrate that our model outperforms those state-of-the-art methods.   
\end{abstract}


 \section{Introduction}
Image completion, which aims to recover lost entries based on limited known pixels of an image, has a wide range of applications in computer vision. It has received extensive attention from researchers in recent years \cite{hu2012fast, liu2015truncated, gu2017weighted, miao2021color}. Low-rank matrix minimization problems are constructed and solved using low-rank properties of matrices. There has been great success with low-rank matrix completion (LRMC)-based methods. The low-rank property of LRMC-based methods is generally controlled by a constrained rank minimization problem. 
Let $\mathbf{M} \in \mathbb{R}^{M \times N}, \ M \geq N >0$ be an incomplete matrix, and then the traditional LRMC-based method can be formulated as follows:
\begin{equation}
\mathop{\text{min}}\limits_{\mathbf{X}}\text{rank}(\mathbf{X}),   \: \text{s.t.}, \ {P}_{\Omega}(\mathbf{X}-\mathbf{M})=\mathbf{0}, 
\label{LRMC_TRA}
\end{equation}
where $\mathbf{X}\in \mathbb{R}^{M \times N}$ is a completed output matrix, rank$(\cdot)$ is the rank function, and $\Omega$ is the observed entries set.
${P}_{\Omega}(\mathbf{X})$ is defined as
\[({P}_{\Omega}(\mathbf{X}))_{mn}=\begin{cases}
	\mathbf{X}_{mn}, \ (m,n)\in \Omega,\\
	0, \  \  \  \quad  (m,n)\notin \Omega.
\end{cases}
\]
However, the rank function is discontinuous and non-convex, making it NP-hard. The nuclear norm has been shown to be the tightest convex relaxation of the above rank minimization problem \cite{Cands2009Exact}. Thus, the nuclear norm as a convex surrogate of the rank function is frequently used to deal with the image completion problem \cite{Cai2010Asingular, Toh2010}. 
The corresponding minimization method is formulated as 
\begin{equation}
\mathop{\text{min}}\limits_{\mathbf{X}}\|\mathbf{X}\|_{\ast},   \: \text{s.t.}, \ {P}_{\Omega}(\mathbf{X}-\mathbf{M})=\mathbf{0}, 
\label{LRMC_NNM}
\end{equation}
However, the authors in \cite{hu2012fast} stated that these existing methods based on nuclear norm minimization minimize all singular values simultaneously, so the rank may not be well approximated. And they proposed the truncated nuclear norm (TNN) regularization, which approximates the rank function more accurately than the nuclear norm. Gu $\it{et}$ $ \it{al}$. \cite{Gu2014CVPR, gu2017weighted} introduced a weighted nuclear norm (WNNM) technique that gives different weights to singular values in order to treat them differently. Additionally, there are other enhanced nuclear norm-based approaches that have shown good image completion outcomes, such as \cite{Nie2012Low, LIU2014218, 7539605}.\\
\indent
When processing color images, the aforementioned LRMC-based approaches essentially process the RGB's three color channels individually before combining them to get the final restoration outcome. The loss of coupling information between the three color channels could result from this strategy's disregard for the relationships between the three channels.\\
\indent
Researchers have paid close attention to the quaternion representation of color images in recent years. Numerous studies have demonstrated the effectiveness of quaternions in describing color images, and techniques based on them have produced competitive outcomes for a variety of image processing issues, including color image edge identification \cite{Hu2018Phase}, color face recognition \cite{zou2016quaternion}, color image denoising \cite{yu2019quaternion}, color image completion\cite{miao2021color}. For a color image, its three color channels exactly match the three imaginary sections of a quaternion matrix. A pixel of a color image can be represented by a pure quaternion as follows:
\[ \dot p = 0+p_R \,  i+p_G \,  j +p_B \,   k,\]
where $\dot p$ denotes a color pixel, $p_R$, $p_G$, and $p_B$ are the pixel values of the three channels RGB of $\dot p$, respectively, and $i$, $j$, $k$ are the three imaginary units of a quaternion.
The relationship between the three channels can be better leveraged since the pure quaternion matrix represents the three color channels of the color image holistically. Recently, approaches for image completion using quaternion-based LRMC have also been proposed. Similar to the problem (\ref{LRMC_NNM}), the widely used low-rank quaternion matrix completion algorithm (LRQMC) is formulated as follows:
\begin{equation}
\mathop{\text{min}}\limits_{\dot{\mathbf{X}}}\|\dot{\mathbf{X}}\|_{\ast},   \: \text{s.t.}, \ {P}_{\Omega}(\dot{\mathbf{X}}-\dot{\mathbf{M}})=\mathbf{0}, 
\label{LRQMC_NNM}
\end{equation}
where $\dot{\mathbf{X}}\in \mathbb{H}^{M \times N}$ is the recovered output quaternion matrix, $\dot{\mathbf{M}}\in \mathbb{H}^{M \times N}$ is the observed quaternion matrix, and $\|\dot{\mathbf{X}}\|_{\ast}$ is the nuclear norm of a quaternion matrix $\dot{\mathbf{X}}$.
${P}_{\Omega}(\dot{\mathbf{X}})$ is defined as
\[({P}_{\Omega}(\dot{\mathbf{X}}))_{mn}=\begin{cases}
	\dot{\mathbf{X}}_{mn}, \ (m,n)\in \Omega,\\
	0, \  \  \  \quad  (m,n)\notin \Omega.
\end{cases}
\]
Chen $\it{et}$ $ \it{al}$. \cite{chen2019low} provided a general model for low-rank quaternion matrix completion based on the nuclear norm and three nonconvex rank surrogates (i.e., functions based on the Laplace function, Geman function, and weighted Schatten norm, respectively). These techniques can produce better results when compared to some traditional LRMC-based techniques. However, to use these techniques, one must first solve the quaternion singular value decomposition (QSVD) of the complete quaternion matrix, just as some other conventional LRMC-based techniques. 
Large processing matrices necessitate complex algorithms that take a long time to complete. 
Therefore, it would be wise to look into more effective algorithms.
Three LRQMC approaches are suggested in \cite{miao2020quaternion} based on the quaternion double Frobenius norm (Q-DFN), quaternion double nuclear norm (Q-DNN), and quaternion Frobenius/nuclear norm (Q-FNN) minimization models. This algorithm only needs to process two-factor quaternion matrices of lesser sizes after these three quaternion-based bilinear factorizations of the quaternion matrix, which lessens the computing difficulty of solving the QSVD. However, this factorization might become entangled in local minima \cite{cabral2013unifying}.\\
\indent
Furthermore, more information \cite{yang2010image, liu2018fast} should be taken into account for improved image completion rather than just relying on the low-rank qualities of images. The image completion problem can be impacted significantly by the sparsity of images in a particular domain, such as transform domains where many signals have a naturally sparse structure. To provide an approach for image completion, the authors \cite{dong2018low} combined low-rank and sparse priors. The $l_1$-norm regularizer is used to formulate the sparse prior, and the truncated nuclear norm is chosen as the surrogate of the rank function. The authors in \cite{yang2022quaternion} extended this method to the quaternion system. These methods can achieve relatively good completion results but still rely on the QSVD calculation of large matrices.\\
\indent
Actually, singular values and singular vectors can also be calculated by Qatar Riyal (QR) decomposition \cite{liu2018fast}. Inspired by this, in this paper, we propose a novel quaternion Tri-Factorization method and thus we only need to calculate the QSVD of a small quaternion matrix. We also take the sparsity into account to increase the accuracy of the completion results. The main contributions of this paper are summarized as follows:
\begin{itemize}
    \item A novel method for calculating an approximate QSVD (CQSVD-QQR) is proposed based on quaternion Qatar Riyal decomposition (QQR). The suggested approach can compute the largest $r (r>0)$ singular values of a given quaternion matrix, and its computational complexity is lower than QSVD's.
    \item Based on the novel quaternion Tri-Factorization method and sparse quaternion matrix priors, we propose a novel LRQMC model. The dimension of the quaternion matrix in the quaternion nuclear norm regularization term is reduced based on CQSVD-QQR.
\end{itemize}
Listed below is the organization of this article. In Section 2, numerous notations are introduced along with a brief overview of the fundamentals of quaternion algebra. Section 3 of the paper provides a brief discussion of the theory of the quaternion discrete cosine transform. The CQSVD-QQR technique for computing an approximate QSVD, which is based on the quaternion QR decomposition, is then given. After that our proposed quaternion-based matrix completion algorithm, QQR-QNN-SR is introduced. We also discuss the computational complexities of QQR-QNN-SR. In Section 4, we illustrate the convergence of CQSVD-QQR and provide some examples of our method's utility in real-world situations. We also compare our method with some state-of-the-art methods. The conclusion is presented in Section 5.

\section{Notations and preliminaries}
This section begins with a summary of various mathematical notations and ends with a brief review of some fundamental quaternion algebra concepts. In addition, we recommend \cite{girard2007quaternions} for a more comprehensive understanding of quaternion algebra.
\subsection{Notations}
In this paper, we use $\mathbb{R}$, $\mathbb{C}$, $\mathbb{H}$ to present the set of real numbers, the set of complex numbers, and the set of quaternions, respectively. In real and complex fields, scalars, vectors, and matrices are indicated by lowercase letters, such as, $a$, boldface lowercase letters, such as, $\mathbf{a}$, and boldface capital letters, such as, $\mathbf{A}$, respectively. The symbols $\dot{q}$, $\dot{\mathbf{q}}$, $\dot{\mathbf{Q}}$, respectively, stand for a quaternion scalar, vector, and matrix, respectively. The symbols $(\cdot)^{-1}$, $(\cdot)^{\dagger}$, $(\cdot)^{T}$, $(\cdot)^{\ast}$, and $(\cdot)^{H}$, respectively, stand for the inverse, pseudoinverse, transpose, conjugation, and conjugate transpose. The identity matrix is represented by $\mathbf{I}_{m} \in \mathbb{R}^{m \times m}$.

\subsection{Basic knowledge of quaternion algebra}
Quaternion was introduced by W.R. Hamilton \cite{hamilton1844ii} in 1843. A quaternion $\dot{q} \in \mathbb{H}$ can be represented as follows:
\[\dot{q}=q_0+q_1 i +q_2 j+q_3 k,\]
where $q_0$, $q_1$, $q_2$, $q_3 \in \mathbb{R}$, and $i$, $j$, $k$ are imaginary number units which obey the following multiplication rules:
\[\begin{cases}
	i^2 =j^2=k^2=ijk=-1,\\
	ij=-ji=k, jk=-kj=i, ki=-ik=j.
\end{cases}
\]
We can rewrite a quaternion $\dot{q}=q_0+q_1 i +q_2 j+q_3 k$ as $\dot{q}=q_0+q_1 i +(q_2 +q_3 i) j = c_1+c_2 j$, where $c_1$, $c_2 \in \mathbb{C}$. A quaternion $\dot{q} \in \mathbb{H}$ can also be decomposed into a real part $\mathfrak{R}(\dot{q})$ and a vector part $\mathbf{V}(\dot{q})$ as follows:
\[\dot{q} = \mathfrak{R}(\dot{q}) + \mathbf{V}(\dot{q}),\]
where $\mathfrak{R}(\dot{q})= q_0 \in \mathbb{R}$, and $\mathbf{V}(\dot{q}) = q_1 i +q_2 j+q_3 k$. Moreover, if $\mathfrak{R}(\dot{q}) =0$, then $\dot{q}$ is called a pure quaternion. And $\mathbf{V}(\mathbb{H})$ denotes the set of pure quaternions. 

The conjugate of a quaternion is given by 
$\dot{q}^{\ast}=q_0-q_1 i -q_2 j-q_3 k$. And the norm of a quaternion is defined as
$\lvert{\dot{q}} \rvert =\sqrt{qq^{\ast}}=\sqrt{q^{\ast}q}=\sqrt{q_0^2+q_1^2+q_2^2+q_3^2}.$
It is worth noting that the commutative law of multiplication on the quaternion system does not hold, i.e., $\dot{q}_1 \dot{q}_2 \neq \dot{q}_2 \dot{q}_1$ in general.

Similarly, Let $\dot{\mathbf Q}=(\dot q_{mn}) \in \mathbb{H}^{M \times N}$ be a quaternion matrix, then $\dot{\mathbf Q} = \mathbf{ Q}_0 +\mathbf{ Q}_1 i +\mathbf{ Q}_2 j+\mathbf{ Q}_3 k$, where $\mathbf{ Q}_t \in  \mathbb{R}^{M \times N} (t=0, \, 1, \, 2, \, 3)$. A quaternion matrix is named a pure quaternion matrix, if $\mathfrak{R}(\dot{\mathbf Q}) =\mathbf{ Q}_0=\mathbf{0}$. The Frobenius norm of a quaternion matrix $\dot{\mathbf {Q}}$ is defined as:
$\left\| \dot{\mathbf {Q}} \right\|_{F}=\sqrt {\sum_{m=1}^{M}\sum_{n=1}^{N}\lvert {\dot {q}}_{mn}\rvert ^{2}} = \sqrt{\text{tr}(\dot{\mathbf {Q}}^{H}\dot{\mathbf {Q})}}.$
And the nuclear norm of $\dot{\mathbf{Q}}$ is defined as: $\|\dot{\mathbf{Q}}\|_{\ast}=\sum_{t}\sigma_{t}$, where $\sigma_{t}$ is the nonzero singular value of $\dot{\mathbf{Q}}$. 

\begin{definition}[The complex representation of a quaternion matrix \cite{ZHANG199721}] \label{Complex representation}  Given a quaternion matrix $\dot{\mathbf{Q}}= \mathbf{ Q}_0 +\mathbf{ Q}_1 i +\mathbf{ Q}_2 j+\mathbf{ Q}_3 k \in \mathbb{H}^{M \times N}$, it can be written as $\dot{\mathbf{Q}}=\mathbf{Q_a}+\mathbf{Q_b}j$, which is called the Cayley-Dickson form of $\dot{\mathbf{Q}}$. And the complex representation of $\dot{\mathbf{Q}}$ is defined as follows 
	\begin{equation}
		{\chi_{\dot{\mathbf Q}}} =
		\begin{bmatrix}
			\mathbf{Q_a}  & \mathbf{Q_b} \\
			-{\mathbf{Q_b}}^\ast & {\mathbf{Q_a}}^\ast
		\end{bmatrix},
	\end{equation}
	where $\mathbf{Q_a}=\mathbf{ Q}_0 +\mathbf{ Q}_1 i$, $\mathbf{Q_b}=\mathbf{ Q}_2 +\mathbf{ Q}_3 i \in \mathbb{C}^{M \times N}$, ${\mathbf{Q_a}}^\ast=\mathbf{ Q}_0 -\mathbf{ Q}_1 i$, ${\mathbf{Q_b}}^\ast=\mathbf{ Q}_2 -\mathbf{ Q}_3 i$, and ${\chi_{\dot{\mathbf Q}}} \in \mathbb{C}^{2M \times 2N}$.
\end{definition}
The complex representation matrices of quaternions and associated matrices share a number of characteristics.  We suggest reading \cite{ZHANG199721} for more details.

\noindent
\begin{theorem}[Singular value decomposition of a quaternion matrix (QSVD) \cite{ZHANG199721}] \label{QSVD}  Given a quaternion matrix $\dot{\mathbf{Q}} \in \mathbb{H}^{M \times N}$ of rank $r$, there exist two quaternion unitary matrices $\dot{\mathbf{U}} \in \mathbb{H}^{M \times M}$ and $\dot{\mathbf{V}} \in \mathbb{H}^{N \times N}$ such that
	\begin{equation}
		\dot{\mathbf{Q}} = \dot{\mathbf{U}}\begin{bmatrix}
			\mathbf{\Sigma}_r  & \mathbf{0} \\
			\mathbf{0} & \mathbf{0}
		\end{bmatrix}\dot{\mathbf{V}}^H=\dot{\mathbf{U}}\mathbf{\Lambda} \dot{\mathbf{V}}^H,\end{equation}
	where $\mathbf{\Sigma}_r$ is a real diagonal matrix with $r$ positive singular values of $\dot{\mathbf{Q}}$ on its diagonal.
\end{theorem}

\noindent
\begin{theorem}[The quaternion Qatar Riyal decomposition (QQR) \cite{Wei2018QuaMatComs}] \label{QQR} Given any quaternion matrix $\dot{\mathbf{A}} \in \mathbb{H}^{M \times N}$ of rank $r$, there exist a unitary matrix $\dot{\mathbf{Q}} \in \mathbb{H}^{M \times M}$ and a weakly upper triangular quaternion matrix $\dot{\mathbf{R}} \in \mathbb{H}^{M \times N}$ such that
	\begin{equation}
		\dot{\mathbf{A}} = \dot{\mathbf{Q}} \dot{\mathbf{R}}.
  \end{equation}
	That is, there exists a permutation matrix $\mathbf{P}\in \mathbb{R}^{N \times N}$ such that  $\dot{\mathbf{R}}\mathbf{P}$ is an upper triangular quaternion matrix.
\end{theorem}

\section{Our method of quaternion completion}
\subsection{Quaternion discrete cosine transform }
In this section, we give a brief introduction to the quaternion discrete cosine transform. 
\begin{definition}[Forward quaternion discrete cosine transform (FQDCT) \cite{feng2008quaternion}] 
Since quaternions do not satisfy commutativity, quaternion discrete cosine transform (QDCT) of a quaternion matrix $\dot{\mathbf{Q}} \in \mathbb{H}^{M \times N}$ has two forms \cite{feng2008quaternion}, i.e., left-handed form and right-handed form: 
	\begin{equation}
        \mathop {FQDCT}\nolimits ^{L} (u,v) = \alpha(u)\alpha (v)\sum \limits _{m =0}^{M - 1} {\sum \limits _{n = 0}^{N - 1} {\dot{q}\, \dot{\mathbf{Q}} \left ({{m,n} }\right)N\left ({{u,v,m,n} }\right)} },  
	\end{equation}
        \begin{equation}
        \mathop {FQDCT}\nolimits ^{R} (u,v) = \alpha(u)\alpha (v)\sum \limits _{m =0}^{M - 1} {\sum \limits _{n = 0}^{N - 1} {\dot{\mathbf{Q}} \left ({{m,n} }\right)N\left ({{u,v,m,n} }\right)\,\dot{q} } },
	\end{equation}
	where $\mathfrak{R}(\dot{q}) =0$ named a quaternionzation factor and satisfies $\dot{q}^2=-1$.
 Similar to discrete cosine transform (DCT) in the real domain, $\alpha(u)$, $\alpha(v)$, and $N\left ({u,v,m,n} \right)$ take the values 
 \begin{equation}
\alpha(u) = \begin{cases}
\sqrt{\frac{1}{M}}\quad {\text {for}}\quad u=0\cr \sqrt{\frac{2}{M}}\quad {\text {for}}\quad u\neq 0
\end{cases}, \quad \alpha(v) = \begin{cases}
\sqrt{\frac{1}{N}}\quad {\text {for}}\quad v=0\cr \sqrt{\frac{2}{N}}\quad {\text {for}}\quad v\neq 0
\end{cases},
 \end{equation}
 \begin{equation}
N\left ({u,v,m,n} \right) = \cos\left[{\frac{\pi (2m+1)u}{2M} }\right]\cos\left[{\frac{\pi (2n+1)v}{2N} }\right].
 \end{equation}
 \end{definition}
 According to the FQDCT, the form of inverse quaternion discrete cosine transform (IQDCT) also has two categories, which are given as follows:
 \begin{equation} \mathop {IQDCT}\nolimits ^{L}(m,n)=\sum \limits _{u=0}^{M-1}{\sum \limits _{v=0}^{N-1} {\alpha (u)\alpha (v)\,{\dot{q}}\, \dot{\mathbf C}\left ({{u,v}}\right)N\left ({{u,v,m,n}}\right)} },
 \end{equation}
 \begin{equation} \mathop {IQDCT}\nolimits ^{R}(m,n)=\sum \limits _{u=0}^{M-1}{\sum \limits _{v=0}^{N-1} {\alpha (u)\alpha (v)\dot{\mathbf C}\left ({{u,v}}\right)N\left ({{u,v,m,n}}\right)}\, {\dot{q}}},
 \end{equation}
where $\dot{\mathbf{C}} \in \mathbb{H}^{M \times N}$.

\begin{theorem}[The relationship between FQDCT and IQDCT \cite{feng2008quaternion}] 
	\begin{equation}
         \begin{aligned}
		\dot{\mathbf{Q}}(m,n) &= \mathop {IQDCT}\nolimits ^{L}\left[\mathop {FQDCT}\nolimits ^{L}(\dot{\mathbf{Q}}(m,n))\right]\\
                                  &=\mathop {IQDCT}\nolimits ^{R}\left[\mathop {FQDCT}\nolimits ^{R}(\dot{\mathbf{Q}}(m,n))\right].
         \end{aligned}
        \end{equation}
\end{theorem}
The computation technique for $\text{FQDCT}^{L}$, which may be summed up as the following steps \cite{feng2008quaternion}, is used to establish our suggested model, thus we only briefly discuss it here.
\begin{enumerate}
    \item Given a quaternion matrix $\dot {\mathbf{Q}}(m,n) \in \mathbb{H}^{M \times N}$, write it in its Cayley-Dickson form, i.e., 
    $\dot {\mathbf{Q}}(m,n)=\mathbf{Q}_{a}(m,n)+\mathbf{Q}_{b}(m,n)j$, where $\mathbf{Q}_{a}(m,n)$ and $\mathbf{Q}_{b}(m,n)\in \mathbb{C}^{M \times N}$.
    \item Calculate the DCT of $\mathbf{Q}_{a}(m,n)$ and $\mathbf{Q}_{b}(m,n)$, and the corresponding results are written as $\mathop {DCT}(\mathbf{Q}_{a}(m,n))$, and $\mathop {DCT}(\mathbf{Q}_{b}(m,n))$, respectively. 
    \item Forming a quaternion matrix from $\mathop {DCT}(\mathbf{Q}_{a}(m,n))$, and $\mathop {DCT}(\mathbf{Q}_{b}(m,n))$, which is given by $\dot{\hat{{\mathbf{Q}}}} (m,n)=\mathop {DCT}(\mathbf{Q}_{a}(m,n))+\mathop {DCT}(\mathbf{Q}_{b}(m,n)) j$.
    \item The final result is given by multiplying $\dot{\hat{{\mathbf{Q}}}}$ by the quaternization factor $\dot{q}$:
    \[\mathop {FQDCT}\nolimits ^{L}\left[\dot {\mathbf{Q}}(m,n)\right]=\dot{q}\, \dot{\hat{{\mathbf{Q}}}}(m,n).\]
\end{enumerate}

\subsection{A new method for computing an approximate QSVD based on quaternion QR decomposition (CQSVD-QQR)}
Assume that $\dot{\mathbf{X}}  \in \mathbb{H}^{M \times N}$ is a given quaternion matrix. In this section, we propose a method for computing the largest $r$ ($r\in (0, N]$) singular values and associated singular vectors of $\dot{\mathbf{X}}$ by using QQR decomposition, named CQSVD-QQR.  In this method, $\dot{\mathbf{X}}$ is decomposed into three quaternion matrices, $\dot{\mathbf{L}} \in \mathbb{H}^{M \times r}$, $\dot{\mathbf{D}}\in \mathbb{H}^{r \times r}$, $\dot{\mathbf{R}}\in \mathbb{H}^{r \times N}$, such that
\begin{equation}
	 \left\| \dot{\mathbf{X}}-\dot{\mathbf{L}} \dot{\mathbf{D}} \dot{\mathbf{R}}  \right\|_{F} ^{2} \leq \varepsilon, \label{qldr}
\end{equation}
where $\varepsilon$ represents a positive tolerance. Therefore, we can formulate a quaternion minimization problem
\begin{equation}
	\mathop{\text{min}}\limits_{\dot{\mathbf{L}},\dot{\mathbf{D}},\dot{\mathbf{R}}}\left\| \dot{\mathbf{X}}-\dot{\mathbf{L}} \dot{\mathbf{D}} \dot{\mathbf{R}}  \right\|_{F} ^{2},   \ \; \text{s.t.}, \dot{\mathbf{L}}^{H} \dot{\mathbf{L}}=\mathbf{I}_r, \ \dot{\mathbf{R}} \dot{\mathbf{R}}^{H}=\mathbf{I}_r, \label{M_QLDR}
\end{equation}
As can be shown from (\ref{M_QLDR}), for three variables $\dot{\mathbf{L}}$, $\dot{\mathbf{D}}$, and $\dot{\mathbf{R}}$, the minimization function is convex to the third variable when two of them are arbitrarily fixed. As a result, we can update these three variables sequentially. The results of the three variables in the $\tau$th iteration are denoted as $\dot{\mathbf{L}}^{\tau}$, $\dot{\mathbf{D}}^{\tau}$, and $\dot{\mathbf{R}}^{\tau}$, respectively. Let $\dot{\mathbf{L}}^{1}=\text{eye}(M,r)$, $\dot{\mathbf{D}}^{1}=\text{eye}(r,r)$, and $\dot{\mathbf{R}}^{1}=\text{eye}(r,N)$. By fixing $\dot{\mathbf{D}}^{\tau}$ and $\dot{\mathbf{R}}^{\tau}$, $\dot{\mathbf{L}}^{\tau+1}$ can be updated by solving the following problem
\begin{equation}
	\dot{\mathbf{L}}^{\tau+1}=\mathop{\text{arg min}}\limits_{\dot{\mathbf{L}}}\left\| \dot{\mathbf{X}}-\dot{\mathbf{L}} \dot{\mathbf{D}}^{\tau} \dot{\mathbf{R}}^{\tau}  \right\|_{F} ^{2},   \label{SOLV_L}
\end{equation}
It can be easily verified that the optimal solution to (\ref{SOLV_L}) is given by
\begin{equation}
	\dot{\mathbf{L}}^{\tau+1}=\dot{\mathbf{X}} (\dot{\mathbf{R}}^{\tau })^{H} (\dot{\mathbf{D}}^{\tau })^{\dagger},   \label{SOLU_L}
\end{equation}
where $(\dot{\mathbf{D}}^{\tau })^{\dagger}$ is the quaternionic pseudoinverse of $\dot{\mathbf{D}}^{\tau }$. According to (\ref{M_QLDR}), the optimal $\dot{\mathbf{L}}$ is a column orthogonal quaternion matrix, so $\dot{\mathbf{L}}^{\tau+1}$ can be chosen as the orthogonal basis of the range space spanned by the columns of $\dot{\mathbf{X}} (\dot{\mathbf{R}}^{\tau })^{H} (\dot{\mathbf{D}}^{\tau })^{\dagger}$, that is,
\begin{equation}
	\dot{\mathbf{L}}^{\tau+1}=\text{orth}(\dot{\mathbf{X}} (\dot{\mathbf{R}}^{\tau })^{H} (\dot{\mathbf{D}}^{\tau })^{\dagger}),   \label{or_L}
\end{equation}
where $\text{orth}(\dot{\mathbf{X}})$ represents an operator that computes the orthogonal basis of the columns of $\dot{\mathbf{X}}$. From (\ref{qldr}), it is clear that $\dot{\mathbf{L}}^{\tau+1}$ is the orthogonal basis for the range of $\dot{\mathbf{X}}$. And $\dot{\mathbf{L}}^{\tau+1}$ can be computed as the orthogonal basis of $\dot{\mathbf{X}}\dot{\mathbf{\Omega}}$, where $\dot{\mathbf{\Omega}} \in \mathbb{H}^{N \times r} $ is a quaternion random matrix \cite{liu2022randomized}. Thus, $\dot{\mathbf{L}}^{\tau+1}$ can be computed as follows:
\begin{equation}
	\dot{\mathbf{L}}^{\tau+1}=\text{orth}(\dot{\mathbf{X}} (\dot{\mathbf{R}}^{\tau })^{H}).   \label{or_L_byran}
\end{equation}
Like the method CSVD-QR in the real domain \cite{liu2018fast}, quaternion QR decomposition is also used to calculate the orthogonal basis of $\dot{\mathbf{X}} (\dot{\mathbf{R}}^{\tau })^{H}$. $\dot{\mathbf{L}}^{\tau+1}$ is given by
\begin{equation}
	\left[\dot{\mathbf{Q}}, \dot{\mathbf{S}}\right]=\text{QQR}(\dot{\mathbf{X}} (\dot{\mathbf{R}}^{\tau })^{H}),   \label{or_L_QR}
\end{equation}
\begin{equation}
	\dot{\mathbf{L}}^{\tau+1} =\dot{\mathbf{Q}} (\dot{\mathbf{q}_1}, \dots, \dot{\mathbf{q}_r}),   \label{or_L_QRr}
\end{equation}
where $\dot{\mathbf{Q}} \in \mathbb{H}^{M \times M}$ and $\dot{\mathbf{S}} \in \mathbb{H}^{M \times r}$, and $\text{QQR}(\dot{\mathbf{X}})$ is an operator that computes the quaternion QR decomposition of $\dot{\mathbf{X}}$. Thus, it can be seen from (\ref{or_L_QR}) that $\dot{\mathbf{X}} (\dot{\mathbf{R}}^{\tau })^{H}=\dot{\mathbf{Q}}\dot{\mathbf{S}}$. Similarly, $\dot{\mathbf{R}}^{\tau+1}$ is updated as follows:
\begin{equation}
	\left[\dot{\mathbf{G}}, \dot{\mathbf{T}}\right]=\text{QQR}(\dot{\mathbf{X}}^{H}\mathbf{L}^{\tau+1}),   \label{or_R_QR}
\end{equation}
\begin{equation}
	\dot{\mathbf{R}}^{\tau+1} =\dot{\mathbf{G}} (\dot{\mathbf{q}_1}, \dots, \dot{\mathbf{q}_r}),   \label{or_R_QRr}
\end{equation}
where $\dot{\mathbf{G}} \in \mathbb{H}^{N \times N}$ and $\dot{\mathbf{T}} \in \mathbb{H}^{N \times r}$, and $\dot{\mathbf{X}}^{H}\mathbf{L}^{\tau+1}=\dot{\mathbf{G}}\dot{\mathbf{T}}$.
In view of (\ref{M_QLDR}), the optimal $\dot{\mathbf{R}}$ is a row orthogonal quaternion matrix, which can be set as follows:
\begin{equation}
	\dot{\mathbf{R}}^{\tau+1} =(\dot{\mathbf{R}}^{\tau+1})^{H}. \label{R_UPDATE}
\end{equation}
With fixing the values of $\dot{\mathbf{L}}^{\tau+1}$ and $\dot{\mathbf{R}}^{\tau+1}$, $\dot{\mathbf{D}}^{\tau+1}$ can be updated by solving the following optimization problem:
\begin{equation}
	\dot{\mathbf{D}}=\mathop{\text{arg min}}\limits_{\dot{\mathbf{D}}}\left\| \dot{\mathbf{X}}-\dot{\mathbf{L}}^{\tau+1}\dot{\mathbf{D}} \dot{\mathbf{R}}^{\tau+1}  \right\|_{F} ^{2},   \label{SOLV_D}
\end{equation}
Since $(\dot{\mathbf{L}}^{\tau+1})^{H}\dot{\mathbf{L}}^{\tau+1}=\mathbf{I}_r, \ \dot{\mathbf{R}}^{\tau+1} (\dot{\mathbf{R}}^{\tau+1})^{H}=\mathbf{I}_r$, it is easy to prove that the optimal solution to (\ref{SOLV_D}) is as follows:
\begin{equation}
	\dot{\mathbf{D}}^{\tau+1}=(\dot{\mathbf{L}}^{\tau+1})^{H}\dot{\mathbf{X}} (\dot{\mathbf{R}}^{\tau+1})^{H}.   \label{EXPLI_SOLU_D}
\end{equation}
Since $\dot{\mathbf{X}}^{H}\mathbf{L}^{\tau+1}=\dot{\mathbf{G}}\dot{\mathbf{T}}$, we can get that
\begin{equation}
	\dot{\mathbf{T}}^{H}=(\dot{\mathbf{L}}^{\tau+1})^{H}\dot{\mathbf{X}} \dot{\mathbf{G}}.   
\end{equation}
According to (\ref{or_R_QRr}) and (\ref{R_UPDATE}) for updating $\dot{\mathbf{R}}^{\tau+1}$, we have
\begin{equation}
	\dot{\mathbf{D}}^{\tau+1} =\dot{\mathbf{T}}^{H} (1 \, \dots \, r,\, 1 \, \dots \, r),   \label{D_UPDATE}
\end{equation}
The three sequences $\{\dot{\mathbf{L}}^{\tau+1}\}$, $\{\dot{\mathbf{R}}^{\tau+1}\}$, and $\{\dot{\mathbf{D}}^{\tau+1}\}$ ($\tau=1,\dots,n,\dots$) obtained by (\ref{or_L_QRr}), (\ref{R_UPDATE}), and (\ref{D_UPDATE}) can converge to quaternion matrices $\dot{\mathbf{L}}$, $\dot{\mathbf{R}}$, and $\dot{\mathbf{D}}$, respectively. And $\dot{\mathbf{D}}$ satisfies
\begin{equation}
	\|\dot{\mathbf{D}}_{ss}\|_{1} =\|\Lambda_{ss}\|_{1},   
\end{equation}
where $s=1, \dots, r$ and $\Lambda_{ss}$ is the $s$th singular value of $\dot{\mathbf{X}}$. In the next subsection, we will use the output of CQSVD-QQR with one iteration to initialize the date of our proposed quaternion matrix completion model.

\begin{table}\footnotesize
	\caption{Computing an approximate QSVD based on quaternion QR decomposition (CQSVD-QQR).}
	\hrule
	\label{tab_algorithm1}
	\begin{algorithmic}[1]
		\Require The quaternion matrix  
                $\dot{\mathbf{X}}\in\mathbb{H}^{M\times N}$.
		\State \textbf{Initialize} $r>0$; $\tau=1$;  $\varepsilon>0$; 
                $\text{It}_{\text{max}}>0$; $\dot{\mathbf{L}}^{1}=eye(M,r)$; $\dot{\mathbf{D}}^{1}=eye(r,r)$; $\dot{\mathbf{R}}^{1}=eye(r,N)$.
		\State \textbf{Repeat}
		\State $\dot{\mathbf{L}}^{\tau+1}$: (\ref{or_L_QR} - 
                     \ref{or_L_QRr}). 
		\State $\dot{\mathbf{R}}^{\tau+1}$: (\ref{or_R_QR} - 
                     \ref{R_UPDATE}). 
		\State   $\dot{\mathbf{D}}^{\tau+1}$: (\ref{or_R_QR},  
                     \ref{D_UPDATE}).
		\State \textbf{Until convergence}: $\|\dot{\mathbf{L}}^{\tau}\dot{\mathbf{D}}^{\tau}\dot{\mathbf{R}}^{\tau}-\dot{\mathbf{X}}\|_{F}^{2}\leq \varepsilon$ \text{or} $\tau>\text{It}_{\text{max}}$.
		\Ensure   $\dot{\mathbf{L}}^{t+1}$, $\dot{\mathbf{D}}^{t+1}$, $\dot{\mathbf{R}}^{t+1}$.
	\end{algorithmic}
	\hrule
\end{table}

\subsection{Proposed quaternion-based methods for color image completion }
This section presents the formulation of the proposed method.
\begin{lemma}[] \label{QLDR} Given quaternion matrices $\dot{\mathbf{L}} \in \mathbb{H}^{M \times r}$, $\dot{\mathbf{R}} \in \mathbb{H}^{r \times N}$, and $\dot{\mathbf{D}} \in \mathbb{H}^{r \times r}$, and assume that $\dot{\mathbf{L}}$ and $\dot{\mathbf{R}}$ have orthogonal columns and orthogonal rows, respectively, i.e., $\dot{\mathbf{L}}^{H}\dot{\mathbf{L}}=\mathbf{I}_r, \ \dot{\mathbf{R}} (\dot{\mathbf{R}})^{H}=\mathbf{I}_r$, then $\|\dot{\mathbf{D}}\|_{\ast}=\|\dot{\mathbf{L}}\dot{\mathbf{D}}\dot{\mathbf{R}}\|_{\ast}.$
\end{lemma}

\begin{proof}
   The QSVD of $\dot{\mathbf{D}}$ can be obtained by using Theorem \ref{QSVD} as follows:
    \begin{equation}
		\dot{\mathbf{D}} = \dot{\mathbf{U}}\begin{bmatrix}
			\mathbf{\Sigma}_s  & \mathbf{0} \\
			\mathbf{0} & \mathbf{0}
		\end{bmatrix}\dot{\mathbf{V}}^H=\dot{\mathbf{U}}_s \mathbf{\Sigma}_s {\dot{\mathbf{V}}_s}^H,
    \end{equation}
    where $\dot{\mathbf{U}}_s \in \mathbb{H}^{r \times s}$ and $\dot{\mathbf{V}}_s \in \mathbb{H}^{r \times s}$ satisfy ${\dot{\mathbf{U}}_s}^H\dot{\mathbf{U}}_s =\mathbf{I}_s$ and ${\dot{\mathbf{V}}_s}^H\dot{\mathbf{V}}_s =\mathbf{I}_s$, respectively, $\mathbf{\Sigma}_s=\text{diag}(\sigma_1, \sigma_2, \dots, \sigma_s)\in \mathbb{H}^{s \times s}$ contains the $s$ positive singular values of $\dot{\mathbf{D}}$, and $s\leq r$.
    Thus, we have 
    \begin{equation}
        \dot{\mathbf{L}}\dot{\mathbf{D}}\dot{\mathbf{R}}=\dot{\mathbf{L}}\dot{\mathbf{U}}_s \mathbf{\Sigma}_s {\dot{\mathbf{V}}_s}^H\dot{\mathbf{R}}=\dot{\mathbf{L}}\dot{\mathbf{U}}_s \mathbf{\Sigma}_s (\dot{\mathbf{R}}^H\dot{\mathbf{V}}_s)^H. \label{QSVD_LDR}
    \end{equation}
  
    It can be checked that 
   $(\dot{\mathbf{L}}\dot{\mathbf{U}}_s)^H\dot{\mathbf{L}}\dot{\mathbf{U}}_s=\mathbf{I}_s$, and  $(\dot{\mathbf{R}}^H\dot{\mathbf{V}}_s)^H\dot{\mathbf{R}}^H\dot{\mathbf{V}}_s=\mathbf{I}_s,$
    Therefore, it can be found that (\ref{QSVD_LDR}) is one of the SVDs of $\dot{\mathbf{L}}\dot{\mathbf{D}}\dot{\mathbf{R}}$. And $\|\dot{\mathbf{D}}\|_{\ast}=\sum_{t=1}^{s}\sigma_{t}=\|\dot{\mathbf{L}}\dot{\mathbf{D}}\dot{\mathbf{R}}\|_{\ast}.$
\end{proof}
Quaternion nuclear norm minimization is a well-established and potent strategy among all models for quaternion-based completion. It can be viewed as the generalization of the nuclear norm minimization of the real matrix. Lemma \ref{QLDR} allows us to determine that $\|\dot{\mathbf{X}}\|_{\ast}=\|\dot{\mathbf{L}}\dot{\mathbf{D}}\dot{\mathbf{R}}\|_{\ast}=\|\dot{\mathbf{D}}\|_{\ast}$. Substituting $\dot{\mathbf{X}}=\dot{\mathbf{L}}\dot{\mathbf{D}}\dot{\mathbf{R}}$ into the original quaternion nuclear norm minimization problem, we can formulate the new model as follows:
\begin{equation}
\mathop{\text{min}}\limits_{\dot{\mathbf{L}},\dot{\mathbf{D}},\dot{\mathbf{R}}}\|\dot{\mathbf{D}}\|_{\ast},   \: \text{s.t.}, 
\begin{cases}
\dot{\mathbf{L}}^{H} \dot{\mathbf{L}}=\mathbf{I}_r, \ \dot{\mathbf{R}} \dot{\mathbf{R}}^{H}=\mathbf{I}_r,\\
\dot{\mathbf{X}}=\dot{\mathbf{L}}\dot{\mathbf{D}}\dot{\mathbf{R}}, \ {P}_{\Omega}(\dot{\mathbf{L}}\dot{\mathbf{D}}\dot{\mathbf{R}}-\dot{\mathbf{M}})=\mathbf{0}.
\end{cases}
\label{M_DNN}
\end{equation}
As we previously stated, the color image completion problem cannot be solved by just taking into account the image's low-rank prior information. To address it, we, therefore, integrate the image's low-rank prior with its sparseness in a transform domain. By introducing sparsity, we can obtain the following model:  
\begin{equation}
\mathop{\text{min}}\limits_{\dot{\mathbf{L}},\dot{\mathbf{D}},\dot{\mathbf{R}}}\|\dot{\mathbf{D}}\|_{\ast}+\lambda \|\dot{\mathbf{W}}\|_{1},   \: \text{s.t.}, 
\begin{cases}
\dot{\mathbf{L}}^{H} \dot{\mathbf{L}}=\mathbf{I}_r, \ \dot{\mathbf{R}} \dot{\mathbf{R}}^{H}=\mathbf{I}_r,\\
\dot{\mathbf{X}}=\dot{\mathbf{L}}\dot{\mathbf{D}}\dot{\mathbf{R}}, \ {P}_{\Omega}(\dot{\mathbf{X}}-\dot{\mathbf{L}}\dot{\mathbf{D}}\dot{\mathbf{R}})=\mathbf{0}, \ \mathcal{T}(\dot{\mathbf{X}})=\dot{\mathbf{W}},
\end{cases}
\label{M_DNN_SP}
\end{equation}
where $\mathcal{T}(\cdot)$ represents the transform operator, $\dot{\mathbf{W}}=\mathcal{T}(\dot{\mathbf{X}})$ is the transformed quaternion matrix, and the parameter $\lambda > 0$. Specifically, we use the $\text{FQDCT}^{L}$ to establish the model, so the operator $\mathcal{T}(\dot{\mathbf{X}})$ here refers to the $\text{FQDCT}^{L}$ of $\dot{\mathbf{X}}$. Since this formulation employs a method for approximating QSVD based on quaternion QR decomposition (CQSVD-QR), quaternion nuclear norm, and a sparse regularizer, we name the proposed method as QQR-QNN-SR.

\subsection{Optimization framework}
The optimization problem (\ref{M_DNN_SP}) is addressed in this section using the alternating direction method of multipliers (ADMM) framework. This framework in the quaternion system is similar to the ADMM framework in the complex number field \cite{li2015alternating}. The augmented Lagrangian function of the problem (\ref{M_DNN_SP}) is given by
\begin{equation}
    \begin{split}
        & \mathcal{L}(\dot{\mathbf{X}}, \dot{\mathbf{L}}, \dot{\mathbf{D}}, \dot{\mathbf{R}}, \dot{\mathbf{W}}, \dot{\mathbf{Y}}, \dot{\mathbf{Z}})\\
        &=\|\dot{\mathbf{D}}\|_{\ast}+\lambda \|\dot{\mathbf{W}}\|_{1}\\
        &+\mathfrak{R}(\langle \dot{\mathbf{Y}}, \, \dot{\mathbf{X}}-\dot{\mathbf{L}} \dot{\mathbf{D}}\dot{\mathbf{R}}\rangle)+\frac{\mu}{2}\|\dot{\mathbf{X}}-\dot{\mathbf{L}} \dot{\mathbf{D}}\dot{\mathbf{R}}\|_{F}^{2}\\
        & +\mathfrak{R}(\langle \dot{\mathbf{Z}}, \, \dot{\mathbf{W}}-\mathcal{T}(\dot{\mathbf{X}})\rangle)+\frac{\mu}{2}\|\dot{\mathbf{W}}-\mathcal{T}(\dot{\mathbf{X}})\|_{F}^{2},
    \end{split}
\end{equation}
where the penalty parameter $\mu>0$, $\dot{\mathbf{Y}}$ and $\dot{\mathbf{Z}}$ are Lagrange multiplier. The optimization problem (\ref{M_DNN_SP}) can have each parameter updated alternatively by fixing other variables, in accordance with the ADMM framework.  
Particularly, all variables can be updated iteratively as follows:
\begin{equation}
   \begin{cases}
   \dot{\mathbf{L}}^{t+1}=\mathop{\text{arg min}}\limits_{\dot{\mathbf{L}}}\mathcal{L}(\dot{\mathbf{X}}^{t}, \dot{\mathbf{L}}, \dot{\mathbf{D}}^{t}, \dot{\mathbf{R}}^{t}, \dot{\mathbf{W}}^{t}, \dot{\mathbf{Y}}^{t}, \dot{\mathbf{Z}}^{t}),\\
   \dot{\mathbf{R}}^{t+1}=\mathop{\text{arg min}}\limits_{\dot{\mathbf{R}}}\mathcal{L}(\dot{\mathbf{X}}^{t}, \dot{\mathbf{L}}^{t+1}, \dot{\mathbf{D}}^{t}, \dot{\mathbf{R}}, \dot{\mathbf{W}}^{t}, \dot{\mathbf{Y}}^{t}, \dot{\mathbf{Z}}^{t}),\\
   \dot{\mathbf{D}}^{t+1}=\mathop{\text{arg min}}\limits_{\dot{\mathbf{D}}}\mathcal{L}(\dot{\mathbf{X}}^{t}, \dot{\mathbf{L}}^{t+1}, \dot{\mathbf{D}}, \dot{\mathbf{R}}^{t+1}, \dot{\mathbf{W}}^{t}, \dot{\mathbf{Y}}^{t}, \dot{\mathbf{Z}}^{t}),\\
    \dot{\mathbf{X}}^{t+1}=\mathop{\text{arg min}}\limits_{\dot{\mathbf{X}}}\mathcal{L}(\dot{\mathbf{X}}, \dot{\mathbf{L}}^{t+1}, \dot{\mathbf{D}}^{t+1}, \dot{\mathbf{R}}^{t+1}, \dot{\mathbf{W}}^{t}, \dot{\mathbf{Y}}^{t}, \dot{\mathbf{Z}}^{t}),\\
     \dot{\mathbf{W}}^{t+1}=\mathop{\text{arg min}}\limits_{\dot{\mathbf{W}}}\mathcal{L}(\dot{\mathbf{X}}^{t+1}, \dot{\mathbf{L}}^{t+1}, \dot{\mathbf{D}}^{t+1}, \dot{\mathbf{R}}^{t+1}, \dot{\mathbf{W}}, \dot{\mathbf{Y}}^{t}, \dot{\mathbf{Z}}^{t}),\\
      \dot{\mathbf{Y}}^{t+1}=\dot{\mathbf{Y}}^{t}+\mu^{t}( \dot{\mathbf{X}}^{t+1}-\dot{\mathbf{L}}^{t+1} \dot{\mathbf{D}}^{t+1}\dot{\mathbf{R}}^{t+1}),\\
       \dot{\mathbf{Z}}^{t+1}=\dot{\mathbf{Z}}^{t}+\mu^{t}( \dot{\mathbf{W}}^{t+1}-\mathcal{T}(\dot{\mathbf{X}}^{t+1})).\\
    \end{cases}
    \label{sub_optimi}
\end{equation}
As shown in (\ref{sub_optimi}), we first fix other variables to solve the corresponding optimization problem to update $\dot{\mathbf{L}}$, $\dot{\mathbf{D}}$, and $\dot{\mathbf{R}}$.\\
\indent
\textbf{Updating $\dot{\mathbf{L}}$, $\dot{\mathbf{D}}$, and $\dot{\mathbf{R}}$}:
In the $t+1$-th iteration, $\dot{\mathbf{L}}^{t+1}$ and $\dot{\mathbf{R}}^{t+1}$ are updated by solving the following minimization problem:
\begin{equation}
\mathop{\text{min}}\limits_{\dot{\mathbf{L}},\dot{\mathbf{R}}}\left\| (\dot{\mathbf{X}}^{t}+\frac{\dot{\mathbf{Y}}^{t}}{\mu^{t}})-\dot{\mathbf{L}} \dot{\mathbf{D}}^{t} \dot{\mathbf{R}}  \right\|_{F} ^{2}.
\end{equation}
According to the analysis of the quaternion optimization problem (\ref{M_QLDR}), $\dot{\mathbf{L}}$, and $\dot{\mathbf{R}}$ can be obtained by using CQSVD-QQR as follows:
\begin{equation}
\begin{cases}
	\left[\dot{\hat{\mathbf{L}}}^{t+1}, {\sim} \right]= \text{QQR}((\dot{\mathbf{X}}^{t}+\frac{\dot{\mathbf{Y}}^{t}}{\mu^{t}}) (\dot{\mathbf{R}}^{t })^{H}),\\ 
	\dot{\mathbf{L}}^{t+1} =\dot{\hat{\mathbf{L}}}^{t+1} (\dot{\mathbf{q}_1}, \dots, \dot{\mathbf{q}_r}),   \label{upd_L}
 \end{cases}
\end{equation}
and 
\begin{equation}
\begin{cases}
	\left[\dot{\hat{\mathbf{R}}}^{t+1},\dot{\hat{\mathbf{D}}}^{t} \right]= \text{QQR}((\dot{\mathbf{X}}^{t}+\frac{\dot{\mathbf{Y}}^{t}}{\mu^{t}})^{H} \dot{\mathbf{L}}^{t+1 }),\\ 
	\dot{\mathbf{R}}^{t+1} =(\dot{\hat{\mathbf{R}}}^{t+1} (\dot{\mathbf{q}_1}, \dots, \dot{\mathbf{q}_r}))^{H},  \label{upd_R}
 \end{cases}
\end{equation}
where QQR is the quaternion QR decomposition. If the two variables $\dot{\mathbf{L}}$ and $\dot{\mathbf{R}}$ are started as $\dot{\mathbf{L}}^{t}$ and $\dot{\mathbf{R}}^{t}$, the CQSVD-QQR method will converge within a limited number of iterations because the quaternion matrices $\dot{\mathbf{L}}$ and $\dot{\mathbf{R}}$ will not change significantly in two consecutive iterations. As a result, in the proposed method for quaternion matrix completion, the output of CQSVD-QQR with one iteration is used. In one iteration of CQSVD-QQR, this method decomposes the quaternion matrix $\dot{\mathbf{X}}^{t}+\frac{\dot{\mathbf{Y}}^{t}}{\mu^{t}}$ as follows:
\begin{equation}
\dot{\mathbf{X}}^{t}+\frac{\dot{\mathbf{Y}}^{t}}{\mu^{t}}=\dot{\mathbf{L}}^{t+1} \dot{\widetilde{\mathbf{D}}}^{t}\dot{\mathbf{R}}^{t+1}, \label{widetilde_D}
\end{equation}
where $\dot{\widetilde{\mathbf{D}}}^{t}=\dot{\hat{\mathbf{D}}}^{t}(1 \, \dots \, r,\, 1 \, \dots \, r) \in \mathbb{H}^{r \times r}$.
According to (\ref{M_QLDR}) and (\ref{widetilde_D}), we also have 
\begin{equation}
\dot{\widetilde{\mathbf{D}}}^{t}=(\dot{\mathbf{L}}^{t+1})^{H} (\dot{\mathbf{X}}^{t}+\frac{\dot{\mathbf{Y}}^{t}}{\mu^{t}})(\dot{\mathbf{R}}^{t+1})^{H}. \label{widetilde_D1}
\end{equation}
After updating $\dot{\mathbf{L}}$ and $\dot{\mathbf{R}}$, we can update $\dot{\mathbf{D}}$ by solving the following problem:
\begin{equation}
\dot{\mathbf{D}}^{t+1}=
\mathop{\text{arg min}}\limits_{\dot{\mathbf{D}}}\|\dot{\mathbf{D}}\|_{\ast}++\frac{\mu^{t}}{2}\|\dot{\mathbf{D}}-(\dot{\mathbf{L}}^{t+1})^{H} (\dot{\mathbf{X}}^{t}+\frac{\dot{\mathbf{Y}}^{t}}{\mu^{t}})(\dot{\mathbf{R}}^{t+1})^{H}\|_{F}^{2},  
\label{UPD_D}
\end{equation}
Substituting (\ref{widetilde_D1}) into (\ref{UPD_D}), then the above problem can be reformulated as follows:
\begin{equation}
\dot{\mathbf{D}}^{t+1}=
\mathop{\text{arg min}}\limits_{\dot{\mathbf{D}}}\|\dot{\mathbf{D}}\|_{\ast}++\frac{\mu^{t}}{2}\|\dot{\mathbf{D}}-\dot{\widetilde{\mathbf{D}}}^{t}\|_{F}^{2}.
\label{UPD_D2}
\end{equation}
The closed solution to the problem (\ref{UPD_D2}) is \cite{chen2019low}
\begin{equation}
\dot{\mathbf{D}}^{t+1}=
\mathbf{\mathcal{D}}_{\frac{1}{\mu^{t}}}(\dot{\widetilde{\mathbf{D}}}^{t})=\dot{\mathbf{U}}\mathbf{\mathcal{S}}_{\frac{1}{\mu^{t}}}(\Sigma)\dot{\mathbf{V}}^{H},
\label{UPD_D3}
\end{equation}
 where $\dot{\widetilde{\mathbf{D}}}^{t}=\dot{\mathbf{U}}\Sigma\dot{\mathbf{V}}^{H}$ is the QSVD of $\dot{\widetilde{\mathbf{D}}}^{t}$, and $\mathbf{\mathcal{S}}_{\frac{1}{\mu^{t}}}(\Sigma)=\text{diag}(\text{max}\{\sigma_{p}(\dot{\widetilde{\mathbf{D}}}^{t})-\frac{1}{\mu^{t}},0\})$ is the soft thresholding operator, and $\sigma_{p}(\dot{\widetilde{\mathbf{D}}}^{t})$ is the p-th positive singular value of $\dot{\widetilde{\mathbf{D}}}^{t}$.\\
\indent
\textbf{Updating $\dot{\mathbf{X}}$, $\dot{\mathbf{W}}$, $\dot{\mathbf{Y}}$, $\dot{\mathbf{Z}}$, and $\mu$}:
In the $t+1$-th iteration, fixing the variables $\dot{\mathbf{L}}^{t+1}$, $\dot{\mathbf{D}}^{t+1}$, $\dot{\mathbf{R}}^{t+1}$, $\dot{\mathbf{X}}^{t+1}$ is the optimal solution of the problem as follows:
\begin{equation}
  \begin{split}
\dot{\mathbf{X}}^{t+1}&=\mathop{\text{arg min}}\limits_{\dot{\mathbf{X}}}\mathfrak{R}(\langle \dot{\mathbf{Y}}^{t}, \, \dot{\mathbf{X}}-\dot{\mathbf{L}}^{t+1} \dot{\mathbf{D}}^{t+1} \dot{\mathbf{R}}^{t+1} \rangle)+\frac{\mu^{t}}{2}\|\dot{\mathbf{X}}-\dot{\mathbf{L}}^{t+1}  \dot{\mathbf{D}}^{t+1} \dot{\mathbf{R}}^{t+1} \|_{F}^{2}\\
& +\mathfrak{R}(\langle \dot{\mathbf{Z}}^{t} , \, \dot{\mathbf{W}}^{t} -\mathcal{T}(\dot{\mathbf{X}})\rangle)+\frac{\mu^{t}}{2}\|\dot{\mathbf{W}}^{t} -\mathcal{T}(\dot{\mathbf{X}})\|_{F}^{2}, \\
&=\mathop{\text{arg min}}\limits_{\dot{\mathbf{X}}}\frac{\mu^{t}}{2}\|\dot{\mathbf{X}}+\frac{\dot{\mathbf{Y}}^{t}}{\mu^t}-\dot{\mathbf{L}}^{t+1} \dot{\mathbf{D}}^{t+1}\dot{\mathbf{R}}^{t+1}\|_{F}^{2}+\frac{\mu^{t}}{2}\|\dot{\mathbf{W}}^{t}+\frac{\dot{\mathbf{Z}}^{t}}{\mu^t}-\mathcal{T}(\dot{\mathbf{X}})\|_{F}^{2}.
\label{UPD_X}
  \end{split}
\end{equation}
Because the existence of the term $\mathcal{T}(\dot{\mathbf{X}})$ in problem (\ref{UPD_X}), the variable $\dot{\mathbf{X}}$ cannot be directly separated from the other variables. Fortunately, the Parseval theorem in the quaternion system isolates the variable $\dot{\mathbf{X}}$ from the operator $\mathcal{T}(\cdot)$, allowing us to reformulate the issue.
Similar to the real domain, the Parseval theorem in the quaternion system \cite{bahri2008uncertainty} asserts that the total energy of the signal remains unchanged after a unitary transformation, such as quaternion discrete Fourier transform (QDFT), and quaternion discrete cosine transform (QDCT). Therefore, by introducing the corresponding inverse transform to the last term of (\ref{UPD_X}), we have
\begin{equation}
    \|\dot{\mathbf{W}}^{t}+\frac{\dot{\mathbf{Z}}^{t}}{\mu^t}-\mathcal{T}(\dot{\mathbf{X}})\|_{F}^{2}= \|\mathcal{I}_{\mathcal{T}}(\dot{\mathbf{W}}^{t}+\frac{\dot{\mathbf{Z}}^{t}}{\mu^t})-\dot{\mathbf{X}}\|_{F}^{2},
    \label{intro_inver}
\end{equation}
where $\mathcal{I}_{\mathcal{T}}(\cdot)$ is the inverse transform of $\mathcal{T}(\cdot)$.
After substituting (\ref{intro_inver}) into (\ref{UPD_X}) and ignoring terms unrelated to the variable $\dot{\mathbf{X}}$, the optimization problem used to update $\dot{\mathbf{X}}$ can be written as follows
\begin{equation}
\dot{\mathbf{X}}^{t+1}=\mathop{\text{arg min}}\limits_{\dot{\mathbf{X}}}
\|\frac{1}{2}(\dot{\mathbf{L}}^{t+1} \dot{\mathbf{D}}^{t+1}\dot{\mathbf{R}}^{t+1}-\frac{\dot{\mathbf{Y}}^{t}}{\mu^t}+\mathcal{I}_{\mathcal{T}}(\dot{\mathbf{W}}^{t}+\frac{\dot{\mathbf{Z}}^{t}}{\mu^t}))-\dot{\mathbf{X}}\|_{F}^{2}.
\label{UPD_X1}
\end{equation}
Thus, we can obtain that
\begin{equation}
\dot{\mathbf{X}}^{t+1}=
\frac{1}{2}(\dot{\mathbf{L}}^{t+1} \dot{\mathbf{D}}^{t+1}\dot{\mathbf{R}}^{t+1}-\frac{\dot{\mathbf{Y}}^{t}}{\mu^t}+\mathcal{I}_{\mathcal{T}}(\dot{\mathbf{W}}^{t}+\frac{\dot{\mathbf{Z}}^{t}}{\mu^t})).
\label{sou_X}
\end{equation}
Combined with the constraint that ${P}_{\Omega}(\dot{\mathbf{L}}\dot{\mathbf{D}}\dot{\mathbf{R}})={P}_{\Omega}(\dot{\mathbf{M}})$, the value of $\dot{\mathbf{X}}$ should remain unchanged on the observed entries set $\Omega$, i.e.,
\begin{equation}
\dot{\mathbf{X}}^{t+1}=
{P}_{\Omega^{c}}(\dot{\mathbf{X}}^{t+1})+{P}_{\Omega}(\dot{\mathbf{M}}),
\label{sou_X1}
\end{equation}
where $\Omega^{c}$ is the set of locations corresponding to missing entries.\\
\indent
Next, in the $t+1$-th iteration, the variable $\dot{\mathbf{W}}^{t+1}$ is updated by solving the following problem:
\begin{equation}
\begin{aligned}
\dot{\mathbf{W}}^{t+1}&=\mathop{\text{arg min}}\limits_{\dot{\mathbf{W}}}
\lambda \|\dot{\mathbf{W}}\|_{1}+\mathfrak{R}(\langle \dot{\mathbf{Z}}^{t}, \, \dot{\mathbf{W}}-\mathcal{T}(\dot{\mathbf{X}}^{t+1})\rangle)+\frac{\mu^{t}}{2}\|\dot{\mathbf{W}}-\mathcal{T}(\dot{\mathbf{X}}^{t+1})\|_{F}^{2}\\
&=\mathop{\text{arg min}}\limits_{\dot{\mathbf{W}}}\lambda \|\dot{\mathbf{W}}\|_{1}
+\frac{\mu^{t}}{2}\|\dot{\mathbf{W}}+\frac{\dot{\mathbf{Z}}^{t}}{\mu^t}-\mathcal{T}(\dot{\mathbf{X}}^{t+1})\|_{F}^{2}. \label{UPD_W}
\end{aligned}
\end{equation}
The optimal solution to (\ref{UPD_W}) is 
\begin{equation}
\dot{\mathbf{W}}^{t+1}=\mathcal{S}_{\frac{4\lambda}{\mu^{t}}}(\mathcal{T}(\dot{\mathbf{X}}^{t+1})-\frac{\dot{\mathbf{Z}}^{t}}{\mu^t}),
\label{sou_W}
\end{equation}
where $\mathcal{S}_{p}(\dot{\mathbf{x}})=\frac{\dot{\mathbf{x}}}{|\dot{\mathbf{x}}|}\text{max}\{|\dot{\mathbf{x}}|-p,0\}$ is the element-wise soft thresholding operator \cite{yang2022quaternion}.\\
\indent
Finally, the variables $\dot{\mathbf{Y}}^{t+1}$, $\dot{\mathbf{Z}}^{t+1}$, and the penalty parameter $\mu^{t+1}$ are updated as follows:
\begin{equation}
\dot{\mathbf{Y}}^{t+1}=\dot{\mathbf{Y}}^{t}+\mu^{t}(\dot{\mathbf{X}}^{t+1}-\dot{\mathbf{L}}^{t+1}\dot{\mathbf{D}}^{t+1}\dot{\mathbf{R}}^{t+1}),
\label{UPD_Y}
\end{equation}
\begin{equation}
\dot{\mathbf{Z}}^{t+1}=\dot{\mathbf{Z}}^{t}+\mu^{t}(\dot{\mathbf{W}}^{t+1}-\mathcal{T}(\dot{\mathbf{X}}^{t+1})),
\label{UPD_Z}
\end{equation}
and
\begin{equation}
\mu^{t+1}=\gamma\mu^{t}.
\label{UPD_mu}
\end{equation}

\begin{table}\footnotesize
	\caption{The QQR-QNN-SR-based quaternion matrix completion method.}
	\hrule
	\label{tab_algorithm2}
	\begin{algorithmic}[1]
		\Require The observed quaternion matrix data 
                $\dot{\mathbf{M}}\in\mathbb{H}^{M\times N}$ with $\mathcal{P}_{\Omega^{c}}(\dot{\mathbf{M}})=\mathbf{0}$; $\lambda$; $\mu_{\max}$; $\gamma$ and $r>0$.
		\State \textbf{Initialize} $t=0$; $\mu^{0}$;  $\varepsilon>0$; 
                $\text{It}_{\text{max}}>0$; $\dot{\mathbf{L}}^{0}=eye(M,r)$; $\dot{\mathbf{R}}^{0}=eye(r,N)$; $\dot{\mathbf{D}}^{0}=eye(r,r)$; $\dot{\mathbf{X}}^{0}=\dot{\mathbf{M}}$; $\dot{\mathbf{W}}^{0}=\mathbf{0}$.
		\State \textbf{Repeat}
		\State \textbf{Step 1.} $\dot{\mathbf{L}}^{t+1}$, $\dot{\mathbf{R}}^{t+1}$: (\ref{upd_L}) and (\ref{upd_R}), respectively.
		\State	\textbf{Step 2.}	
		      $\dot{\mathbf{D}}^{t+1}=
                \mathbf{\mathcal{D}}_{\frac{1}{\mu^{t}}}(\dot{\widetilde{\mathbf{D}}}^{t})=\dot{\mathbf{U}}\mathbf{\mathcal{S}}_{\frac{1}{\mu^{t}}}(\Sigma)\dot{\mathbf{V}}^{H}.$
		\State \textbf{Step 3.} $\dot{\mathbf{X}}^{\tau+1}=\mathcal{P}_{\Omega^{c}}\big(\frac{1} 
                 {2}(\dot{\mathbf{L}}^{t+1} \dot{\mathbf{D}}^{t+1}\dot{\mathbf{R}}^{t+1}-\frac{\dot{\mathbf{Y}}^{t}}{\mu^t}+\mathcal{I}_{\mathcal{T}}(\dot{\mathbf{W}}^{t}+\frac{\dot{\mathbf{Z}}^{t}}{\mu^t}))\big)+\mathcal{P}_{\Omega}(\dot{\mathbf{M}})$.
		\State \textbf{Step 4.} $\dot{\mathbf{W}}^{t+1}=\mathcal{S}_{\frac{4\lambda}{\mu^{t}}} 
                  (\mathcal{T}(\dot{\mathbf{X}}^{t+1})-\frac{\dot{\mathbf{Z}}^{t}}{\mu^t})$.
		\State $\dot{\mathbf{Y}}^{t+1}=\dot{\mathbf{Y}}^{t}+\mu^{t} 
                  (\dot{\mathbf{X}}^{t+1}-\dot{\mathbf{L}}^{t+1}\dot{\mathbf{D}}^{t+1}\dot{\mathbf{R}}^{t+1}).$
            \State $\dot{\mathbf{Z}}^{t+1}=\dot{\mathbf{Z}}^{t}+\mu^{t} 
                  (\dot{\mathbf{W}}^{t+1}-\mathcal{T}(\dot{\mathbf{X}}^{t+1}))$.
		\State $\mu^{t+1}={\rm{min}}(\gamma\mu^{t}, \mu_{max})$.
		\State \textbf{Until convergence}
		\Ensure   $\dot{\mathbf{L}}^{t+1}$, $\dot{\mathbf{D}}^{t+1}$, $\dot{\mathbf{R}}^{t+1}$, $\dot{\mathbf{X}}^{t+1}$, and $\dot{\mathbf{W}}^{t+1}$.
	\end{algorithmic}
	\hrule
\end{table}
The whole procedure of the proposed method is summarized in Table \ref{tab_algorithm2}.

\subsection{Complexity analysis}
In this section, we analyze the computational complexity of the proposed QQR-QNN-SR method. As shown in Table \ref{tab_algorithm2}, the main computational cost of each iteration of QQR-QNN-SR involves the update cost of two types of variables, one is the cost of updating variables $\dot{\mathbf{L}}$, $\dot{\mathbf{D}}$, and $\dot{\mathbf{R}}$, and the other is the cost of updating variables $\dot{\mathbf{X}}$ and $\dot{\mathbf{W}}$.
The complexity when updating variables of the first type is dominated by the quaternion QR decomposition operation of two quaternion matrices of size $M\times r$ and $N\times r$ whose complexity is about $\mathcal{O}(r^2(M+N)-r^{3})$, and the calculation of the QSVD of an $r\times r$ quaternion matrix whose complexity is about $\mathcal{O}(r^{3})$, where $r<\text{min}\{M, N\}$.
However, methods directly based on the quaternion nuclear norm, such as LRQA-G \cite{chen2019low}, need to calculate the QSVD of the quaternion matrix with size $M\times N$, and the complexity is about $\mathcal{O}(\text{min}(MN^2, M^2N))$. It is clear that the computational cost of QR decomposition of two quaternion matrices of size $M\times r$ and $N\times r$ is smaller than that of QSVD of the quaternion matrix with size $M\times N$. In the second part of QQR-QNN-SR, the computational complexity is dominated by the transform operator $\mathcal{T}$. And the complexity is about $\mathcal{O}(M^2N^2+MN)$. Therefore, the total computational complexity of each iteration of the proposed QQR-QNN-SR method is dominated by the second part with the complexity being about $\mathcal{O}(M^2N^2+MN)$.

\section{Experimental results and discussion}
In this section, we first give a test of the convergence of CQSVD-QQR. The effectiveness of the proposed QQR-QNN-SR method is then demonstrated by conducting numerical experiments on color images.\\
\indent 
We perform all the experiments on a MATLAB 2019b platform equipped with an i7-9700 CPU and 16 GB of RAM.
\subsection{Convergence of CQSVD-QQR}
  First, a synthetic quaternion matrix $\dot{\mathbf{X}}$ is generated as follows:
\begin{equation}
    \dot{\mathbf{X}}= \dot{\mathbf{M}}_{1}^{m \times r_1}\dot{\mathbf{M}}_{2}^{r_1 \times n},
\end{equation}
where $r_1 \in [1, n]$ is the rank of $\dot{\mathbf{X}}$, and $ \dot{\mathbf{M}}_{1}^{m \times r_1}$ and  $\dot{\mathbf{M}}_{2}^{r_1 \times n}$ are two quaternion random matrices.
Then, let $m=n=300, \ r_1=250,\ r=120$, with the largest $r$ singular values. Firstly, we use the proposed CQSVD-QQR to obtain the quaternion Tri-Factorization of $\dot{\mathbf{X}}$. And we compare the results using the root-mean-square error (RMSE), which is defined as
\begin{equation}
    \text{RMSE}= \sqrt{\frac{\|\dot{\mathbf{X}}-\dot{\mathbf{Y}}\|_{F}^{2}}{mn}},
\end{equation}
where $\dot{\mathbf{X}}, \ \dot{\mathbf{Y}}\in \mathbb{H}^{m \times n}$, $\dot{\mathbf{X}}$ is the original data, and $\dot{\mathbf{Y}}$ is the reconstructed result.
After that, RMSE for the model CQSVD-QQR with respect to the variable $\dot{\mathbf{X}}_{\text{LDR}}$ can be obtained.
Similarly, we can calculate RMSE for the model QSVD with respect to $\dot{\mathbf{X_{\text{QSVD}}}}$, where $\dot{\mathbf{X}}_{\text{LDR}}$ and $\dot{\mathbf{X_{\text{QSVD}}}}$ denote the reconstructed results of CQSVD-QQR and QSVD, respectively. \\
\begin{figure}
    \centering
    \includegraphics[width=8cm,height=5cm]{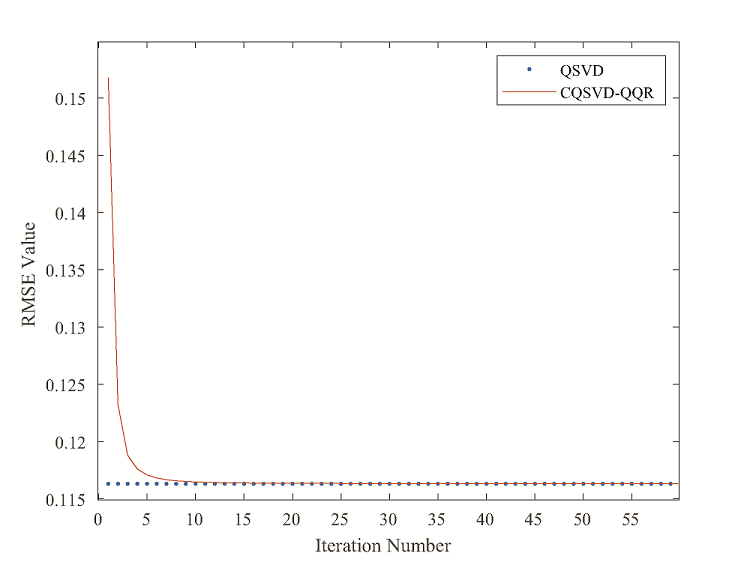}
    \caption{Comparison of RMSE values of CQSVD-QQR and QSVD.}
    \label{fig:1}
\end{figure}
\begin{figure}
    \centering
    \includegraphics[width=15.8cm,height=3.6cm]{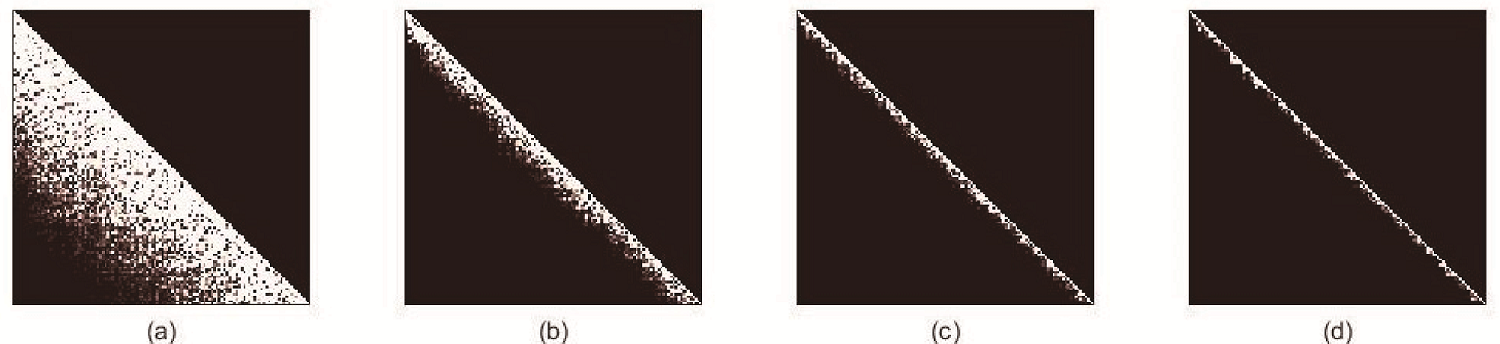}
    \caption{Convergence procedure of $\dot{\mathbf{D}}_{\tau}$ in CQSVD-QQR. (a) $k=5$. (b) $k=20$. (c) $k=40$. (d) $k=60$.}
    \label{fig:2}
\end{figure}
\indent
As shown in Fig. \ref{fig:1} that CQSVD-QQR will converge to a certain value as the number of iterations increases. Moreover, the accuracy of CQSVD-QQR is close to that of the QSVD method after a few iterative steps.\\
\indent
Additionally, the variable $\dot{\mathbf{D}}_{\tau}$ in CQSVD-QQR will converge to a diagonal matrix. Let $\dot{\mathbf{T}}_{\tau} \in \mathbb{R}^{r \times r}$ be a square matrix with each element $\dot{\mathbf{T}}_{\tau}(u,v)=|\dot{\mathbf{D}}_{\tau}(u,v)|$, for $i, j=1, \dots, r$.  
In order to demonstrate the convergence of $\dot{\mathbf{D}}_{\tau}$, we compute $\dot{\mathbf{D}}_{5}$, $\dot{\mathbf{D}}_{20}$, $\dot{\mathbf{D}}_{40}$, and $\dot{\mathbf{D}}_{60}$. And the convergence can be shown by plotting $10\dot{\mathbf{T}}_{\tau}$ ($\tau=5, \ 20,\ 40,\ 60$), as shown in Fig. \ref{fig:2}.

\subsection{Experimental results of QQR-QNN-SR}
This section tests the proposed QQR-QNN-SR utilizing several natural color images and color medical images. We perform numerical experiments to complete (quaternion) matrices with random missing entries and matrices with random block missing entries. The ratio of the missing entries (MR) ranges from a minimum of $50\%$ to a maximum of $90\%$ to better demonstrate the efficacy of our technique.\\
\indent
\textbf{Compared method:} We compare the proposed QQR-QNN-SR with several state-of-the-art completion methods, including IRLNM-QR \cite{liu2018fast}, WNNM \cite{gu2017weighted}, MC-NC \cite{nie2018matrix}, TNNR \cite{hu2012fast}, TNN-SR \cite{dong2018low}, LRQMC \cite{miao2021color}, LRQA-G \cite{chen2019low}, QLNF \cite{yang2022quaternion}, TQLNA \cite{yang2022quaternion}. The first five methods are real matrix-based completion algorithms, and the last four are low-rank completion algorithms based on quaternion matrices.

\textbf{Quantitative assessment and natural color image for testing:} To evaluate the performance of the proposed QQR-QNN-SR, we use two commonly used measures, including peak signal-to-noise ratio (PSNR), and the structure similarity (SSIM). Higher PSNR and SSIM indicate better recovery performance. Among all the methods compared, the best numerical results are bolded. \\
\indent
Eight widely used natural color images with size $256\times 256$ in Fig. \ref{fig:3} are selected as the test images in our experiments. Like the processing in \cite{dong2018low}, among all the methods compared, those based on real matrix completion first process the three color channels separately, and then combine the results of the three channels to obtain the final result. 
\begin{figure}[htbp]
     \centering
     \begin{subfigure}[]{0.15\textwidth}
         \centering
         \includegraphics[width=2.5cm,height=2.5cm]{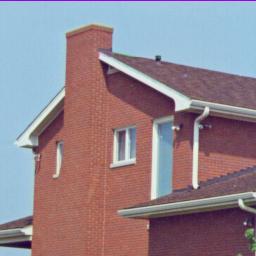}
         \caption*{Image (1)}    
    \end{subfigure}
     \quad
     \begin{subfigure}[]{0.15\textwidth}
         \centering
         \includegraphics[width=2.5cm,height=2.5cm]{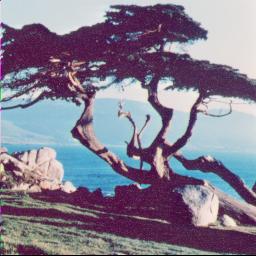}
         \caption*{Image (2)} 
     \end{subfigure}
     \quad
     \begin{subfigure}[]{0.15\textwidth}
         \centering
         \includegraphics[width=2.5cm,height=2.5cm]{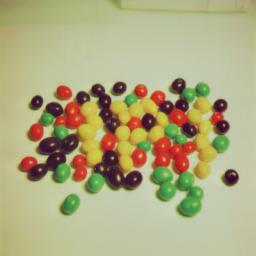}
         \caption*{Image (3)}
     \end{subfigure}
      \quad
     \begin{subfigure}[]{0.15\textwidth}
         \centering
         \includegraphics[width=2.5cm,height=2.5cm]{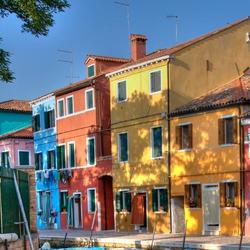}
         \caption*{Image (4)}
     \end{subfigure}

     \begin{subfigure}[]{0.15\textwidth}
         \centering
         \includegraphics[width=2.5cm,height=2.5cm]{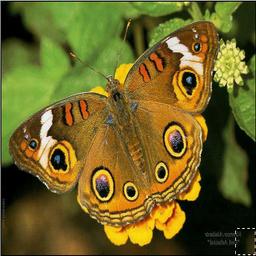}
         \caption*{Image (5)}
     \end{subfigure}
     \quad
      \begin{subfigure}[]{0.15\textwidth}
         \centering
         \includegraphics[width=2.5cm,height=2.5cm]{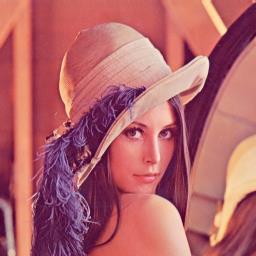}
         \caption*{Image (6)}
     \end{subfigure}
     \quad
      \begin{subfigure}[]{0.15\textwidth}
         \centering
         \includegraphics[width=2.5cm,height=2.5cm]{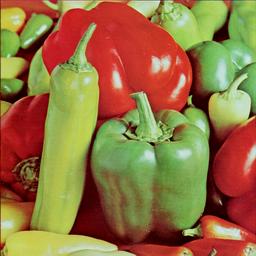}
         \caption*{Image (7)}
     \end{subfigure}
     \quad
      \begin{subfigure}[]{0.15\textwidth}
         \centering
         \includegraphics[width=2.5cm,height=2.5cm]{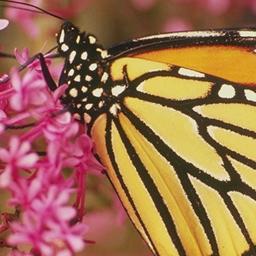}
         \caption*{Image (8)}
     \end{subfigure}
        \caption{The 8 color images with size $256\times 256\times 3$.}
        \label{fig:3}
\end{figure}

\indent
\textbf{Simulations with different parameters of QQR-QNN-SR:}
Since the settings of parameters $\lambda, \ \mu^{0}, \  r$ are closely related to the performance of the proposed method, we investigate the influence of different values of these parameters on the recovery results within a certain range of values.\\
\indent
Firstly, the settings of $\mu^{0}=\{10^{-4},\ 5\times 10^{-4},\ 10^{-3},\ 5\times 10^{-3},\ 10^{-2},\ 5\times 10^{-2}, \ 10^{-1},\ 5\times 10^{-1}, \ 10^{0},\ 10^{1}, \ 10^{2} \}$ are tested with other parameters ($\lambda=10^{-1}, \ r=85, \ \gamma=1.15$) fixed. The results for the eight color images are shown in \cref{fig:4}. After comprehensively considering the recovery results under the different missing ratios of pixels, it can be found from \cref{fig:4} that the proper settings of $\mu^{0}$ are between $[10^{-2}, \ 10^{-1}]$, and the best value of $\mu^{0}$ is $5\times 10^{-2}$.\\
\begin{figure}[htbp]
	\centering
	\begin{subfigure}[]{0.25\textwidth}
		\centering
		\includegraphics[width=4.5cm,height=3cm]{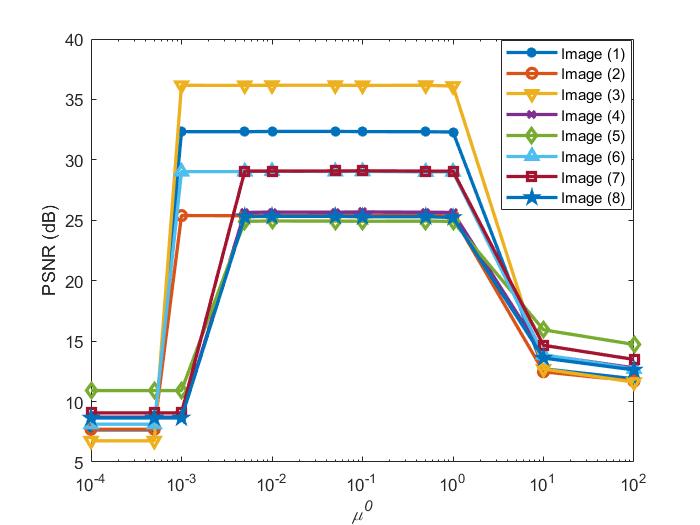}
		\caption{MR=$50\%$, PSNR, $\mu^{0}$}
	\end{subfigure}
	\quad
	\begin{subfigure}[]{0.25\textwidth}
		\centering
		\includegraphics[width=4.5cm,height=3cm]{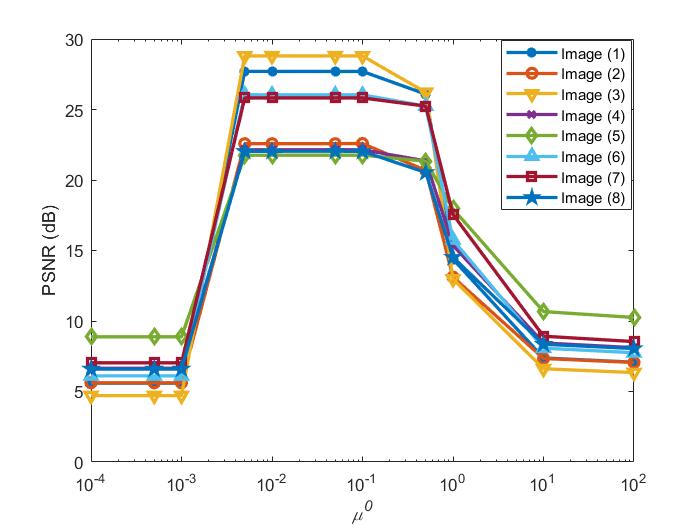}
		\caption{MR=$80\%$, PSNR, $\mu^{0}$}
	\end{subfigure}
	\quad
	\begin{subfigure}[]{0.25\textwidth}
		\centering
		\includegraphics[width=4.5cm,height=3cm]{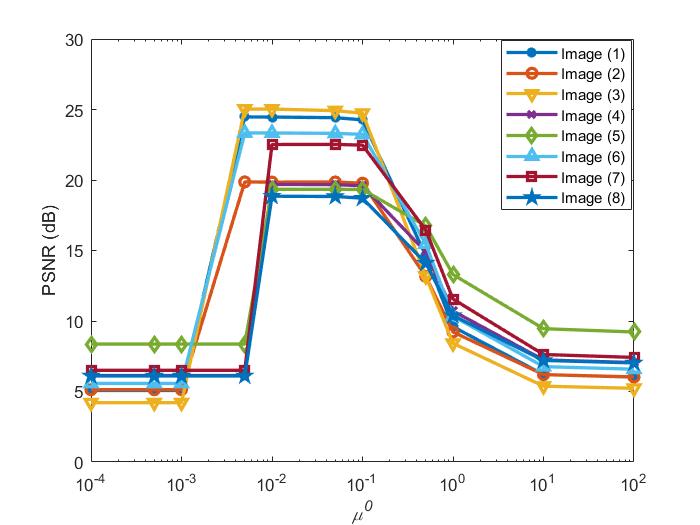}
		\caption{MR=$90\%$, PSNR, $\mu^{0}$}
	\end{subfigure}
	\quad
	
	\begin{subfigure}[]{0.25\textwidth}
		\centering
		\includegraphics[width=4.5cm,height=3cm]{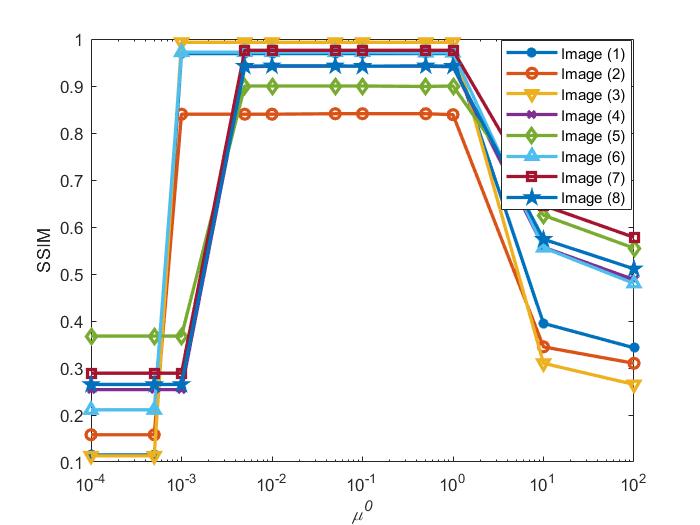}
		\caption{MR=$50\%$, SSIM, $\mu^{0}$}
	\end{subfigure}
	\quad
	\begin{subfigure}[]{0.25\textwidth}
		\centering
		\includegraphics[width=4.5cm,height=3cm]{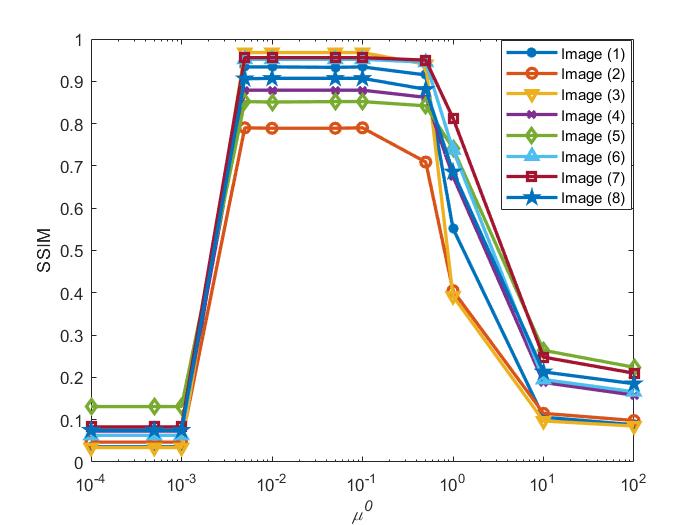}
		\caption{MR=$80\%$, SSIM, $\mu^{0}$}
	\end{subfigure}
	\quad
	\begin{subfigure}[]{0.25\textwidth}
		\centering
		\includegraphics[width=4.5cm,height=3cm]{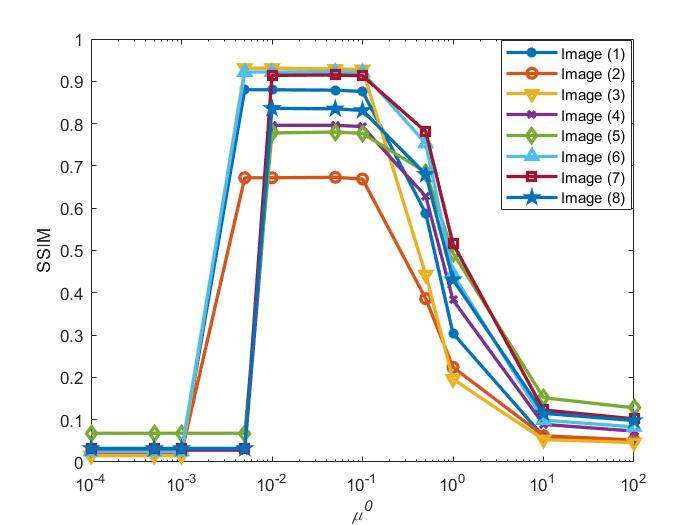}
		\caption{MR=$90\%$, SSIM, $\mu^{0}$}
	\end{subfigure}
	\quad
	
	\begin{subfigure}[]{0.25\textwidth}
		\centering
		\includegraphics[width=4.5cm,height=3cm]{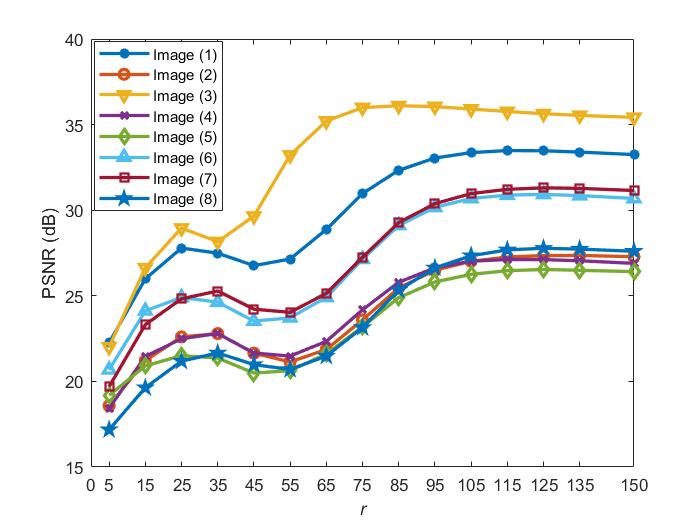}
		\caption{MR=$50\%$, PSNR, $r$}
	\end{subfigure}
	\quad
	\begin{subfigure}[]{0.25\textwidth}
		\centering
		\includegraphics[width=4.5cm,height=3cm]{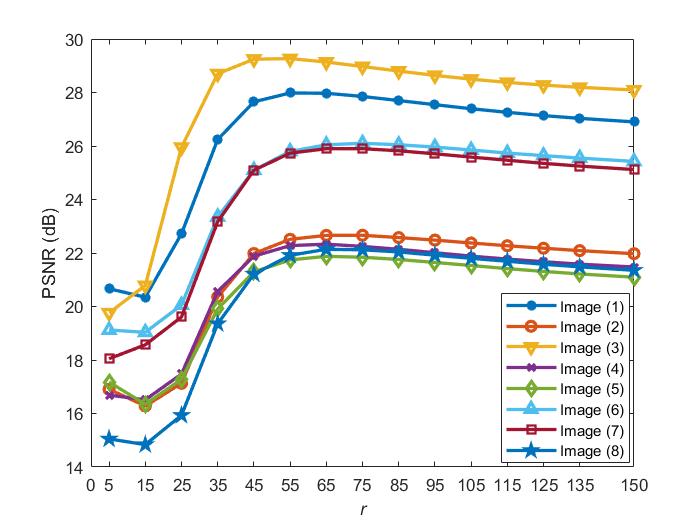}
		\caption{MR=$80\%$, PSNR, $r$}
	\end{subfigure}
	\quad
	\begin{subfigure}[]{0.25\textwidth}
		\centering
		\includegraphics[width=4.5cm,height=3cm]{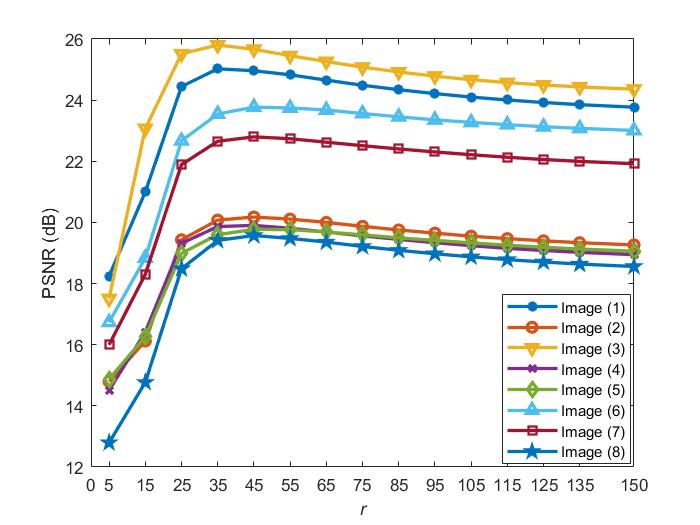}
		\caption{MR=$90\%$, PSNR, $r$}
	\end{subfigure}
	\quad
	
	\begin{subfigure}[]{0.25\textwidth}
		\centering
		\includegraphics[width=4.5cm,height=3cm]{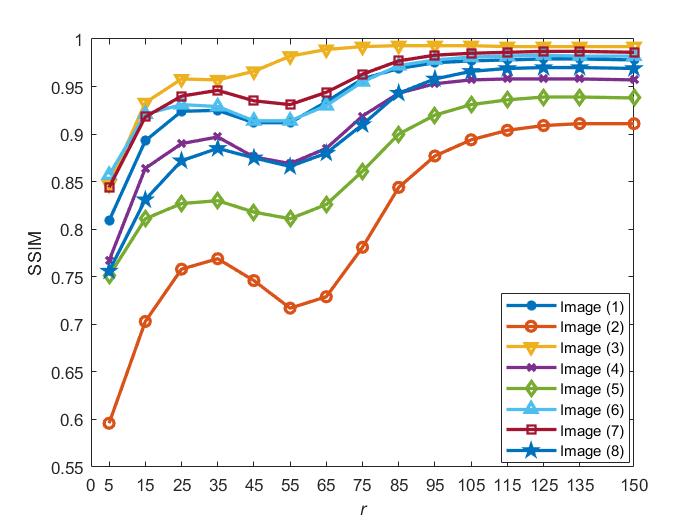}
		\caption{MR=$50\%$, SSIM, $r$}
	\end{subfigure}
	\quad
	\begin{subfigure}[]{0.25\textwidth}
		\centering
		\includegraphics[width=4.5cm,height=3cm]{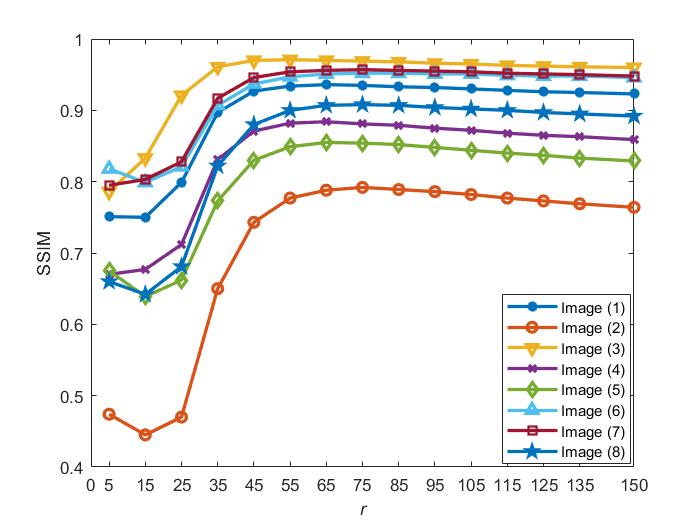}
		\caption{MR=$80\%$, SSIM, $r$}
	\end{subfigure}
	\quad
	\begin{subfigure}[]{0.25\textwidth}
		\centering
		\includegraphics[width=4.5cm,height=3cm]{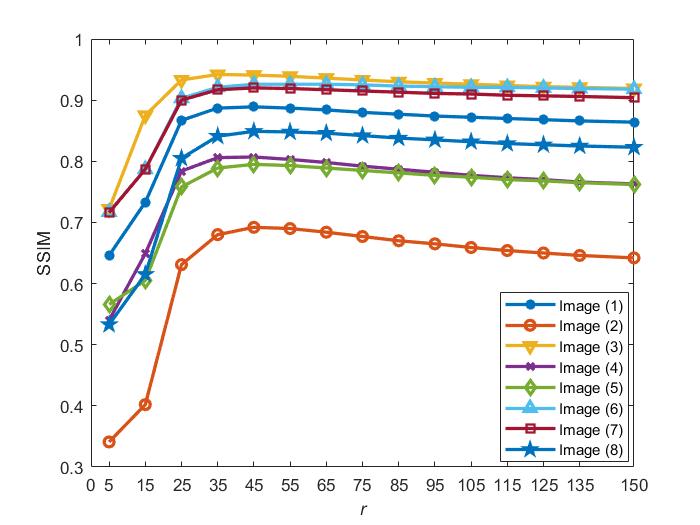}
		\caption{MR=$90\%$, SSIM, $r$}
	\end{subfigure}
	\quad 
	
	\begin{subfigure}[]{0.25\textwidth}
		\centering
		\includegraphics[width=4.5cm,height=3cm]{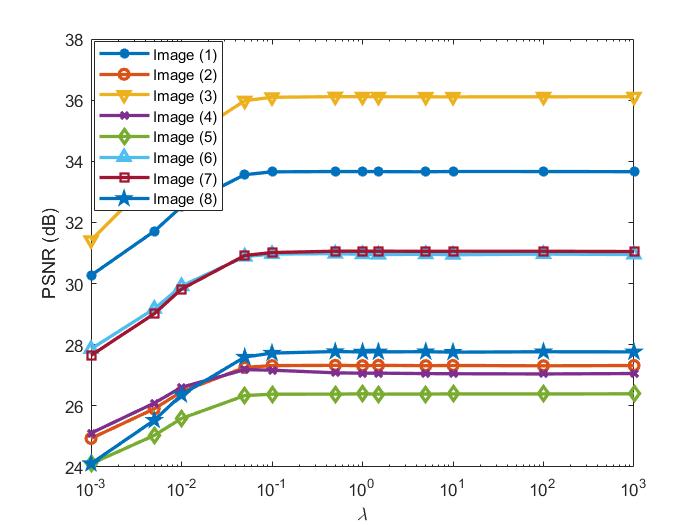}
		\caption{MR=$50\%$, PSNR, $\lambda$}
	\end{subfigure}
	\quad
	\begin{subfigure}[]{0.25\textwidth}
		\centering
		\includegraphics[width=4.5cm,height=3cm]{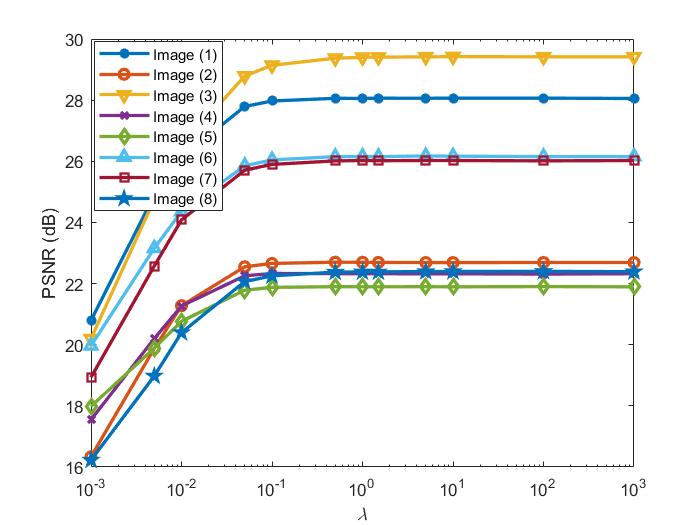}
		\caption{MR=$80\%$, PSNR, $\lambda$}
	\end{subfigure}
	\quad
	\begin{subfigure}[]{0.25\textwidth}
		\centering
		\includegraphics[width=4.5cm,height=3cm]{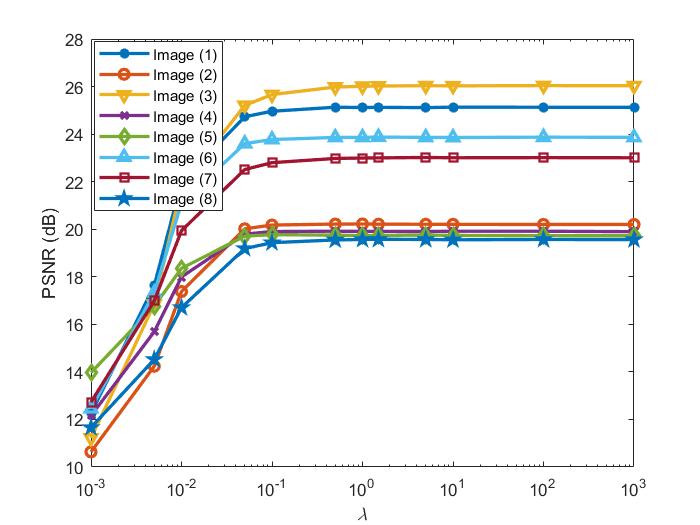}
		\caption{MR=$90\%$, PSNR, $\lambda$}
	\end{subfigure}
	\quad
	
	\begin{subfigure}[]{0.25\textwidth}
		\centering
		\includegraphics[width=4.5cm,height=3cm]{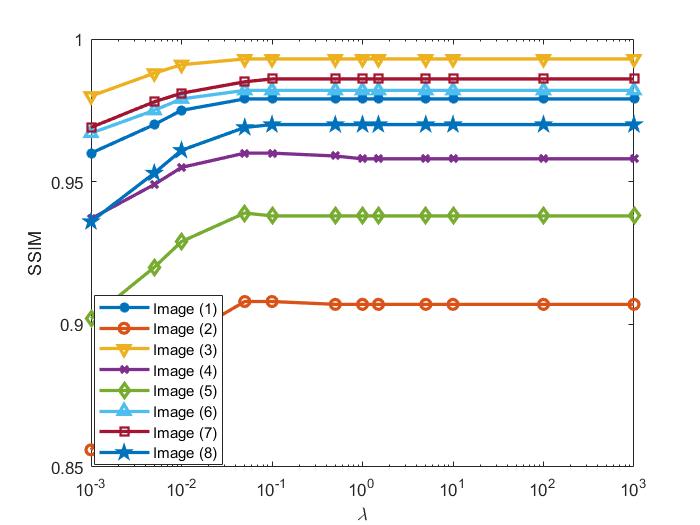}
		\caption{MR=$50\%$, SSIM, $\lambda$}
	\end{subfigure}
	\quad
	\begin{subfigure}[]{0.25\textwidth}
		\centering
		\includegraphics[width=4.5cm,height=3cm]{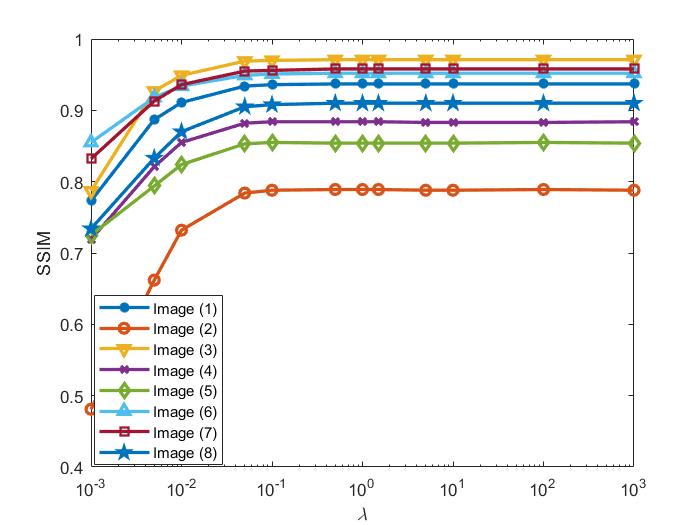}
		\caption{MR=$80\%$, SSIM, $\lambda$}
	\end{subfigure}
	\quad
	\begin{subfigure}[]{0.25\textwidth}
		\centering
		\includegraphics[width=4.5cm,height=3cm]{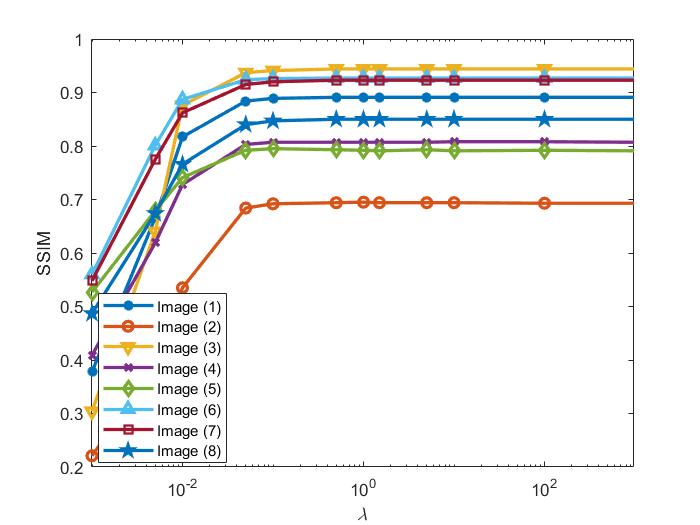}
		\caption{MR=$90\%$, SSIM, $\lambda$}
	\end{subfigure}
	\caption{(a)-(f) The PSNR values (in dB) and SSIM values obtained by the proposed QQR-QNN-SR method using different $\mu^{0}$'s with $\lambda=10^{-1}$ and $r=85$. (g)-(l) The PSNR values (in dB) and SSIM values obtained by the proposed QQR-QNN-SR method using different $r$'s with $\lambda=10^{-1}$ and $\mu^{0}=5\times 10^{-2}$. (m)-(r) The PSNR values (in dB) and SSIM values obtained by the proposed QQR-QNN-SR method using different $\lambda$'s with $\mu^{0}=5\times 10^{-2}$, $r=125$ for MR=$50\%$, $r=65$ for MR=$80\%$, and $r=45$ for MR=$90\%$.}
	\label{fig:4}
\end{figure}
\indent
Secondly, by fixing the values of $\lambda$, $\mu^{0}$, and $\gamma$ to $\lambda=10^{-1}, \ \mu^{0}=5\times 10^{-2}, \ \text{and} \ \gamma=1.15$, the PSNR and SSIM values of the recovered images of the test images obtained by taking different values of $r$, where $r=\{5,\	15,	\ 25,	\ 35,	\ 45,\	55,\	65,\	75,\	85,\	95,\	105,\	115,\	125,\	135,\	150\}$), are shown in the \cref{fig:4}. When the MR is larger, the value of $r$ corresponding to the best recovery result is smaller. This is also in line with the fact that the larger the missing rate, the less original image information contained in the observed image, and the more low-rank the corresponding image is. Specifically, in the case of MR=$50\%$, MR=$70\%$, MR=$80\%$, and MR=$90\%$, the values of $r$ that can obtain the best recovery results are 125, 85, 65, and 45, respectively.\\
\indent
Finally, set $\lambda=\{10^{-3},\ 5\times 10^{-3}, \ 10^{-2},\ 5\times 10^{-2},\ 10^{-1},\ 5\times 10^{-1},\ 10^{0},\ 1.5\times10^{0}, \ 5\times 10^{0},\ 10^{1}, \ 10^{2},$ $\ 10^{3}\}$ and fix other parameters ($\mu^{0}=5\times 10^{-2}, \ r=\{125, 65, 45\} \ \text{correspond to MR}=\{50\%, 80\%, 90\%\}$, respectively, and $\gamma=1.15$) to study the influence of the settings of parameter $\lambda$ on performance of the results. The quantitatively evaluated results are shown in \cref{fig:4}. It can be seen from \cref{fig:4} that when $\lambda \geq 10^{-1}$, the values of PSNR and SSIM no longer increase rapidly and tend to decrease slowly. We can see that relatively good results can be achieved when $\lambda=10^{-1}$.\\
\indent
 \textbf{Experiments on natural color images:} In this simulation, the proposed QQR-QNN-SR method was compared with several state-of-the-art completion methods. For QQR-QNN-SR, we set $\lambda=10^{-1},\ \mu^{0}=5\times 10^{-2}, \gamma =1.15$. And when MR=$\{90\%, \ 80\%, \ 70\%, \ 50\%\}$, $r$ is correspondingly set as 45, 65, 85, and 125. The larger the value of MR, the more missing pixels. The higher the MR value is, the smaller the corresponding rank is.\\
  \indent
 The recovery results of different methods when MR$=70\%$ and the corresponding quantitative evaluation results are shown in \cref{Figuretable1}. As shown in \cref{Figuretable1}, the proposed QQR-QNN-SR method not only achieves the best visual results but also the best quantitative results compared to other methods. Compared with other methods, the two methods TNN-SR and QQR-QNN-SR can achieve better completion results, which illustrates the importance of the sparse prior. And compared with method TNN-SR, the completion effect of method QQR-QNN-SR is better, which is mainly due to the superiority of quaternion to characterize color images. To better demonstrate the superiority of our proposed method, the visual results of Image (2) and Image (8) recovered by the QQR-QNN-SR method and several state-of-the-art methods when MR=$90\%$ are shown in Fig. \ref{fig:0.1IMAGE28}. Comparing the visual results recovered by our proposed method with those obtained by other state-of-the-art methods, it can be found that our method can recover more details of the original images.
\begin{figure}[htbp]
	\begin{minipage}[h]{0.06\linewidth}
		\centering
	\begin{subfigure}{1\textwidth}
		\centering
		\includegraphics[width=1.3cm,height=1.3cm]{random_missing_natural/Image1.jpg}
	\end{subfigure} \hfill\\
	\begin{subfigure}{1\textwidth}
		\centering
		\includegraphics[width=1.3cm,height=1.3cm]{random_missing_natural/Image2.jpg}
	\end{subfigure} \hfill\\
	\begin{subfigure}{1\textwidth}
		\centering
		\includegraphics[width=1.3cm,height=1.3cm]{random_missing_natural/Image3.jpg}
	\end{subfigure} \hfill\\
	\begin{subfigure}{1\textwidth}
		\centering
		\includegraphics[width=1.3cm,height=1.3cm]{random_missing_natural/Image4.jpg}
	\end{subfigure} \hfill\\
	\begin{subfigure}{1\textwidth}
		\centering
		\includegraphics[width=1.3cm,height=1.3cm]{random_missing_natural/Image5.jpg}
	\end{subfigure} \hfill\\
	\begin{subfigure}{1\textwidth}
		\centering
		\includegraphics[width=1.3cm,height=1.3cm]{random_missing_natural/Image6.jpg}
	\end{subfigure} \hfill\\
	\begin{subfigure}{1\textwidth}
		\centering
		\includegraphics[width=1.3cm,height=1.3cm]{random_missing_natural/Image7.jpg}
	\end{subfigure} \hfill\\
	\begin{subfigure}{1\textwidth}
		\centering
		\includegraphics[width=1.3cm,height=1.3cm]{random_missing_natural/Image8.jpg}
	\end{subfigure} 
	\hfill\\
	\subcaption*{(a)}
	\label{a}
\end{minipage} 
    \hfill 
   	\begin{minipage}[h]{0.06\linewidth}
   			\centering
   	\begin{subfigure}{1\textwidth}
   		\centering
   		\includegraphics[width=1.3cm,height=1.3cm]{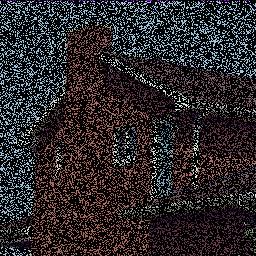}
   	\end{subfigure} \hfill\\
   	\begin{subfigure}{1\textwidth}
   		\centering
   		\includegraphics[width=1.3cm,height=1.3cm]{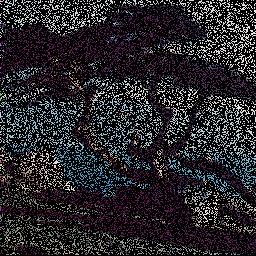}
   	\end{subfigure} \hfill\\
   	\begin{subfigure}{1\textwidth}
   		\centering
   		\includegraphics[width=1.3cm,height=1.3cm]{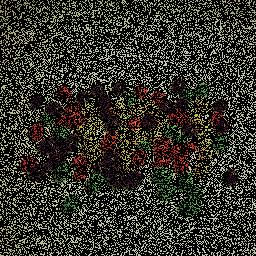}
   	\end{subfigure} \hfill\\
   	\begin{subfigure}{1\textwidth}
   		\centering
   		\includegraphics[width=1.3cm,height=1.3cm]{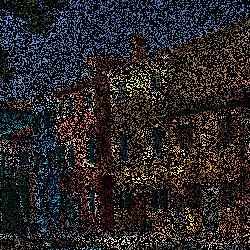}
   	\end{subfigure} \hfill\\
   	\begin{subfigure}{1\textwidth}
   		\centering
   		\includegraphics[width=1.3cm,height=1.3cm]{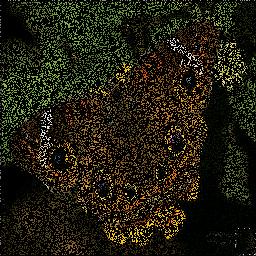}
   	\end{subfigure} \hfill\\
   	\begin{subfigure}{1\textwidth}
   		\centering
   		\includegraphics[width=1.3cm,height=1.3cm]{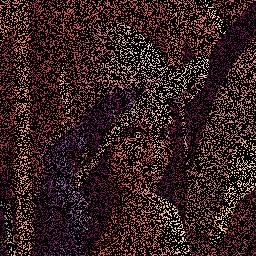}
   	\end{subfigure} \hfill\\
   	\begin{subfigure}{1\textwidth}
   		\centering
   		\includegraphics[width=1.3cm,height=1.3cm]{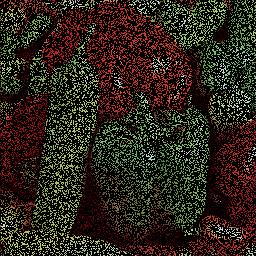}
   	\end{subfigure} \hfill\\
   	\begin{subfigure}{1\textwidth}
   		\centering
   		\includegraphics[width=1.3cm,height=1.3cm]{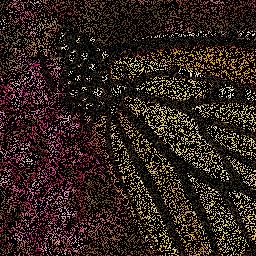}
   	\end{subfigure} 
   	\hfill\\
   	\subcaption*{(b)}
   	\label{a}
   \end{minipage} 
   \hfill 
   	\begin{minipage}[h]{0.06\linewidth}
   			\centering
   	\begin{subfigure}{1\textwidth}
   		\centering
   		\includegraphics[width=1.3cm,height=1.3cm]{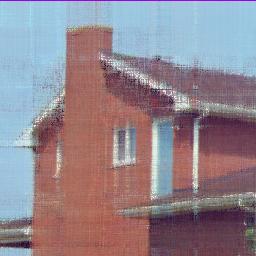}
   	\end{subfigure} \hfill\\
   	\begin{subfigure}{1\textwidth}
   		\centering
   		\includegraphics[width=1.3cm,height=1.3cm]{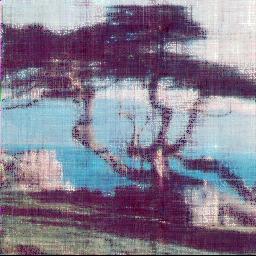}
   	\end{subfigure} \hfill\\
   	\begin{subfigure}{1\textwidth}
   		\centering
   		\includegraphics[width=1.3cm,height=1.3cm]{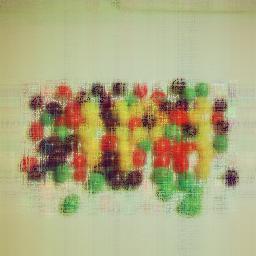}
   	\end{subfigure} \hfill\\
   	\begin{subfigure}{1\textwidth}
   		\centering
   		\includegraphics[width=1.3cm,height=1.3cm]{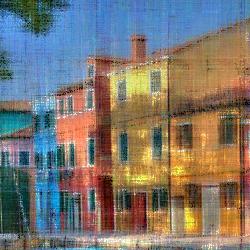}
   	\end{subfigure} \hfill\\
   	\begin{subfigure}{1\textwidth}
   		\centering
   		\includegraphics[width=1.3cm,height=1.3cm]{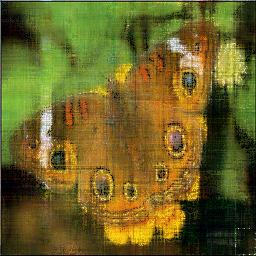}
   	\end{subfigure} \hfill\\
   	\begin{subfigure}{1\textwidth}
   		\centering
   		\includegraphics[width=1.3cm,height=1.3cm]{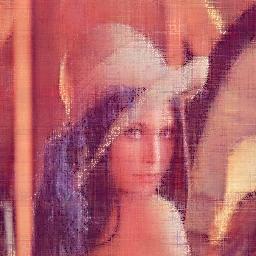}
   	\end{subfigure} \hfill\\
   	\begin{subfigure}{1\textwidth}
   		\centering
   		\includegraphics[width=1.3cm,height=1.3cm]{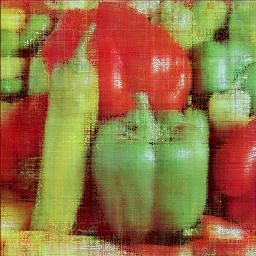}
   	\end{subfigure} \hfill\\
   	\begin{subfigure}{1\textwidth}
   		\centering
   		\includegraphics[width=1.3cm,height=1.3cm]{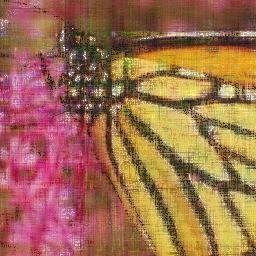}
   	\end{subfigure} 
   	\hfill\\
   	\subcaption*{(c)}
   	\label{a}
   \end{minipage} 
   \hfill 
   \begin{minipage}[h]{0.06\linewidth}
   		\centering
   	\begin{subfigure}{1\textwidth}
   		\centering
   		\includegraphics[width=1.3cm,height=1.3cm]{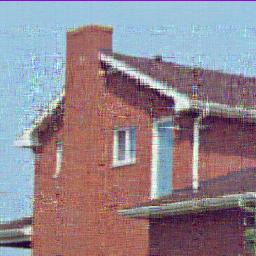}
   	\end{subfigure} \hfill\\
   	\begin{subfigure}{1\textwidth}
   		\centering
   		\includegraphics[width=1.3cm,height=1.3cm]{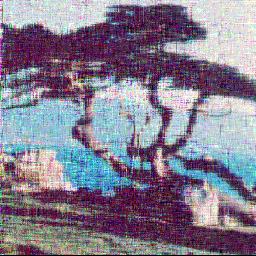}
   	\end{subfigure} \hfill\\
   	\begin{subfigure}{1\textwidth}
   		\centering
   		\includegraphics[width=1.3cm,height=1.3cm]{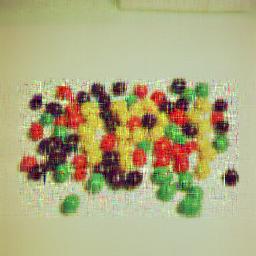}
   	\end{subfigure} \hfill\\
   	\begin{subfigure}{1\textwidth}
   		\centering
   		\includegraphics[width=1.3cm,height=1.3cm]{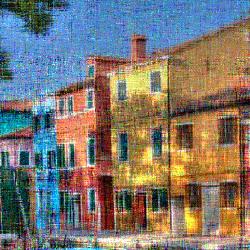}
   	\end{subfigure} \hfill\\
   	\begin{subfigure}{1\textwidth}
   		\centering
   		\includegraphics[width=1.3cm,height=1.3cm]{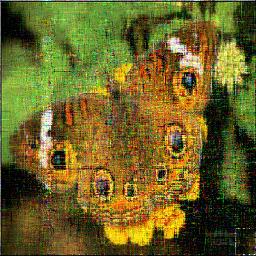}
   	\end{subfigure} \hfill\\
   	\begin{subfigure}{1\textwidth}
   		\centering
   		\includegraphics[width=1.3cm,height=1.3cm]{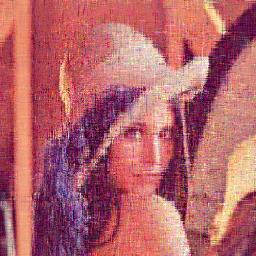}
   	\end{subfigure} \hfill\\
   	\begin{subfigure}{1\textwidth}
   		\centering
   		\includegraphics[width=1.3cm,height=1.3cm]{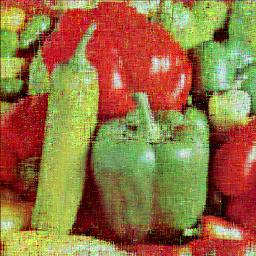}
   	\end{subfigure} \hfill\\
   	\begin{subfigure}{1\textwidth}
   		\centering
   		\includegraphics[width=1.3cm,height=1.3cm]{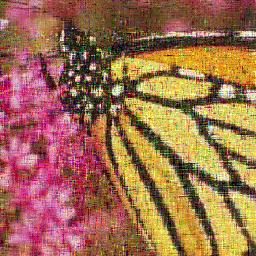}
   	\end{subfigure}                             
   	\hfill\\
   	\subcaption*{(d)}
   	\label{a}
   \end{minipage} 
   \hfill 
   	\begin{minipage}[h]{0.06\linewidth}
   			\centering
   	\begin{subfigure}{1\textwidth}
   		\centering
   		\includegraphics[width=1.3cm,height=1.3cm]{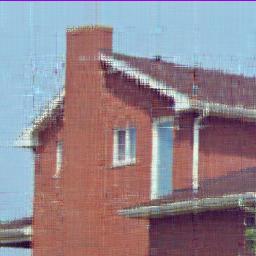}
   	\end{subfigure} \hfill\\
   	\begin{subfigure}{1\textwidth}
   		\centering
   		\includegraphics[width=1.3cm,height=1.3cm]{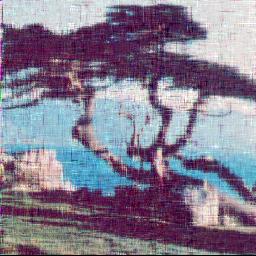}
   	\end{subfigure} \hfill\\
   	\begin{subfigure}{1\textwidth}
   		\centering
   		\includegraphics[width=1.3cm,height=1.3cm]{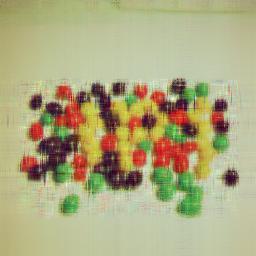}
   	\end{subfigure} \hfill\\
   	\begin{subfigure}{1\textwidth}
   		\centering
   		\includegraphics[width=1.3cm,height=1.3cm]{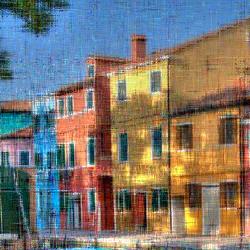}
   	\end{subfigure} \hfill\\
   	\begin{subfigure}{1\textwidth}
   		\centering
   		\includegraphics[width=1.3cm,height=1.3cm]{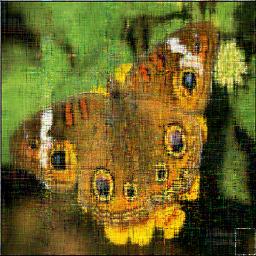}
   	\end{subfigure} \hfill\\
   	\begin{subfigure}{1\textwidth}
   		\centering
   		\includegraphics[width=1.3cm,height=1.3cm]{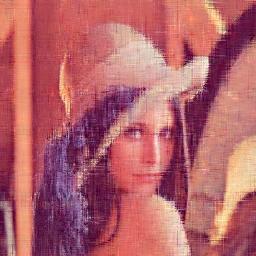}
   	\end{subfigure} \hfill\\
   	\begin{subfigure}{1\textwidth}
   		\centering
   		\includegraphics[width=1.3cm,height=1.3cm]{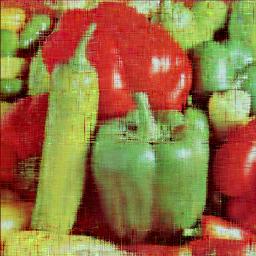}
   	\end{subfigure} \hfill\\
   	\begin{subfigure}{1\textwidth}
   		\centering
   		\includegraphics[width=1.3cm,height=1.3cm]{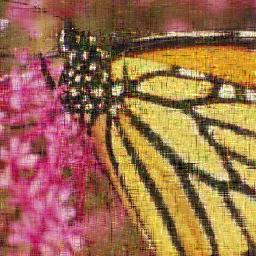}
   	\end{subfigure} 
   	\hfill\\
   	\subcaption*{(e)}
   	\label{a}
   \end{minipage} 
   \hfill 
   \begin{minipage}[h]{0.06\linewidth}
   		\centering
   	\begin{subfigure}{1\textwidth}
   		\centering
   		\includegraphics[width=1.3cm,height=1.3cm]{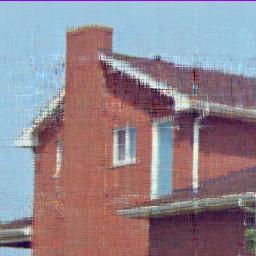}
   	\end{subfigure} \hfill\\
   	\begin{subfigure}{1\textwidth}
   		\centering
   		\includegraphics[width=1.3cm,height=1.3cm]{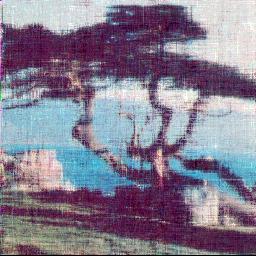}
   	\end{subfigure} \hfill\\
   	\begin{subfigure}{1\textwidth}
   		\centering
   		\includegraphics[width=1.3cm,height=1.3cm]{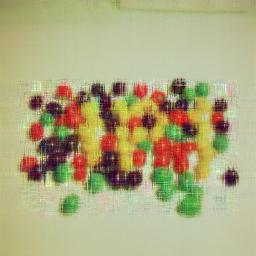}
   	\end{subfigure} \hfill\\
   	\begin{subfigure}{1\textwidth}
   		\centering
   		\includegraphics[width=1.3cm,height=1.3cm]{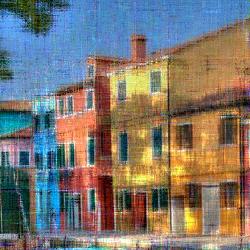}
   	\end{subfigure} \hfill\\
   	\begin{subfigure}{1\textwidth}
   		\centering
   		\includegraphics[width=1.3cm,height=1.3cm]{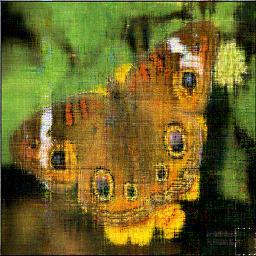}
   	\end{subfigure} \hfill\\
   	\begin{subfigure}{1\textwidth}
   		\centering
   		\includegraphics[width=1.3cm,height=1.3cm]{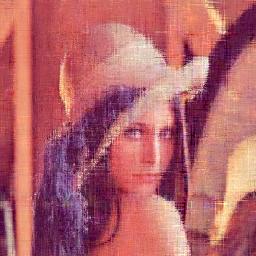}
   	\end{subfigure} \hfill\\
   	\begin{subfigure}{1\textwidth}
   		\centering
   		\includegraphics[width=1.3cm,height=1.3cm]{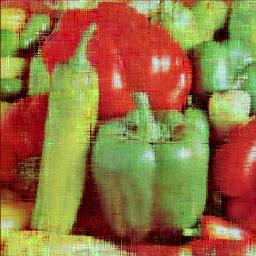}
   	\end{subfigure} \hfill\\
   	\begin{subfigure}{1\textwidth}
   		\centering
   		\includegraphics[width=1.3cm,height=1.3cm]{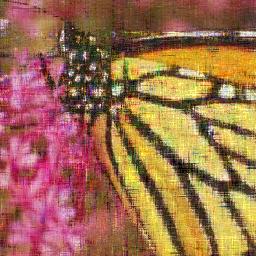}
   	\end{subfigure} 
   	\hfill\\
   	\subcaption*{(f)}
   	\label{a}
   \end{minipage} 
   \hfill 
   \begin{minipage}[h]{0.06\linewidth}
   		\centering
   	\begin{subfigure}{1\textwidth}
   		\centering
   		\includegraphics[width=1.3cm,height=1.3cm]{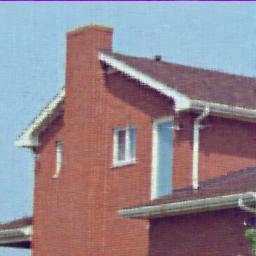}
   	\end{subfigure} \hfill\\
   	\begin{subfigure}{1\textwidth}
   		\centering
   		\includegraphics[width=1.3cm,height=1.3cm]{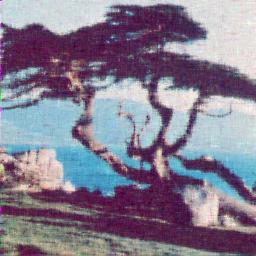}
   	\end{subfigure} \hfill\\
   	\begin{subfigure}{1\textwidth}
   		\centering
   		\includegraphics[width=1.3cm,height=1.3cm]{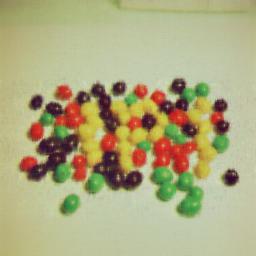}
   	\end{subfigure} \hfill\\
   	\begin{subfigure}{1\textwidth}
   		\centering
   		\includegraphics[width=1.3cm,height=1.3cm]{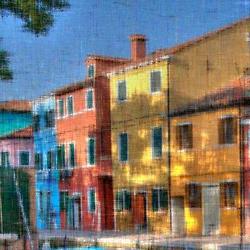}
   	\end{subfigure} \hfill\\
   	\begin{subfigure}{1\textwidth}
   		\centering
   		\includegraphics[width=1.3cm,height=1.3cm]{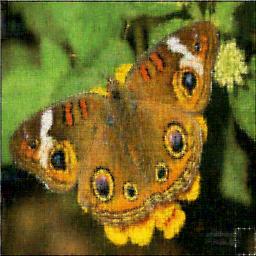}
   	\end{subfigure} \hfill\\
   	\begin{subfigure}{1\textwidth}
   		\centering
   		\includegraphics[width=1.3cm,height=1.3cm]{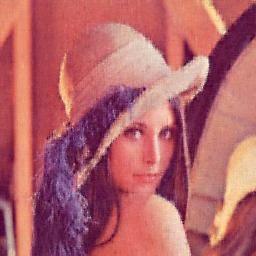}
   	\end{subfigure} \hfill\\
   	\begin{subfigure}{1\textwidth}
   		\centering
   		\includegraphics[width=1.3cm,height=1.3cm]{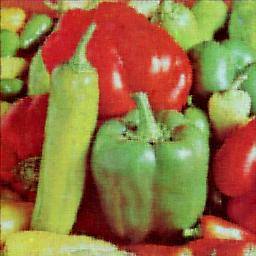}
   	\end{subfigure} \hfill\\
   	\begin{subfigure}{1\textwidth}
   		\centering
   		\includegraphics[width=1.3cm,height=1.3cm]{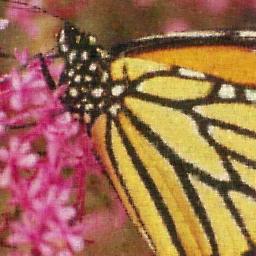}
   	\end{subfigure} 
   	\hfill\\
   	\subcaption*{(g)}
   	\label{a}
   \end{minipage} 
   \hfill 
   \begin{minipage}[h]{0.06\linewidth}
   		\centering
   	\begin{subfigure}{1\textwidth}
   		\centering
   		\includegraphics[width=1.3cm,height=1.3cm]{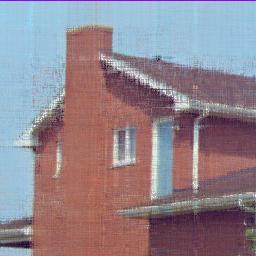}
   	\end{subfigure} \hfill\\
   	\begin{subfigure}{1\textwidth}
   		\centering
   		\includegraphics[width=1.3cm,height=1.3cm]{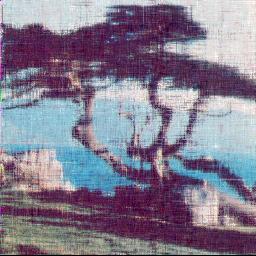}
   	\end{subfigure} \hfill\\
   	\begin{subfigure}{1\textwidth}
   		\centering
   		\includegraphics[width=1.3cm,height=1.3cm]{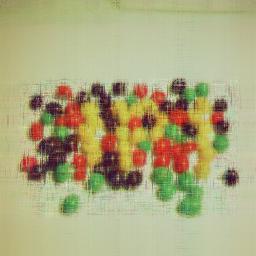}
   	\end{subfigure} \hfill\\
   	\begin{subfigure}{1\textwidth}
   		\centering
   		\includegraphics[width=1.3cm,height=1.3cm]{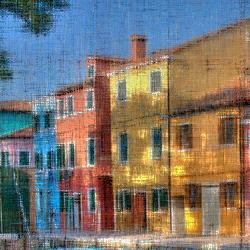}
   	\end{subfigure} \hfill\\
   	\begin{subfigure}{1\textwidth}
   		\centering
   		\includegraphics[width=1.3cm,height=1.3cm]{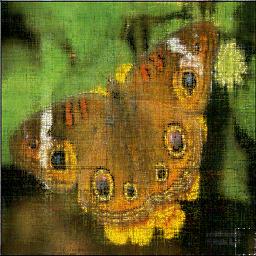}
   	\end{subfigure} \hfill\\
   	\begin{subfigure}{1\textwidth}
   		\centering
   		\includegraphics[width=1.3cm,height=1.3cm]{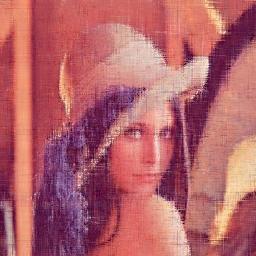}
   	\end{subfigure} \hfill\\
   	\begin{subfigure}{1\textwidth}
   		\centering
   		\includegraphics[width=1.3cm,height=1.3cm]{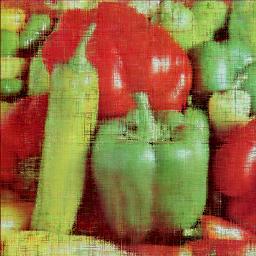}
   	\end{subfigure} \hfill\\
   	\begin{subfigure}{1\textwidth}
   		\centering
   		\includegraphics[width=1.3cm,height=1.3cm]{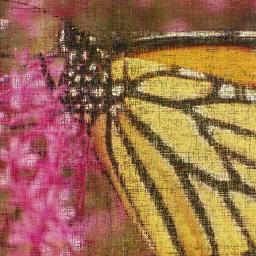}
   	\end{subfigure} 
   	\hfill\\
   	\subcaption*{(h)}
   	\label{a}
   \end{minipage} 
   \hfill 
   	\begin{minipage}[h]{0.06\linewidth}
   			\centering
   	\begin{subfigure}{1\textwidth}
   		\centering
   		\includegraphics[width=1.3cm,height=1.3cm]{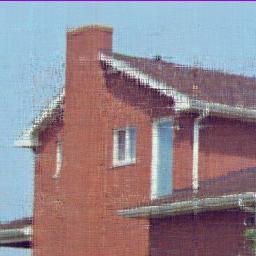}
   	\end{subfigure} \hfill\\
   	\begin{subfigure}{1\textwidth}
   		\centering
   		\includegraphics[width=1.3cm,height=1.3cm]{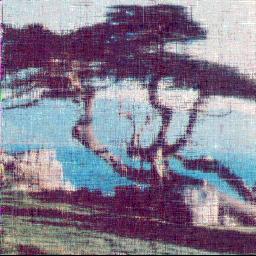}
   	\end{subfigure} \hfill\\
   	\begin{subfigure}{1\textwidth}
   		\centering
   		\includegraphics[width=1.3cm,height=1.3cm]{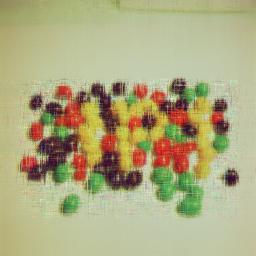}
   	\end{subfigure} \hfill\\
   	\begin{subfigure}{1\textwidth}
   		\centering
   		\includegraphics[width=1.3cm,height=1.3cm]{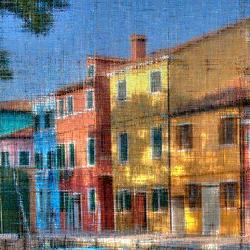}
   	\end{subfigure} \hfill\\
   	\begin{subfigure}{1\textwidth}
   		\centering
   		\includegraphics[width=1.3cm,height=1.3cm]{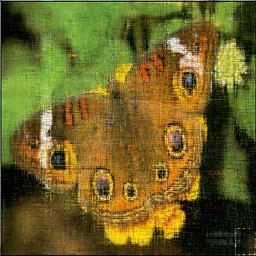}
   	\end{subfigure} \hfill\\
   	\begin{subfigure}{1\textwidth}
   		\centering
   		\includegraphics[width=1.3cm,height=1.3cm]{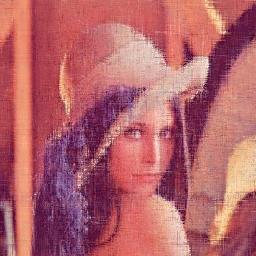}
   	\end{subfigure} \hfill\\
   	\begin{subfigure}{1\textwidth}
   		\centering
   		\includegraphics[width=1.3cm,height=1.3cm]{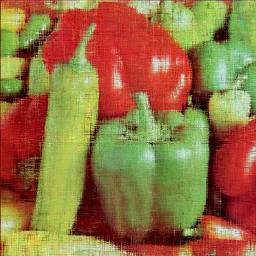}
   	\end{subfigure} \hfill\\
   	\begin{subfigure}{1\textwidth}
   		\centering
   		\includegraphics[width=1.3cm,height=1.3cm]{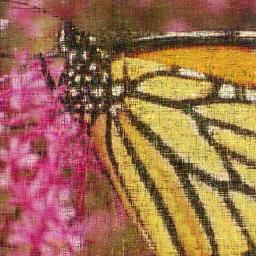}
   	\end{subfigure} 
   	\hfill\\
   	\subcaption*{(i)}
   	\label{a}
   \end{minipage} 
   \hfill 
   \begin{minipage}[h]{0.06\linewidth}
   		\centering
   	\begin{subfigure}{1\textwidth}
   		\centering
   		\includegraphics[width=1.3cm,height=1.3cm]{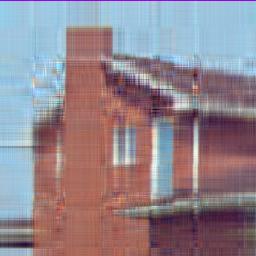}
   	\end{subfigure} \hfill\\
   	\begin{subfigure}{1\textwidth}
   		\centering
   		\includegraphics[width=1.3cm,height=1.3cm]{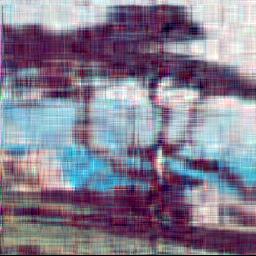}
   	\end{subfigure} \hfill\\
   	\begin{subfigure}{1\textwidth}
   		\centering
   		\includegraphics[width=1.3cm,height=1.3cm]{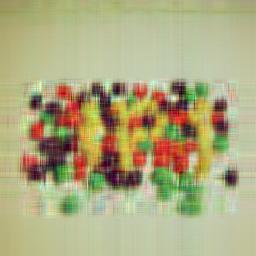}
   	\end{subfigure} \hfill\\
   	\begin{subfigure}{1\textwidth}
   		\centering
   		\includegraphics[width=1.3cm,height=1.3cm]{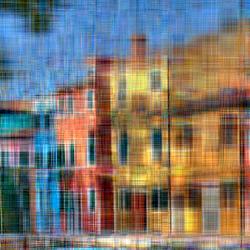}
   	\end{subfigure} \hfill\\
   	\begin{subfigure}{1\textwidth}
   		\centering
   		\includegraphics[width=1.3cm,height=1.3cm]{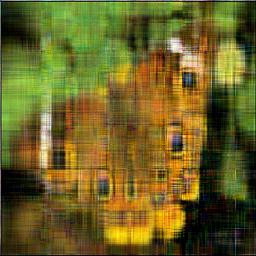}
   	\end{subfigure} \hfill\\
   	\begin{subfigure}{1\textwidth}
   		\centering
   		\includegraphics[width=1.3cm,height=1.3cm]{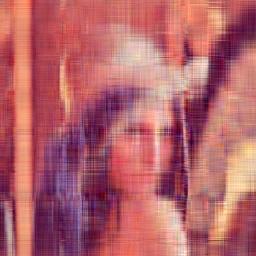}
   	\end{subfigure} \hfill\\
   	\begin{subfigure}{1\textwidth}
   		\centering
   		\includegraphics[width=1.3cm,height=1.3cm]{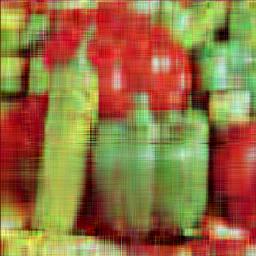}
   	\end{subfigure} \hfill\\
   	\begin{subfigure}{1\textwidth}
   		\centering
   		\includegraphics[width=1.3cm,height=1.3cm]{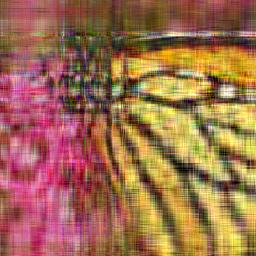}
   	\end{subfigure} 
   	\hfill\\
   	\subcaption*{(j)}
   	\label{a}
   \end{minipage} 
   \hfill 
   \begin{minipage}[h]{0.06\linewidth}
   		\centering
   	\begin{subfigure}{1\textwidth}
   		\centering
   		\includegraphics[width=1.3cm,height=1.3cm]{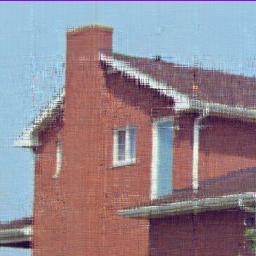}
   	\end{subfigure} \hfill\\
   	\begin{subfigure}{1\textwidth}
   		\centering
   		\includegraphics[width=1.3cm,height=1.3cm]{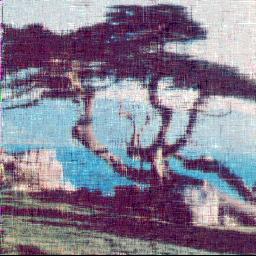}
   	\end{subfigure} \hfill\\
   	\begin{subfigure}{1\textwidth}
   		\centering
   		\includegraphics[width=1.3cm,height=1.3cm]{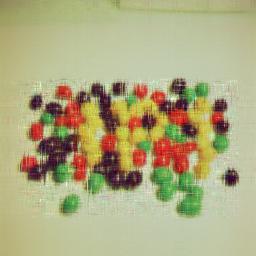}
   	\end{subfigure} \hfill\\
   	\begin{subfigure}{1\textwidth}
   		\centering
   		\includegraphics[width=1.3cm,height=1.3cm]{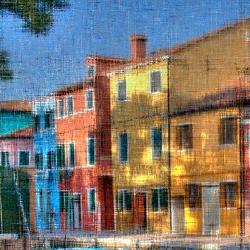}
   	\end{subfigure} \hfill\\
   	\begin{subfigure}{1\textwidth}
   		\centering
   		\includegraphics[width=1.3cm,height=1.3cm]{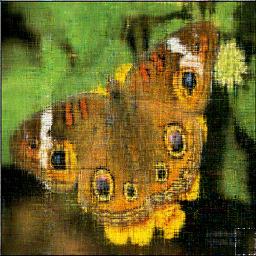}
   	\end{subfigure} \hfill\\
   	\begin{subfigure}{1\textwidth}
   		\centering
   		\includegraphics[width=1.3cm,height=1.3cm]{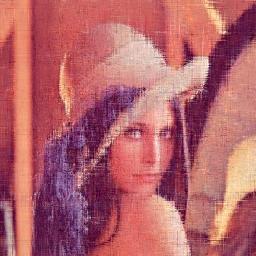}
   	\end{subfigure} \hfill\\
   	\begin{subfigure}{1\textwidth}
   		\centering
   		\includegraphics[width=1.3cm,height=1.3cm]{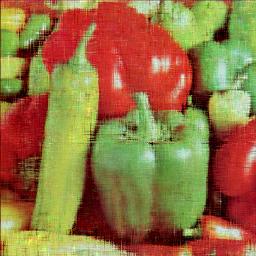}
   	\end{subfigure} \hfill\\
   	\begin{subfigure}{1\textwidth}
   		\centering
   		\includegraphics[width=1.3cm,height=1.3cm]{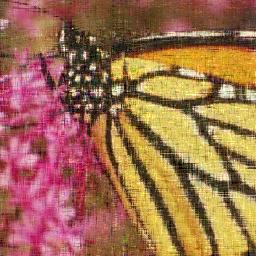}
   	\end{subfigure} 
   	\hfill\\
   	\subcaption*{(k)}
   	\label{a}
   \end{minipage} 
   \hfill 
   \begin{minipage}[h]{0.06\linewidth}
   	\centering
   	\begin{subfigure}{1\textwidth}
   		\centering
   		\includegraphics[width=1.3cm,height=1.3cm]{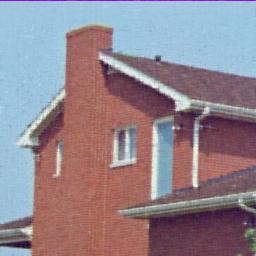}
   	\end{subfigure} \hfill\\
   	\begin{subfigure}{1\textwidth}
   		\centering
   		\includegraphics[width=1.3cm,height=1.3cm]{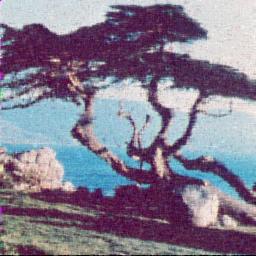}
   	\end{subfigure} \hfill\\
   	\begin{subfigure}{1\textwidth}
   		\centering
   		\includegraphics[width=1.3cm,height=1.3cm]{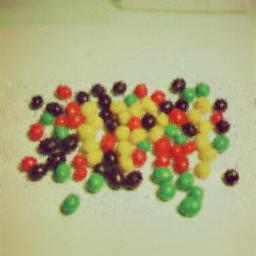}
   	\end{subfigure} \hfill\\
   	\begin{subfigure}{1\textwidth}
   		\centering
   		\includegraphics[width=1.3cm,height=1.3cm]{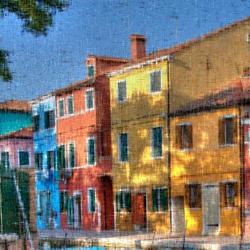}
   	\end{subfigure} \hfill\\
   	\begin{subfigure}{1\textwidth}
   		\centering
   		\includegraphics[width=1.3cm,height=1.3cm]{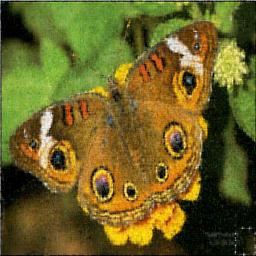}
   	\end{subfigure} \hfill\\
   	\begin{subfigure}{1\textwidth}
   		\centering
   		\includegraphics[width=1.3cm,height=1.3cm]{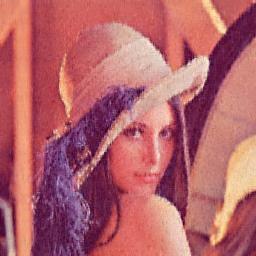}
   	\end{subfigure} \hfill\\
   	\begin{subfigure}{1\textwidth}
   		\centering
   		\includegraphics[width=1.3cm,height=1.3cm]{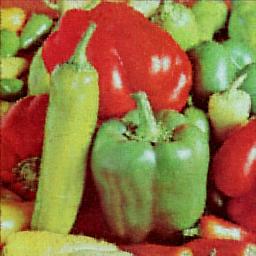}
   	\end{subfigure} \hfill\\
   	\begin{subfigure}{1\textwidth}
   		\centering
   		\includegraphics[width=1.3cm,height=1.3cm]{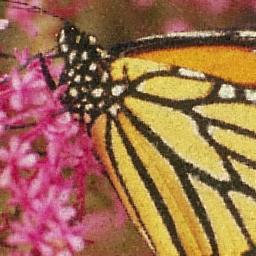}
   	\end{subfigure} 
   	\hfill\\
   	\subcaption*{(l)}
   	\label{a}
   \end{minipage} 
   \hfill \\
	\begin{minipage}{1\linewidth}
	\begin{table}[H]
		\centering
		\resizebox{15.5cm}{4.5cm}{
			\begin{tabular}{|c|c|c|c|c|c|c|c|c|c|c|}		
				\hline
				Methods:& IRLNM-QR  & WNNM  & MC-NC  & TNNR & TNN-SR &LRQMC & LRQA-G  & QLNF &TQLNA& QQR-QNN-SR  
				\\ \toprule
				\hline
				Images:  &\multicolumn{10}{c|}{${\rm{MR}}=50\%$}\\
				\hline
				Image (1)&  27.701/0.938 & 27.919/0.922	& 29.703/0.957	& 30.171/0.953 &33.276/0.978 &	30.221/0.960 &	30.107/0.958 & 24.308/0.858	& 30.867/0.959 &	\textbf{33.470}/\textbf{0.979}\\
				Image (2)&23.481/0.842	& 22.099/0.732&	25.470/0.880	&24.873/0.812	&25.484/0.903&	25.374/0.873&	24.992/0.858&	19.622/0.604	&25.378/0.823	&\textbf{27.351}/\textbf{0.909}\\
				Image (3)& 27.308/0.955&	29.334/0.962	&30.277/0.974	&31.115/0.975	&35.482/0.992	&31.183/0.979&	31.064/0.978&	24.664/0.916&	31.965/0.980	&\textbf{35.636}/\textbf{0.992}\\
				Image (4)&23.610/0.915&	22.298/0.884	&25.296/0.940 & 25.068/0.934&	26.664/0.955	&25.567/0.942&	25.169/0.936&	19.976/0.800&	25.594/0.940&	\textbf{27.113}/\textbf{0.958}\\
				Image (5)&23.243/0.888	&21.235/0.808	&24.292/0.897	&24.003/0.877	&26.013/0.926	&24.399/0.912&	24.132/0.901&	19.822/0.766	&24.426/0.882&	\textbf{26.538}/\textbf{0.939}\\
				Image (6)&26.135/0.956	&25.027/0.932	&27.804/0.967	&27.784/0.964	&30.686/0.982	&28.083/0.968	&27.773/0.966	&22.660/0.911	&28.121/0.966	&\textbf{30.922}/\textbf{0.983}\\
				Image (7)&25.178/0.949	&25.019/0.944	&27.989/0.972	&27.656/0.969&	30.663/0.985&	28.434/0.975	&27.779/0.970&	21.415/0.890	&28.483/0.974&	\textbf{31.307}/\textbf{0.987}\\
				Image (8)&22.125/0.906	&21.271/0.874	&24.617/0.940	&24.041/0.929	&27.471/0.969&	24.338/0.938	&24.158/0.934	&18.013/0.796	&24.529/0.935	&\textbf{27.888}/\textbf{0.971}\\
				\hline
				Aver. &24.848&24.275&	26.931&	26.839&	29.467	&27.200	&26.897&	21.310&	27.420	&\textbf{30.028}\\ \toprule
				Images:  &\multicolumn{10}{c|}{${\rm{MR}}=70\%$}\\
				\hline
				Image (1)& 24.310/0.881 & 23.195/0.816& 25.895/0.904& 25.299/0.883& 29.867/0.955& 25.044/0.893& 25.656/0.897& 23.205/0.831& 25.970/0.895& \textbf{30.283}/\textbf{0.958}\\
				Image (2)& 20.263/0.716 & 17.991/0.526& 21.273/0.726& 20.734/0.663& 23.247/0.834& 21.204/0.731& 21.051/0.710& 18.964/0.554& 21.310/0.678& \textbf{24.360}/\textbf{0.843}  \\
				Image (3)& 23.723/0.912 & 23.650/0.890& 25.867/0.936& 25.242/0.923& 31.078/0.980& 25.250/0.930& 25.776/0.935& 23.600/0.896& 26.172/0.936& \textbf{31.658}/\textbf{0.983}\\
				Image (4)& 20.396/0.832 & 18.175/0.748& 21.218/0.856& 20.871/0.842& 23.389/0.906& 20.928/0.847& 21.215/0.852& 19.266/0.774& 21.557/0.859& \textbf{24.039}/\textbf{0.918}\\
				Image (5)& 20.279/0.808 & 17.438/0.674& 20.175/0.791& 20.415/0.783& 23.006/0.878& 20.390/0.811& 20.647/0.811& 19.088/0.738& 20.849/0.797& \textbf{23.470}/\textbf{0.891}\\
				Image (6)& 22.714/0.916 & 20.740/0.842& 23.902/0.925& 23.460/0.915& 27.301/0.963& 23.899/0.927& 23.556/0.920& 21.868/0.894& 23.791/0.920& \textbf{27.649}/\textbf{0.965}\\
				Image (7)& 21.559/0.898 & 20.502/0.860& 23.562/0.928& 22.920/0.916& 26.940/0.966& 23.031/0.921& 23.265/0.922& 20.766/0.873& 23.746/0.928& \textbf{27.670}/\textbf{0.971}\\
				Image (8)& 18.670/0.826 & 16.757/0.722& 19.629/0.840& 19.164/0.822& 23.603/0.931& 19.117/0.838& 19.620/0.843& 17.267/0.763& 19.930/0.845& \textbf{24.139}/\textbf{0.937 }\\
				\hline
				Aver. &21.489 &	19.806 &22.690 &22.263 &	26.054 &22.358 	&22.598 &	20.503 &22.916 &\textbf{26.659} \\ \toprule
				\hline
				Images  &\multicolumn{10}{c|}{${\rm{MR}}=80\%$}\\
				\hline	
				Image (1)&22.457/0.831  &	20.238/0.714 &	23.218/0.840 &	22.783/0.821 &	27.546/0.930 &	23.267/0.847 &	23.146/0.837 &	22.061/0.791 &	23.208/0.827 &	\textbf{27.972}/\textbf{0.936} \\
				Image (2)&18.554/0.621 &	15.330/0.368 &	17.766/0.524 &	18.540/0.550 &	21.838/0.777 &	18.296/0.586 &	18.882/0.590 &	17.867/0.482 &	18.911/0.550 &	\textbf{22.659}/\textbf{0.788} \\
				Image (3)&21.665/0.872 &	20.183/0.811 &	22.493/0.873 &	22.085/0.867 &	28.581/0.965 &	22.749/0.887 &	22.781/0.886 &	21.951/0.858 &	23.051/0.883 &\textbf{29.140}/\textbf{0.970} \\
				Image (4)&18.554/0.760 &	15.574/0.617 &	17.920/0.727 &	18.535/0.751 &	21.800/0.868 &	19.093/0.778 &	19.066/0.774 &	18.299/0.729 &	19.255/0.778 &	\textbf{22.331}/\textbf{0.884} \\
				Image (5)&18.599/0.750 &	15.455/0.579 &	17.358/0.678 &	18.381/0.716 &	21.366/0.842 &	18.375/0.744 &	18.723/0.743 &	18.010/0.695 &	18.786/0.729 &	\textbf{21.876}/\textbf{0.855} \\
				Image (6)&20.996/0.885 &	18.637/0.772 &	21.429/0.878 &	21.309/0.874 &	25.655/0.948 &	21.725/0.892 &	21.638/0.886 &	20.883/0.869 &	21.756/0.884 &\textbf{26.045}/\textbf{0.951} \\
				Image (7)&19.650/0.857 &	17.415/0.761 &	20.065/0.855 &	20.386/0.863 &	25.245/0.950 &	19.898/0.859 &	20.848/0.874 &	19.858/0.847 &	21.258/0.882 &\textbf{25.898}/\textbf{0.956} \\
				Image (8)&16.947/0.767 &	13.941/0.578 &	16.013/0.703&	17.006/0.758 &	22.177/0.898 &	17.239/0.776 &	17.407/0.770 &	16.160/0.709 &	17.461/0.764 &	\textbf{22.177}/\textbf{0.908} \\
				\hline
				Aver. & 19.678 &	17.097 &	19.533&	19.878 &	24.209 &	20.080 &	20.311 	&19.386 	&20.461	&\textbf{24.762} \\ \toprule	
				\hline
				Images  &\multicolumn{10}{c|}{${\rm{MR}}=90\%$}\\
				\hline
				Image (1)&19.364/0.700 & 16.715/0.544 & 16.827/0.535 &	19.654/0.693 &	24.940/0.885 &	19.271/0.710 &20.053/0.722 &19.605/0.689 &	19.739/0.699 &	\textbf{24.962}/\textbf{0.889} \\
				Image (2)& 15.872/0.434 & 12.203/0.196 & 12.601/0.213 & 15.602/0.372 &	19.797/0.679 &15.731/0.421 & 15.824/0.398 &	15.274/0.320 &	15.123/0.320 &	\textbf{20.176}/\textbf{0.692} \\
				Image (3)&18.574/0.762 &	15.504/0.661 &	16.423/0.564& 19.055/0.787 & 25.568/0.936 &	19.313/0.801 &	19.388/0.798 &	19.000/0.770 &	19.197/0.777 &	\textbf{25.663}/\textbf{0.941} \\
				Image (4)&15.417/0.585 &12.183/0.399 &	13.050/0.407 &	15.679/0.587 &	19.598/0.789 &	15.630/0.598 &	15.891/0.608 &	15.142/0.550 &	15.341/0.571 &	\textbf{19.903}/\textbf{0.807} \\
				Image (5)&16.177/0.647 &12.683/0.430 &	13.895/0.483 &	15.991/0.627& 	19.470/0.782 &	16.351/0.661 &	16.296/0.640 &	15.745/0.593 &	15.443/0.586 & 	\textbf{19.771}/\textbf{0.795} \\
				Image (6)&18.057/0.799 &	14.913/0.607 &	14.911/0.629 &	18.343/0.795 &	23.567/0.924 &	17.979/0.813 &	18.789/0.813 &	18.557/0.800 &	18.320/0.782 &	\textbf{23.772}/\textbf{0.926} \\
				Image (7)& 	16.787/0.767 & 	13.579/0.582 & 14.290/0.613 & 	16.713/0.753 	& 22.417/0.913 & 16.346/0.746 & 17.162/0.764 & 16.946/0.749 & 17.018/0.757 & \textbf{22.799}/\textbf{0.920} \\
				Image (8)& 	13.986/0.626 & 10.782/0.389 & 	12.174/0.469 & 14.651/0.658 & 	19.090/0.838 & 	14.415/0.661 & 	14.472/0.642 & 13.437/0.565 & 13.384/0.571  & \textbf{19.364}/\textbf{0.846 }\\
				\hline
				Aver.&16.779 &	13.570 &	14.271 &	16.961 & 	21.806 	& 16.879 &	17.234 & 	16.713 &	16.696& \textbf{22.051}\\ \toprule
		\end{tabular}}
	\end{table}
     \subcaption*{(m)}
	\end{minipage}
   \setcounter{figure}{5} 
	\caption{The completion results with $\text{MR}=70\%$. (a) is the original image. (b) is the observed image ($\text{MR}=70\%$). (c)-(l) are the completion results of IRLNM-QR, WNNM, MC-NC, TNNR, TNN-SR, LRQMC, LRQA-G, QLNF, TQLNA, and QQR-QNN-SR, respectively. (m) Quantitative quality indexes (PSNR/SSIM) of different methods.}
	\label{Figuretable1}
\end{figure}	
\begin{figure}[htbp]
	\centering
	\includegraphics[width=12cm,height=8cm]{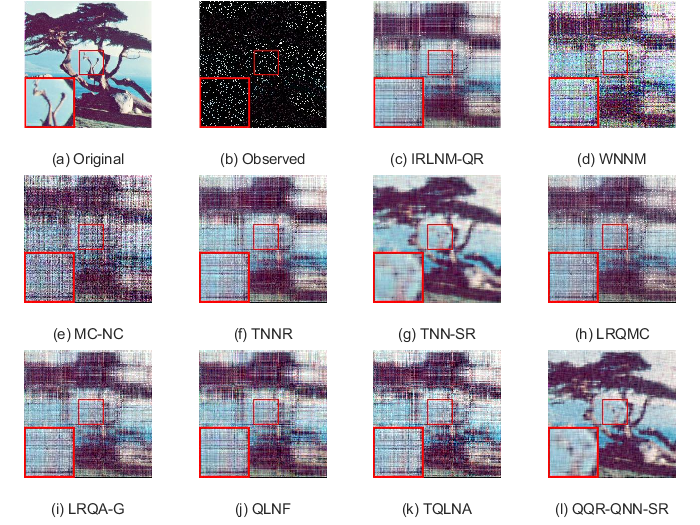}
	\includegraphics[width=12cm,height=8cm]{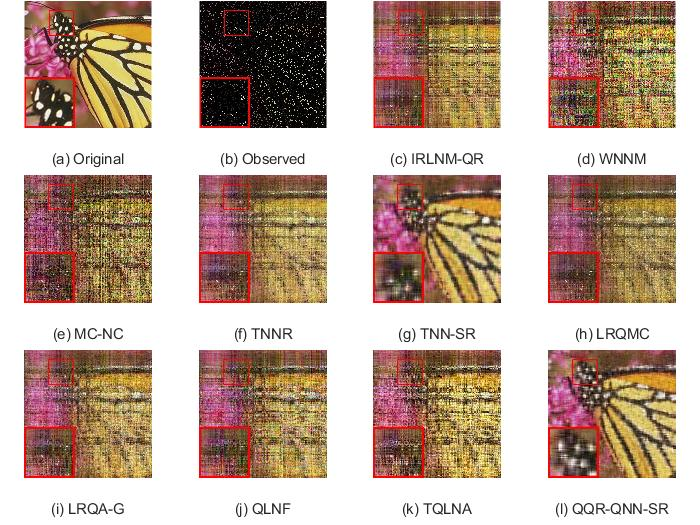}
	\caption{The completion results with $\text{MR}=90\%$ on Image (2) and Image (8). (a) Original image. (b) Observed image ($\text{MR}=90\%$). (c)-(l) are the completion results of IRLNM-QR, WNNM, MC-NC, TNNR, TNN-SR, LRQMC, LRQA-G, QLNF, TQLNA, and QQR-QNN-SR, respectively.}
	\label{fig:0.1IMAGE28}
\end{figure}

 \textbf{Experiments on color medical images:} Medical imaging technology has an irreplaceable role and wide application in clinical diagnosis. In the medical field, quaternion completion facilitates medical diagnosis based on medical images such as roentgen-ray computed tomography (CT), magnetic resonance imaging (MRI), single photon/positron emission computed tomography (SPECT/PET), and their overlay. For example, the medical color image recovery of the human brain based on quaternion completion methods can help doctors to better diagnose various brain diseases, such as Alzheimer's disease (AD). Specifically, because AD often causes atrophy of the region of the hippocampus in patients, MRIs of their hippocampus may become blurred. In order to understand the degree of atrophy of the patient's hippocampus, the patient's healthy hippocampus data needs to be estimated and compared with the atrophic hippocampus data. Therefore, the MRI data of the hippocampus region completed by the quaternion completion method can better help doctors diagnose AD. Furthermore, recovering the data of specific regions of interest in medical images through the quaternion completion method can better help doctors diagnose diseases. In addition, due to equipment or human factors (such as patient body movements during medical examinations), the acquired medical image will be incomplete, which in turn affects the doctor's diagnosis of the disease. Hence, quaternion completion can also be used to solve such problems.\\
 \indent
 In the experiments, the test images are eight color $256\times256$ medical images from the publicly available medical image database “The Whole Brain Atlas”$\footnote{http://www.med.harvard.edu/AANLIB/home.html}$ provided by Harvard Medical School. The eight color images are shown in Fig. \ref{fig:randommissing}. To more comprehensively verify the effectiveness of our method, we not only implement experiments with random missing with MR=$90\%$ but also random block missing based on random block masks. For the random block missing, 2 random rhombus blocks of each image are masked, as shown in Fig. \ref{fig:blockmissing}. And the lengths of the two diagonals of each rhombus are around 44 and 32. \\
 \indent
 In the experiment with MR=$90\%$, the parameters of QQR-QNN-SR are set as $\lambda=10^{-1},\ \mu^{0}=5\times 10^{-2}, \gamma =1.15$, and $r=45$. In the experiment with 2 random rhombus blocks missing, we set the parameters of QQR-QNN-SR as $\lambda=5\times10^{-1},\ \mu^{0}=5\times 10^{-2}, \gamma =1.6$, and $r=190$. We compare the PSNR and SSIM values of the color medical images recovered by our method and IRLNM-QR, WNNM, MC-NC, TNNR, TNN-SR, LRQMC, LRQA-G, QLNF, and TQLNA, and the numerical results of different methods are shown in Fig. \ref{fig:randommissing}. Color medical images recovered by different methods are shown in Fig. \ref{fig:randommissing} and Fig. \ref{fig:blockmissing}. When MR=$90\%$, improving the quality of recovered images is not a simple task. And as can be seen from the Fig. \ref{fig:randommissing}, our method achieves the highest PSNR and SSIM values in almost all recovered images compared to other methods. Additionally, Fig. \ref{fig:randommissing} shows the visual advantage of our approach. Compared with recovering the images with random missing, recovering the images with the loss of random blocks is a more challenging problem to handle because the corruptions of the pixels are not evenly distributed and some important texture information contained in random blocks is lost as a whole. In such experiments, as can be seen from Fig. \ref{fig:blockmissing}, compared with other methods, our method achieves competitive results both numerically and visually.

\begin{figure}[htbp]
\begin{minipage}[h]{0.06\textwidth}
	\centering
	\begin{subfigure}{1\textwidth}
		\centering
		\includegraphics[width=1.3cm,height=1.3cm]{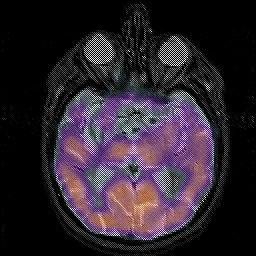}
	\end{subfigure} 
	\hfill\\
	\begin{subfigure}{1\textwidth}
		\centering
		\includegraphics[width=1.3cm,height=1.3cm]{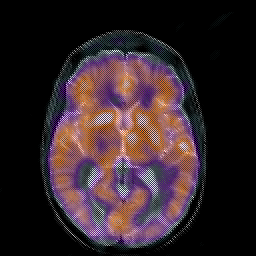}
	\end{subfigure} 
	\hfill\\
	\begin{subfigure}{1\textwidth}
		\centering
		\includegraphics[width=1.3cm,height=1.3cm]{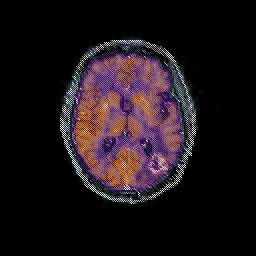}
	\end{subfigure} 
	\hfill\\
	\begin{subfigure}{1\textwidth}
		\centering
		\includegraphics[width=1.3cm,height=1.3cm]{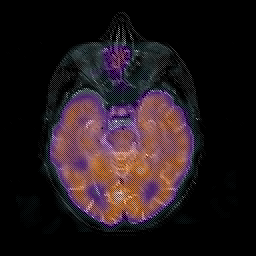}
	\end{subfigure} 
	\hfill\\
	\begin{subfigure}{1\textwidth}
		\centering
		\includegraphics[width=1.3cm,height=1.3cm]{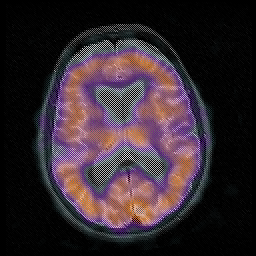}
	\end{subfigure} 
	\hfill\\
	\begin{subfigure}{1\textwidth}
		\centering
		\includegraphics[width=1.3cm,height=1.3cm]{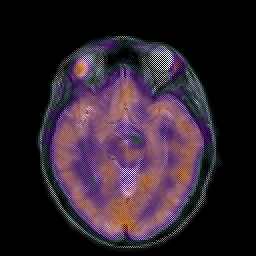}
	\end{subfigure} 
	\hfill\\
	\begin{subfigure}{1\textwidth}
		\centering
		\includegraphics[width=1.3cm,height=1.3cm]{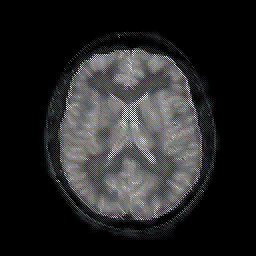}
	\end{subfigure} 
	\hfill\\
	\begin{subfigure}{1\textwidth}
		\centering
		\includegraphics[width=1.3cm,height=1.3cm]{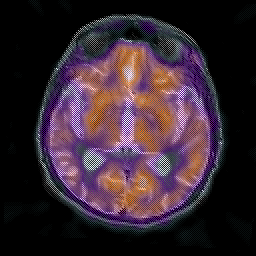}
	\end{subfigure} 
	\hfill\\
	\subcaption*{(a)}
\end{minipage}
\hfill
\begin{minipage}[h]{0.06\textwidth}
	\centering
	\begin{subfigure}{1\textwidth}
		\centering
		\includegraphics[width=1.3cm,height=1.3cm]{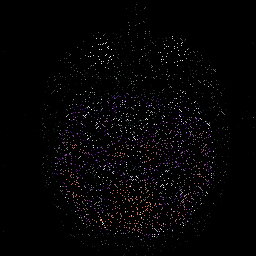}
	\end{subfigure} 
	\hfill\\
	\begin{subfigure}{1\textwidth}
		\centering
		\includegraphics[width=1.3cm,height=1.3cm]{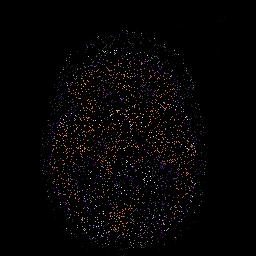}
	\end{subfigure} 
	\hfill\\
	\begin{subfigure}{1\textwidth}
		\centering
		\includegraphics[width=1.3cm,height=1.3cm]{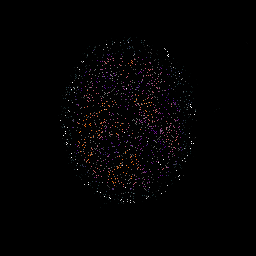}
	\end{subfigure} 
	\hfill\\
	\begin{subfigure}{1\textwidth}
		\centering
		\includegraphics[width=1.3cm,height=1.3cm]{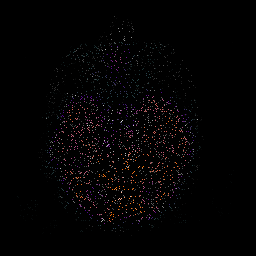}
	\end{subfigure} 
	\hfill\\
	\begin{subfigure}{1\textwidth}
		\centering
		\includegraphics[width=1.3cm,height=1.3cm]{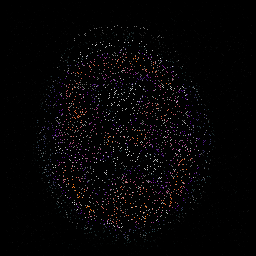}
	\end{subfigure} 
	\hfill\\
	\begin{subfigure}{1\textwidth}
		\centering
		\includegraphics[width=1.3cm,height=1.3cm]{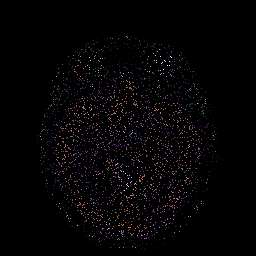}
	\end{subfigure} 
	\hfill\\
	\begin{subfigure}{1\textwidth}
		\centering
		\includegraphics[width=1.3cm,height=1.3cm]{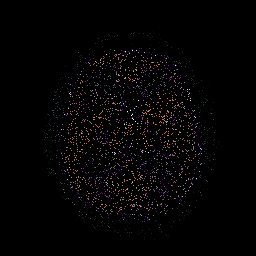}
	\end{subfigure} 
	\hfill\\
	\begin{subfigure}{1\textwidth}
		\centering
		\includegraphics[width=1.3cm,height=1.3cm]{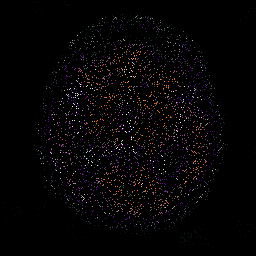}
	\end{subfigure} 
	\hfill\\
	\subcaption*{(b)}
\end{minipage}
\hfill
\begin{minipage}[h]{0.06\textwidth}
	\centering
	\begin{subfigure}{1\textwidth}
		\centering
		\includegraphics[width=1.3cm,height=1.3cm]{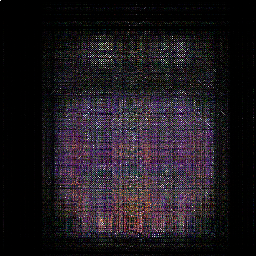}
	\end{subfigure} 
	\hfill\\
	\begin{subfigure}{1\textwidth}
		\centering
		\includegraphics[width=1.3cm,height=1.3cm]{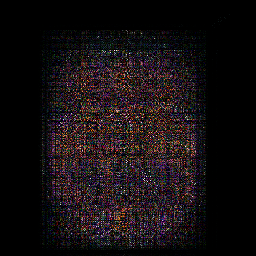}
	\end{subfigure} 
	\hfill\\
	\begin{subfigure}{1\textwidth}
		\centering
		\includegraphics[width=1.3cm,height=1.3cm]{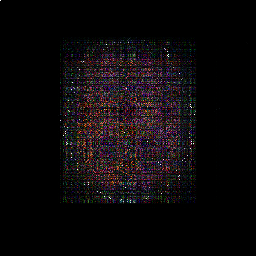}
	\end{subfigure} 
	\hfill\\
	\begin{subfigure}{1\textwidth}
		\centering
		\includegraphics[width=1.3cm,height=1.3cm]{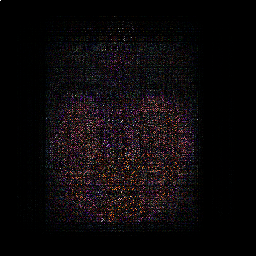}
	\end{subfigure} 
	\hfill\\
	\begin{subfigure}{1\textwidth}
		\centering
		\includegraphics[width=1.3cm,height=1.3cm]{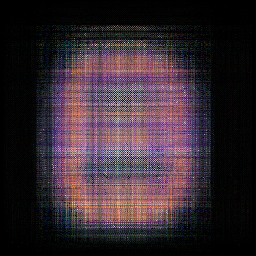}
	\end{subfigure} 
	\hfill\\
	\begin{subfigure}{1\textwidth}
		\centering
		\includegraphics[width=1.3cm,height=1.3cm]{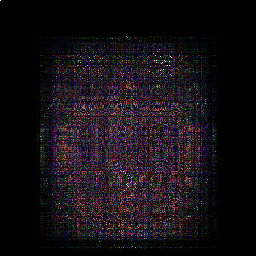}
	\end{subfigure} 
	\hfill\\
	\begin{subfigure}{1\textwidth}
		\centering
		\includegraphics[width=1.3cm,height=1.3cm]{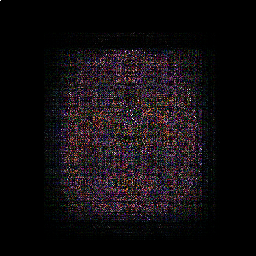}
	\end{subfigure} 
	\hfill\\
	\begin{subfigure}{1\textwidth}
		\centering
		\includegraphics[width=1.3cm,height=1.3cm]{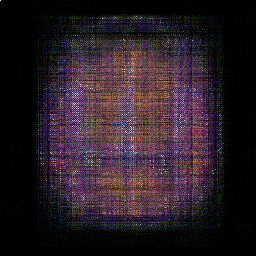}
	\end{subfigure} 
	\hfill\\
	\subcaption*{(c)}
\end{minipage}
\hfill
\begin{minipage}[h]{0.06\textwidth}
	\centering
	\begin{subfigure}{1\textwidth}
		\centering
		\includegraphics[width=1.3cm,height=1.3cm]{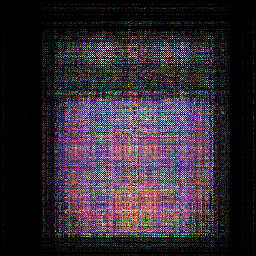}
	\end{subfigure} 
	\hfill\\
	\begin{subfigure}{1\textwidth}
		\centering
		\includegraphics[width=1.3cm,height=1.3cm]{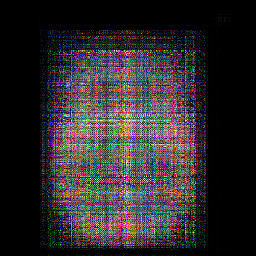}
	\end{subfigure} 
	\hfill\\
	\begin{subfigure}{1\textwidth}
		\centering
		\includegraphics[width=1.3cm,height=1.3cm]{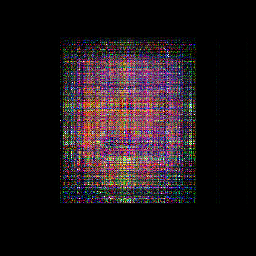}
	\end{subfigure} 
	\hfill\\
	\begin{subfigure}{1\textwidth}
		\centering
		\includegraphics[width=1.3cm,height=1.3cm]{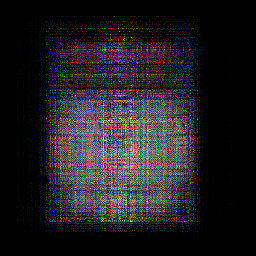}
	\end{subfigure} 
	\hfill\\
	\begin{subfigure}{1\textwidth}
		\centering
		\includegraphics[width=1.3cm,height=1.3cm]{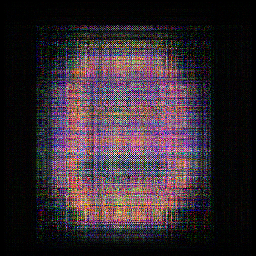}
	\end{subfigure} 
	\hfill\\
	\begin{subfigure}{1\textwidth}
		\centering
		\includegraphics[width=1.3cm,height=1.3cm]{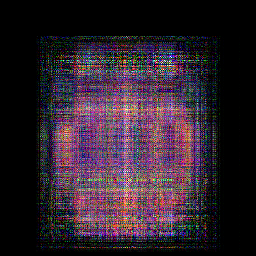}
	\end{subfigure} 
	\hfill\\
	\begin{subfigure}{1\textwidth}
		\centering
		\includegraphics[width=1.3cm,height=1.3cm]{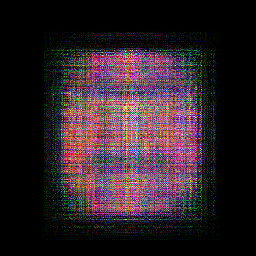}
	\end{subfigure} 
	\hfill\\
	\begin{subfigure}{1\textwidth}
		\centering
		\includegraphics[width=1.3cm,height=1.3cm]{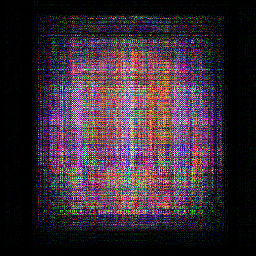}
	\end{subfigure} 
	\hfill\\
	\subcaption*{(d)}
\end{minipage}
\hfill
\begin{minipage}[h]{0.06\textwidth}
	\centering
	\begin{subfigure}{1\textwidth}
		\centering
		\includegraphics[width=1.3cm,height=1.3cm]{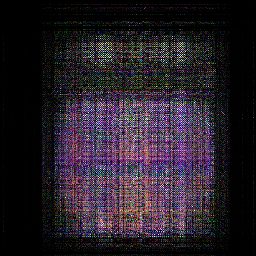}
	\end{subfigure} 
	\hfill\\
	\begin{subfigure}{1\textwidth}
		\centering
		\includegraphics[width=1.3cm,height=1.3cm]{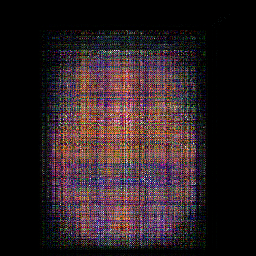}
	\end{subfigure} 
	\hfill\\
	\begin{subfigure}{1\textwidth}
		\centering
		\includegraphics[width=1.3cm,height=1.3cm]{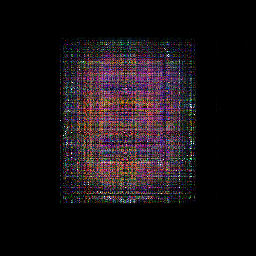}
	\end{subfigure} 
	\hfill\\
	\begin{subfigure}{1\textwidth}
		\centering
		\includegraphics[width=1.3cm,height=1.3cm]{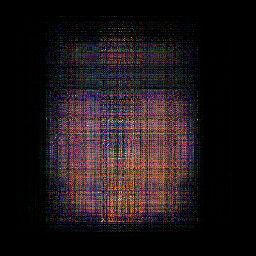}
	\end{subfigure} 
	\hfill\\
	\begin{subfigure}{1\textwidth}
		\centering
		\includegraphics[width=1.3cm,height=1.3cm]{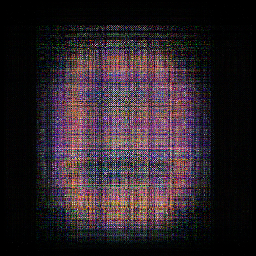}
	\end{subfigure} 
	\hfill\\
	\begin{subfigure}{1\textwidth}
		\centering
		\includegraphics[width=1.3cm,height=1.3cm]{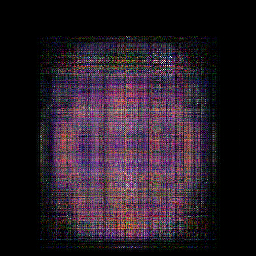}
	\end{subfigure} 
	\hfill\\
	\begin{subfigure}{1\textwidth}
		\centering
		\includegraphics[width=1.3cm,height=1.3cm]{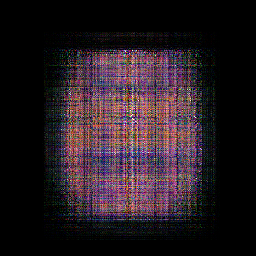}
	\end{subfigure} 
	\hfill\\
	\begin{subfigure}{1\textwidth}
		\centering
		\includegraphics[width=1.3cm,height=1.3cm]{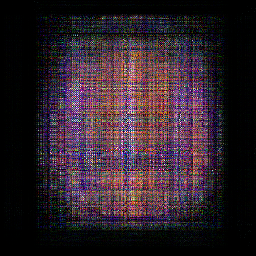}
	\end{subfigure} 
	\hfill\\
	\subcaption*{(e)}
\end{minipage}
\hfill
\begin{minipage}[h]{0.06\textwidth}
	\centering
	\begin{subfigure}{1\textwidth}
		\centering
		\includegraphics[width=1.3cm,height=1.3cm]{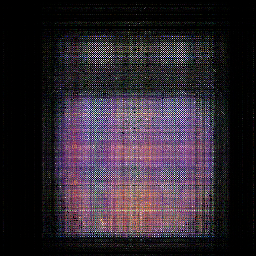}
	\end{subfigure} 
	\hfill\\
	\begin{subfigure}{1\textwidth}
		\centering
		\includegraphics[width=1.3cm,height=1.3cm]{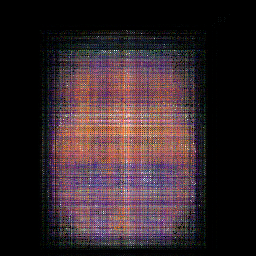}
	\end{subfigure} 
	\hfill\\
	\begin{subfigure}{1\textwidth}
		\centering
		\includegraphics[width=1.3cm,height=1.3cm]{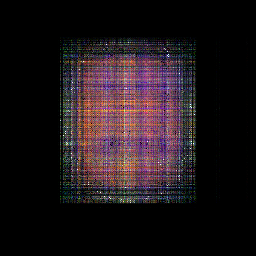}
	\end{subfigure} 
	\hfill\\
	\begin{subfigure}{1\textwidth}
		\centering
		\includegraphics[width=1.3cm,height=1.3cm]{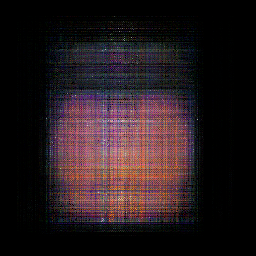}
	\end{subfigure} 
	\hfill\\
	\begin{subfigure}{1\textwidth}
		\centering
		\includegraphics[width=1.3cm,height=1.3cm]{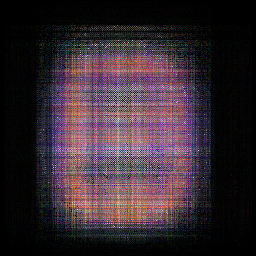}
	\end{subfigure} 
	\hfill\\
	\begin{subfigure}{1\textwidth}
		\centering
		\includegraphics[width=1.3cm,height=1.3cm]{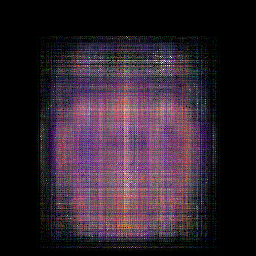}
	\end{subfigure} 
	\hfill\\
	\begin{subfigure}{1\textwidth}
		\centering
		\includegraphics[width=1.3cm,height=1.3cm]{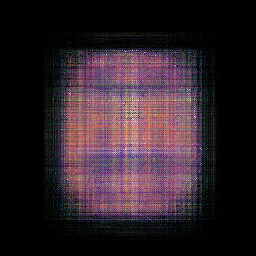}
	\end{subfigure} 
	\hfill\\
	\begin{subfigure}{1\textwidth}
		\centering
		\includegraphics[width=1.3cm,height=1.3cm]{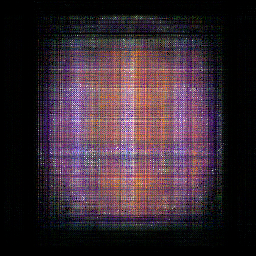}
	\end{subfigure} 
	\hfill\\
	\subcaption*{(f)}
\end{minipage}
\hfill
\begin{minipage}[h]{0.06\textwidth}
	\centering
	\begin{subfigure}{1\textwidth}
		\centering
		\includegraphics[width=1.3cm,height=1.3cm]{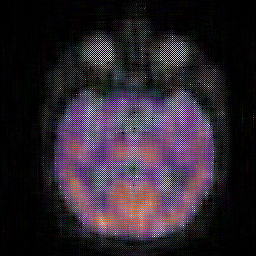}
	\end{subfigure} 
	\hfill\\
	\begin{subfigure}{1\textwidth}
		\centering
		\includegraphics[width=1.3cm,height=1.3cm]{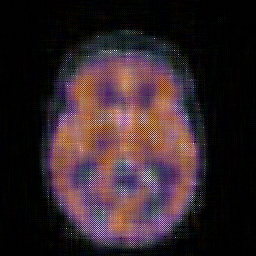}
	\end{subfigure} 
	\hfill\\
	\begin{subfigure}{1\textwidth}
		\centering
		\includegraphics[width=1.3cm,height=1.3cm]{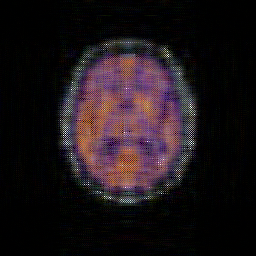}
	\end{subfigure} 
	\hfill\\
	\begin{subfigure}{1\textwidth}
		\centering
		\includegraphics[width=1.3cm,height=1.3cm]{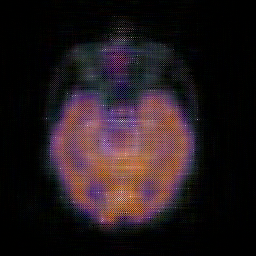}
	\end{subfigure} 
	\hfill\\
	\begin{subfigure}{1\textwidth}
		\centering
		\includegraphics[width=1.3cm,height=1.3cm]{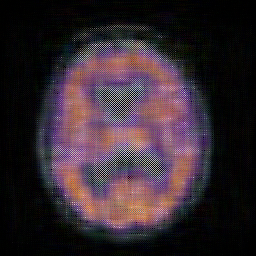}
	\end{subfigure} 
	\hfill\\
	\begin{subfigure}{1\textwidth}
		\centering
		\includegraphics[width=1.3cm,height=1.3cm]{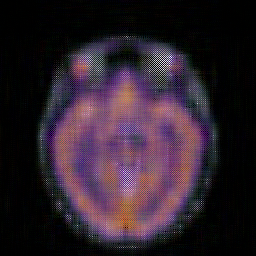}
	\end{subfigure} 
	\hfill\\
	\begin{subfigure}{1\textwidth}
		\centering
		\includegraphics[width=1.3cm,height=1.3cm]{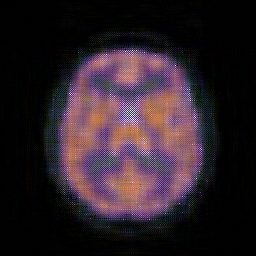}
	\end{subfigure} 
	\hfill\\
	\begin{subfigure}{1\textwidth}
		\centering
		\includegraphics[width=1.3cm,height=1.3cm]{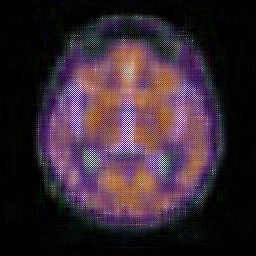}
	\end{subfigure} 
	\hfill\\
	\subcaption*{(g)}
\end{minipage}
\hfill
\begin{minipage}[h]{0.06\textwidth}
	\centering
	\begin{subfigure}{1\textwidth}
		\centering
		\includegraphics[width=1.3cm,height=1.3cm]{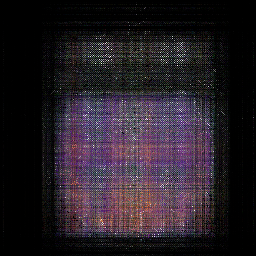}
	\end{subfigure} 
	\hfill\\
	\begin{subfigure}{1\textwidth}
		\centering
		\includegraphics[width=1.3cm,height=1.3cm]{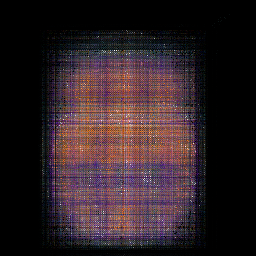}
	\end{subfigure} 
	\hfill\\
	\begin{subfigure}{1\textwidth}
		\centering
		\includegraphics[width=1.3cm,height=1.3cm]{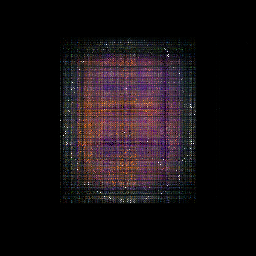}
	\end{subfigure} 
	\hfill\\
	\begin{subfigure}{1\textwidth}
		\centering
		\includegraphics[width=1.3cm,height=1.3cm]{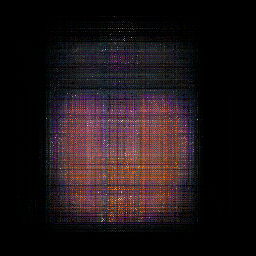}
	\end{subfigure} 
	\hfill\\
	\begin{subfigure}{1\textwidth}
		\centering
		\includegraphics[width=1.3cm,height=1.3cm]{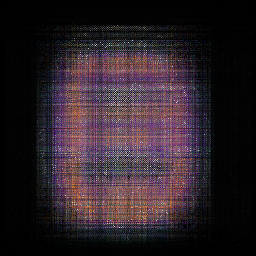}
	\end{subfigure} 
	\hfill\\
	\begin{subfigure}{1\textwidth}
		\centering
		\includegraphics[width=1.3cm,height=1.3cm]{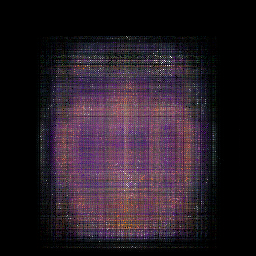}
	\end{subfigure} 
	\hfill\\
	\begin{subfigure}{1\textwidth}
		\centering
		\includegraphics[width=1.3cm,height=1.3cm]{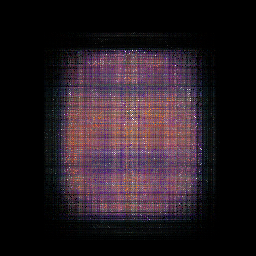}
	\end{subfigure} 
	\hfill\\
	\begin{subfigure}{1\textwidth}
		\centering
		\includegraphics[width=1.3cm,height=1.3cm]{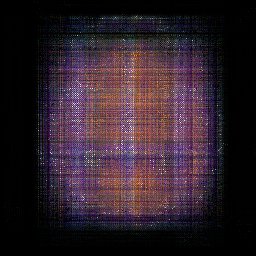}
	\end{subfigure} 
	\hfill\\
	\subcaption*{(h)}
\end{minipage}
\hfill
\begin{minipage}[h]{0.06\textwidth}
	\centering
	\begin{subfigure}{1\textwidth}
		\centering
		\includegraphics[width=1.3cm,height=1.3cm]{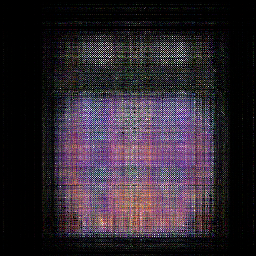}
	\end{subfigure} 
	\hfill\\
	\begin{subfigure}{1\textwidth}
		\centering
		\includegraphics[width=1.3cm,height=1.3cm]{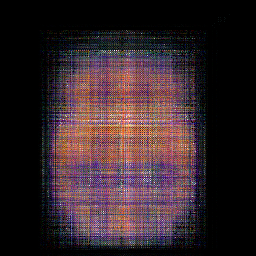}
	\end{subfigure} 
	\hfill\\
	\begin{subfigure}{1\textwidth}
		\centering
		\includegraphics[width=1.3cm,height=1.3cm]{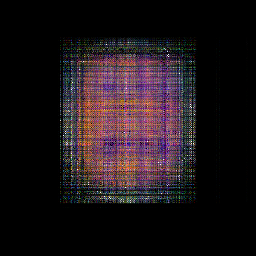}
	\end{subfigure} 
	\hfill\\
	\begin{subfigure}{1\textwidth}
		\centering
		\includegraphics[width=1.3cm,height=1.3cm]{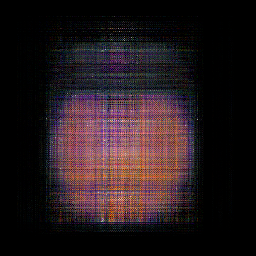}
	\end{subfigure} 
	\hfill\\
	\begin{subfigure}{1\textwidth}
		\centering
		\includegraphics[width=1.3cm,height=1.3cm]{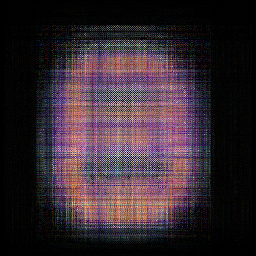}
	\end{subfigure} 
	\hfill\\
	\begin{subfigure}{1\textwidth}
		\centering
		\includegraphics[width=1.3cm,height=1.3cm]{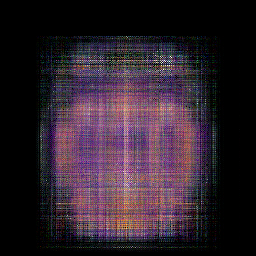}
	\end{subfigure} 
	\hfill\\
	\begin{subfigure}{1\textwidth}
		\centering
		\includegraphics[width=1.3cm,height=1.3cm]{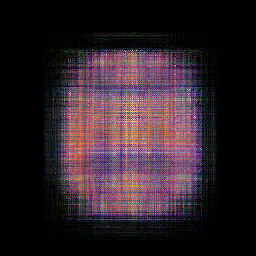}
	\end{subfigure} 
	\hfill\\
	\begin{subfigure}{1\textwidth}
		\centering
		\includegraphics[width=1.3cm,height=1.3cm]{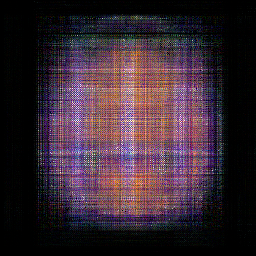}
	\end{subfigure} 
	\hfill\\
	\subcaption*{(i)}
\end{minipage}
\hfill
\begin{minipage}[h]{0.06\textwidth}
	\centering
	\begin{subfigure}{1\textwidth}
		\centering
		\includegraphics[width=1.3cm,height=1.3cm]{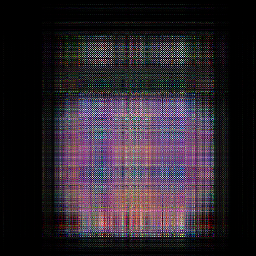}
	\end{subfigure} 
	\hfill\\
	\begin{subfigure}{1\textwidth}
		\centering
		\includegraphics[width=1.3cm,height=1.3cm]{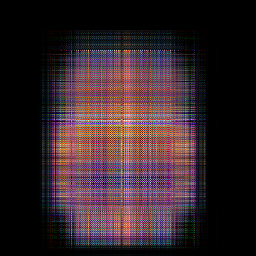}
	\end{subfigure} 
	\hfill\\
	\begin{subfigure}{1\textwidth}
		\centering
		\includegraphics[width=1.3cm,height=1.3cm]{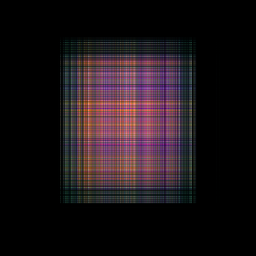}
	\end{subfigure} 
	\hfill\\
	\begin{subfigure}{1\textwidth}
		\centering
		\includegraphics[width=1.3cm,height=1.3cm]{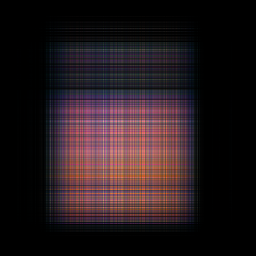}
	\end{subfigure} 
	\hfill\\
	\begin{subfigure}{1\textwidth}
		\centering
		\includegraphics[width=1.3cm,height=1.3cm]{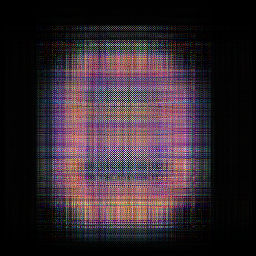}
	\end{subfigure} 
	\hfill\\
	\begin{subfigure}{1\textwidth}
		\centering
		\includegraphics[width=1.3cm,height=1.3cm]{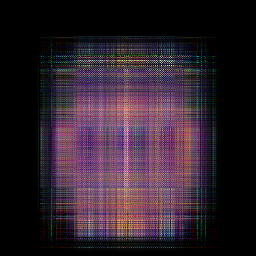}
	\end{subfigure} 
	\hfill\\
	\begin{subfigure}{1\textwidth}
		\centering
		\includegraphics[width=1.3cm,height=1.3cm]{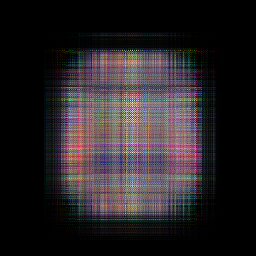}
	\end{subfigure} 
	\hfill\\
	\begin{subfigure}{1\textwidth}
		\centering
		\includegraphics[width=1.3cm,height=1.3cm]{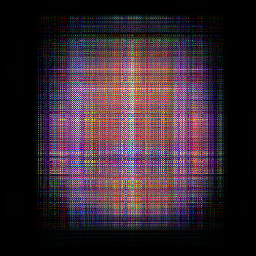}
	\end{subfigure} 
	\hfill\\
	\subcaption*{(j)}
\end{minipage}
\hfill
\begin{minipage}[h]{0.06\textwidth}
	\centering
	\begin{subfigure}{1\textwidth}
		\centering
		\includegraphics[width=1.3cm,height=1.3cm]{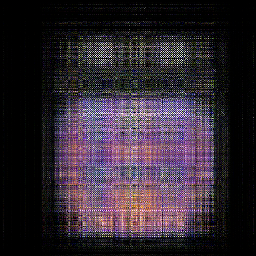}
	\end{subfigure} 
	\hfill\\
	\begin{subfigure}{1\textwidth}
		\centering
		\includegraphics[width=1.3cm,height=1.3cm]{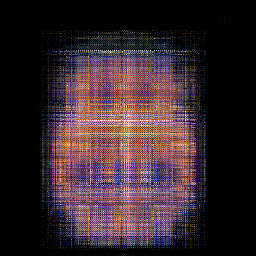}
	\end{subfigure} 
	\hfill\\
	\begin{subfigure}{1\textwidth}
		\centering
		\includegraphics[width=1.3cm,height=1.3cm]{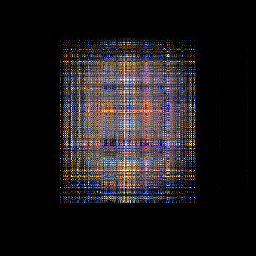}
	\end{subfigure} 
	\hfill\\
	\begin{subfigure}{1\textwidth}
		\centering
		\includegraphics[width=1.3cm,height=1.3cm]{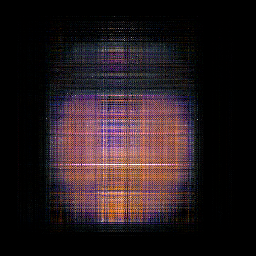}
	\end{subfigure} 
	\hfill\\
	\begin{subfigure}{1\textwidth}
		\centering
		\includegraphics[width=1.3cm,height=1.3cm]{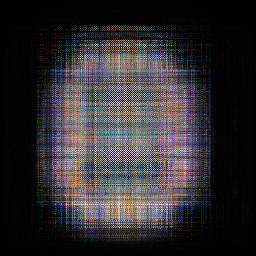}
	\end{subfigure} 
	\hfill\\
	\begin{subfigure}{1\textwidth}
		\centering
		\includegraphics[width=1.3cm,height=1.3cm]{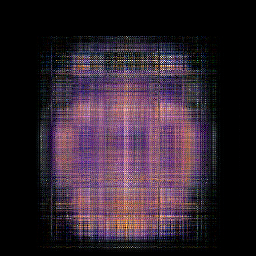}
	\end{subfigure} 
	\hfill\\
	\begin{subfigure}{1\textwidth}
		\centering
		\includegraphics[width=1.3cm,height=1.3cm]{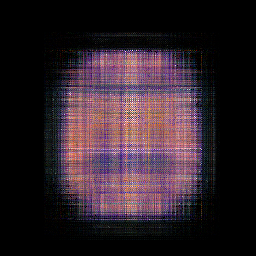}
	\end{subfigure} 
	\hfill\\
	\begin{subfigure}{1\textwidth}
		\centering
		\includegraphics[width=1.3cm,height=1.3cm]{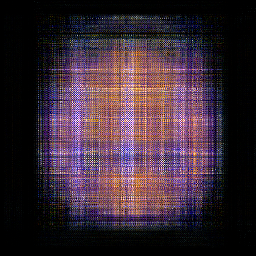}
	\end{subfigure} 
	\hfill\\
	\subcaption*{(k)}
\end{minipage}
\hfill
\begin{minipage}[h]{0.06\textwidth}
	\centering
	\begin{subfigure}{1\textwidth}
		\centering
		\includegraphics[width=1.3cm,height=1.3cm]{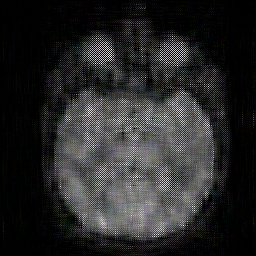}
	\end{subfigure} 
	\hfill\\
	\begin{subfigure}{1\textwidth}
		\centering
		\includegraphics[width=1.3cm,height=1.3cm]{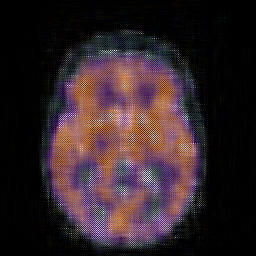}
	\end{subfigure} 
	\hfill\\
	\begin{subfigure}{1\textwidth}
		\centering
		\includegraphics[width=1.3cm,height=1.3cm]{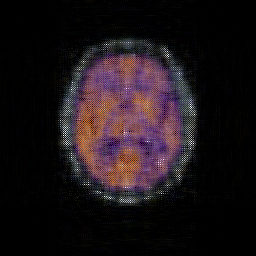}
	\end{subfigure} 
	\hfill\\
	\begin{subfigure}{1\textwidth}
		\centering
		\includegraphics[width=1.3cm,height=1.3cm]{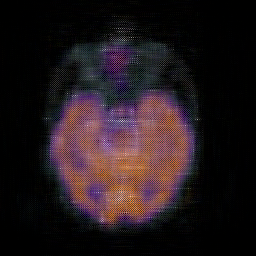}
	\end{subfigure} 
	\hfill\\
	\begin{subfigure}{1\textwidth}
		\centering
		\includegraphics[width=1.3cm,height=1.3cm]{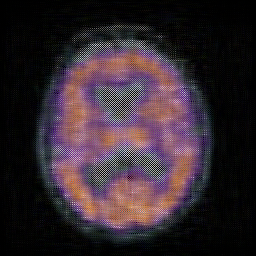}
	\end{subfigure} 
	\hfill\\
	\begin{subfigure}{1\textwidth}
		\centering
		\includegraphics[width=1.3cm,height=1.3cm]{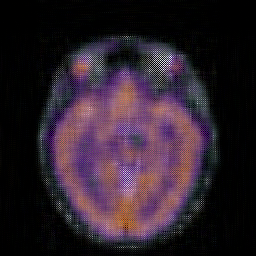}
	\end{subfigure} 
	\hfill\\
	\begin{subfigure}{1\textwidth}
		\centering
		\includegraphics[width=1.3cm,height=1.3cm]{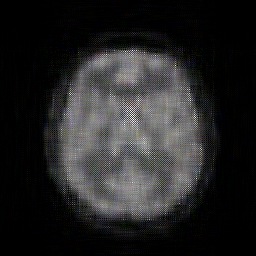}
	\end{subfigure} 
	\hfill\\
	\begin{subfigure}{1\textwidth}
		\centering
		\includegraphics[width=1.3cm,height=1.3cm]{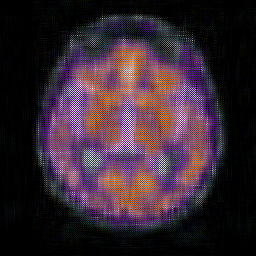}
	\end{subfigure} 
	\hfill\\
	\subcaption*{(l)}
\end{minipage}
	\hfill\\
		\begin{minipage}{1\linewidth}
		   \begin{table}[H]
				\centering
				\resizebox{15.5cm}{2cm}{
					\begin{tabular}{|c|c|c|c|c|c|c|c|c|c|c|}		
						\hline
						Methods:& IRLNM-QR  & WNNM  & MC-NC  & TNNR & TNN-SR &LRQMC & LRQA-G  & QLNF &TQLNA& QQR-QNN-SR  
						\\ \toprule
						\hline
						Images:  &\multicolumn{10}{c|}{${\rm{MR}}=90\%$}\\
						\hline
						Image (9)& 15.090/0.473	& 13.634/0.380	&15.348/0.445	&17.071/0.277	&21.259/0.673	&17.067/0.552	&17.403/0.540	&17.053/0.260	&16.623/0.261 & \textbf{21.375}/\textbf{0.706}\\
						Image (10)&12.624/0.501	&14.347/0.528	&15.719/0.562	&17.541/0.499	&22.202/0.692	&17.861/0.632	&17.755/0.629	&17.523/0.500	&17.208/0.512	&\textbf{22.366}/\textbf{0.733}\\
						Image (11)&15.254/0.685	&14.811/0.684	&16.272/0.694	&17.724/0.616	&20.973/0.667	&17.737/0.731	&17.793/\textbf{0.733}	&17.511/0.588	&16.659/0.609	&\textbf{21.014}/0.724\\
						Image (12)&14.336/0.547	&17.055/0.563	&18.023/0.594	&20.070/0.382	&24.414/0.752	&19.730/0.669	&20.393/0.666	&18.880/0.308	&20.052/0.406	&\textbf{24.573}/\textbf{0.785}\\
						Image (13)&18.235/0.564 &14.861/0.387	&15.918/0.489	&17.947/0.298	&22.672/0.741	&17.534/0.572	&18.334/0.580	& 17.977/0.303	&18.273/0.354	&\textbf{22.989}/\textbf{0.761}\\
						Image (14)&13.908/0.491	&15.239/0.506	&16.772/0.542	&18.458/0.488	&23.742/0.731	&18.386/0.612	&18.642/0.616	&18.214/0.487	&18.557/0.518	&\textbf{24.035}/\textbf{0.771}\\
						Image (15)&12.978/0.545	&15.982/0.557	&17.022/0.585	&19.168/0.549	&23.627/0.743	&18.764/0.659	&19.346	/0.654	&19.179/0.546	&19.539/0.572	&\textbf{23.715}/\textbf{0.780}\\
						Image (16)&15.945/0.507	&13.588/0.358	&15.262/0.454	&16.980/0.218	&21.500/0.662	&16.961/0.547	&17.179/0.527	&16.601/0.230	&16.756/0.247	&\textbf{21.767}/\textbf{0.703}\\
						\hline
						Aver. &14.796		&14.940		&16.292		&18.120		&22.549		&18.005		&18.356		&17.867		&17.958		&\textbf{22.729}\\ \toprule       				
				\end{tabular}}
				\label{Table4}
			\end{table}
		\subcaption*{(m)}
	    \end{minipage}
    \setcounter{figure}{7} 
	\caption{\label{fig:randommissing} (a) Original color medical images. (b) The observed color medical images (MR=$90\%$). (c)-(l) are the completion results of IRLNM-QR, WNNM, MC-NC, TNNR, TNN-SR, LRQMC, LRQA-G, QLNF, TQLNA, and QQR-QNN-SR, respectively. (m) Quantitative quality indexes (PSNR/SSIM) of different methods on color medical images (MR=$90\%$).}
\end{figure}

\begin{figure}[htbp]	
	\begin{minipage}[h]{0.06\textwidth}
		\centering
		\begin{subfigure}{1\textwidth}
			\centering
			\includegraphics[width=1.3cm,height=1.3cm]{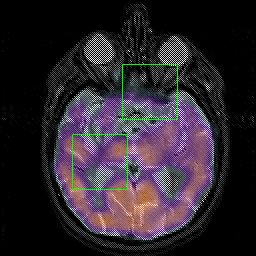}
		\end{subfigure} 
		\hfill\\
		\begin{subfigure}{1\textwidth}
			\centering
			\includegraphics[width=1.3cm,height=1.3cm]{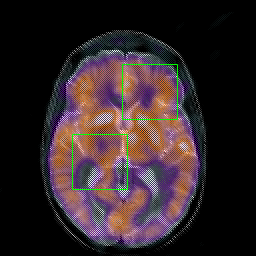}
		\end{subfigure} 
		\hfill\\
		\begin{subfigure}{1\textwidth}
			\centering
			\includegraphics[width=1.3cm,height=1.3cm]{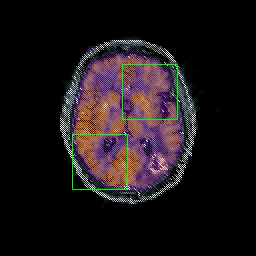}
		\end{subfigure} 
		\hfill\\
		\begin{subfigure}{1\textwidth}
			\centering
			\includegraphics[width=1.3cm,height=1.3cm]{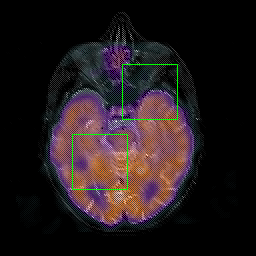}
		\end{subfigure} 
		\hfill\\
		\begin{subfigure}{1\textwidth}
			\centering
			\includegraphics[width=1.3cm,height=1.3cm]{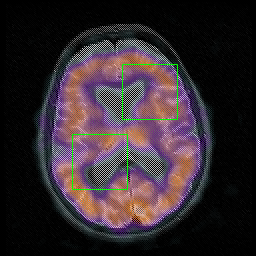}
		\end{subfigure} 
		\hfill\\
		\begin{subfigure}{1\textwidth}
			\centering
			\includegraphics[width=1.3cm,height=1.3cm]{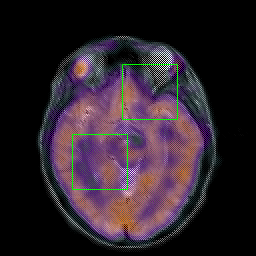}
		\end{subfigure} 
		\hfill\\
		\begin{subfigure}{1\textwidth}
			\centering
			\includegraphics[width=1.3cm,height=1.3cm]{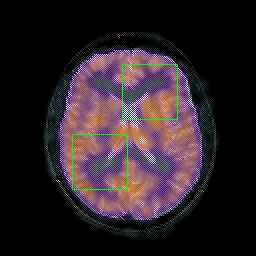}
		\end{subfigure} 
		\hfill\\
		\begin{subfigure}{1\textwidth}
			\centering
			\includegraphics[width=1.3cm,height=1.3cm]{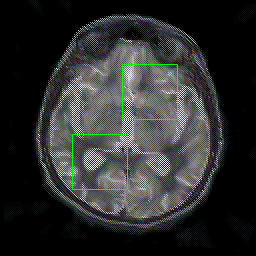}
		\end{subfigure} 
		\hfill\\
		\subcaption*{(a)}
	\end{minipage}
	\hfill
	\begin{minipage}[h]{0.06\textwidth}
		\centering
		\begin{subfigure}{1\textwidth}
			\centering
			\includegraphics[width=1.3cm,height=1.3cm]{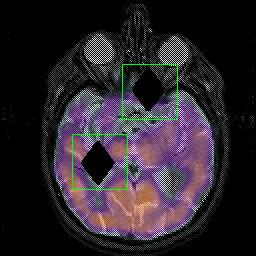}
		\end{subfigure} 
		\hfill\\
		\begin{subfigure}{1\textwidth}
			\centering
			\includegraphics[width=1.3cm,height=1.3cm]{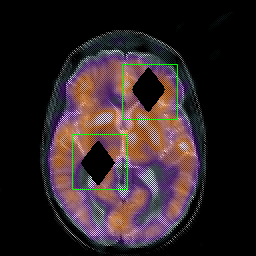}
		\end{subfigure} 
		\hfill\\
		\begin{subfigure}{1\textwidth}
			\centering
			\includegraphics[width=1.3cm,height=1.3cm]{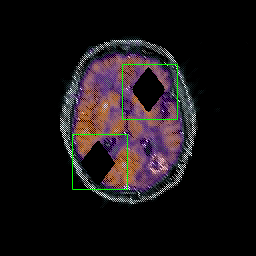}
		\end{subfigure} 
		\hfill\\
		\begin{subfigure}{1\textwidth}
			\centering
			\includegraphics[width=1.3cm,height=1.3cm]{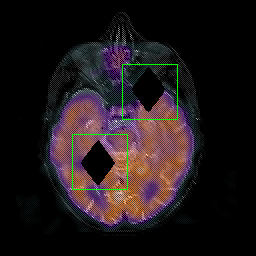}
		\end{subfigure} 
		\hfill\\
		\begin{subfigure}{1\textwidth}
			\centering
			\includegraphics[width=1.3cm,height=1.3cm]{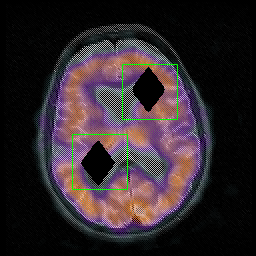}
		\end{subfigure} 
		\hfill\\
		\begin{subfigure}{1\textwidth}
			\centering
			\includegraphics[width=1.3cm,height=1.3cm]{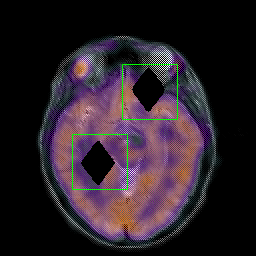}
		\end{subfigure} 
		\hfill\\
		\begin{subfigure}{1\textwidth}
			\centering
			\includegraphics[width=1.3cm,height=1.3cm]{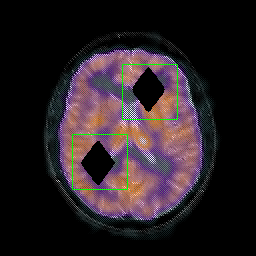}
		\end{subfigure} 
		\hfill\\
		\begin{subfigure}{1\textwidth}
			\centering
			\includegraphics[width=1.3cm,height=1.3cm]{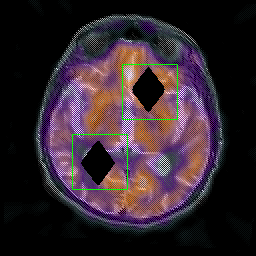}
		\end{subfigure} 
		\hfill\\
		\subcaption*{(b)}
	\end{minipage}
	\hfill
		\begin{minipage}[h]{0.06\textwidth}
		\centering
		\begin{subfigure}{1\textwidth}
			\centering
			\includegraphics[width=1.3cm,height=1.3cm]{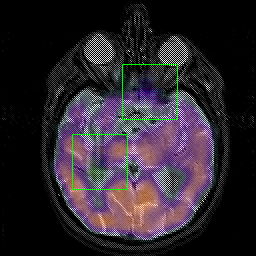}
		\end{subfigure} 
		\hfill\\
		\begin{subfigure}{1\textwidth}
			\centering
			\includegraphics[width=1.3cm,height=1.3cm]{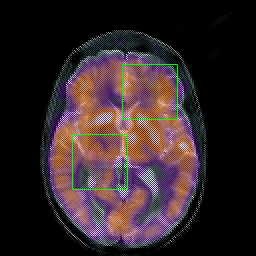}
		\end{subfigure} 
		\hfill\\
		\begin{subfigure}{1\textwidth}
			\centering
			\includegraphics[width=1.3cm,height=1.3cm]{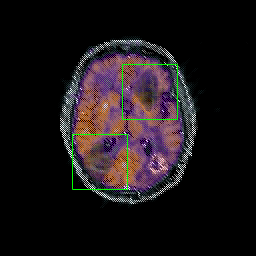}
		\end{subfigure} 
		\hfill\\
		\begin{subfigure}{1\textwidth}
			\centering
			\includegraphics[width=1.3cm,height=1.3cm]{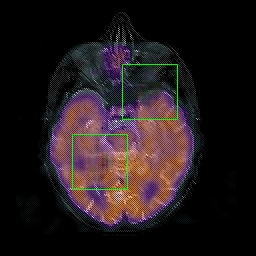}
		\end{subfigure} 
		\hfill\\
		\begin{subfigure}{1\textwidth}
			\centering
			\includegraphics[width=1.3cm,height=1.3cm]{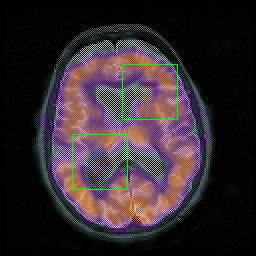}
		\end{subfigure} 
		\hfill\\
		\begin{subfigure}{1\textwidth}
			\centering
			\includegraphics[width=1.3cm,height=1.3cm]{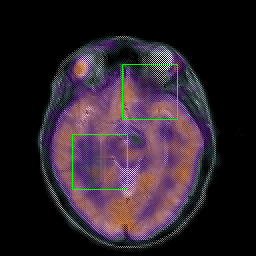}
		\end{subfigure} 
		\hfill\\
		\begin{subfigure}{1\textwidth}
			\centering
			\includegraphics[width=1.3cm,height=1.3cm]{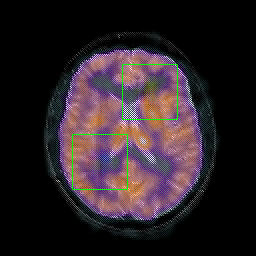}
		\end{subfigure} 
		\hfill\\
		\begin{subfigure}{1\textwidth}
			\centering
			\includegraphics[width=1.3cm,height=1.3cm]{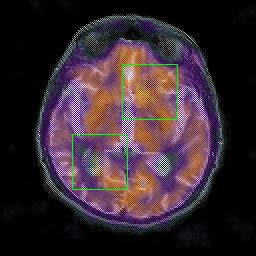}
		\end{subfigure} 
		\hfill\\
		\subcaption*{(c)}
	\end{minipage}
	\hfill
		\begin{minipage}[h]{0.06\textwidth}
		\centering
		\begin{subfigure}{1\textwidth}
			\centering
			\includegraphics[width=1.3cm,height=1.3cm]{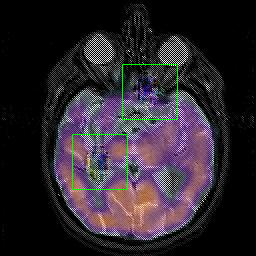}
		\end{subfigure} 
		\hfill\\
		\begin{subfigure}{1\textwidth}
			\centering
			\includegraphics[width=1.3cm,height=1.3cm]{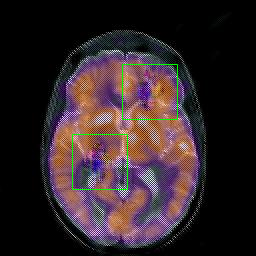}
		\end{subfigure} 
		\hfill\\
		\begin{subfigure}{1\textwidth}
			\centering
			\includegraphics[width=1.3cm,height=1.3cm]{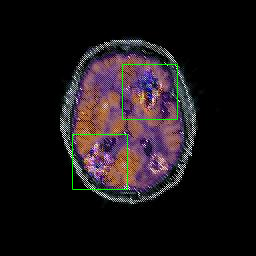}
		\end{subfigure} 
		\hfill\\
		\begin{subfigure}{1\textwidth}
			\centering
			\includegraphics[width=1.3cm,height=1.3cm]{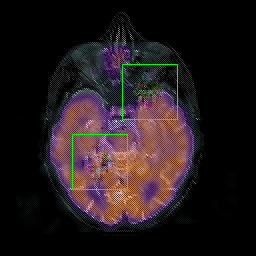}
		\end{subfigure} 
		\hfill\\
		\begin{subfigure}{1\textwidth}
			\centering
			\includegraphics[width=1.3cm,height=1.3cm]{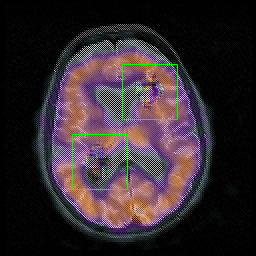}
		\end{subfigure} 
		\hfill\\
		\begin{subfigure}{1\textwidth}
			\centering
			\includegraphics[width=1.3cm,height=1.3cm]{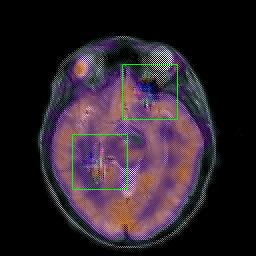}
		\end{subfigure} 
		\hfill\\
		\begin{subfigure}{1\textwidth}
			\centering
			\includegraphics[width=1.3cm,height=1.3cm]{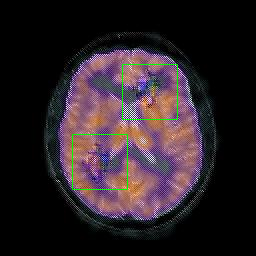}
		\end{subfigure} 
		\hfill\\
		\begin{subfigure}{1\textwidth}
			\centering
			\includegraphics[width=1.3cm,height=1.3cm]{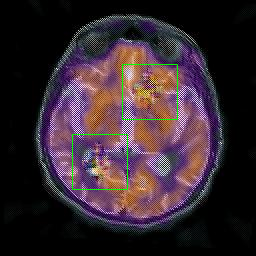}
		\end{subfigure} 
		\hfill\\
		\subcaption*{(d)}
	\end{minipage}
	\hfill
		\begin{minipage}[h]{0.06\textwidth}
		\centering
		\begin{subfigure}{1\textwidth}
			\centering
			\includegraphics[width=1.3cm,height=1.3cm]{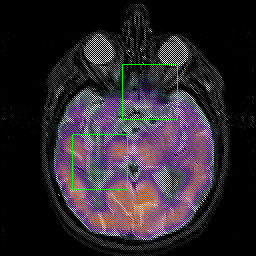}
		\end{subfigure} 
		\hfill\\
		\begin{subfigure}{1\textwidth}
			\centering
			\includegraphics[width=1.3cm,height=1.3cm]{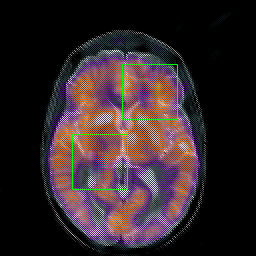}
		\end{subfigure} 
		\hfill\\
		\begin{subfigure}{1\textwidth}
			\centering
			\includegraphics[width=1.3cm,height=1.3cm]{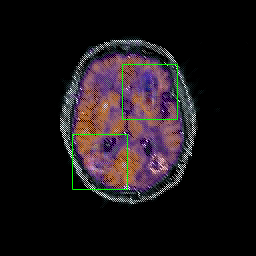}
		\end{subfigure} 
		\hfill\\
		\begin{subfigure}{1\textwidth}
			\centering
			\includegraphics[width=1.3cm,height=1.3cm]{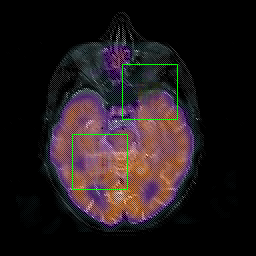}
		\end{subfigure} 
		\hfill\\
		\begin{subfigure}{1\textwidth}
			\centering
			\includegraphics[width=1.3cm,height=1.3cm]{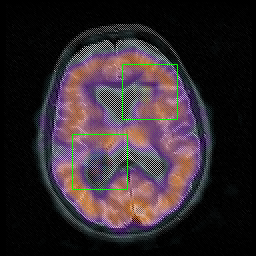}
		\end{subfigure} 
		\hfill\\
		\begin{subfigure}{1\textwidth}
			\centering
			\includegraphics[width=1.3cm,height=1.3cm]{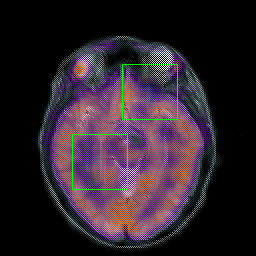}
		\end{subfigure} 
		\hfill\\
		\begin{subfigure}{1\textwidth}
			\centering
			\includegraphics[width=1.3cm,height=1.3cm]{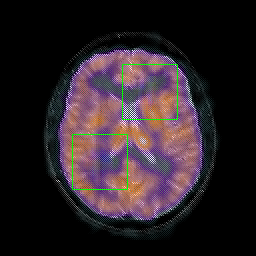}
		\end{subfigure} 
		\hfill\\
		\begin{subfigure}{1\textwidth}
			\centering
			\includegraphics[width=1.3cm,height=1.3cm]{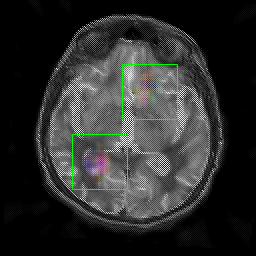}
		\end{subfigure} 
		\hfill\\
		\subcaption*{(e)}
	\end{minipage}
	\hfill
		\begin{minipage}[h]{0.06\textwidth}
		\centering
		\begin{subfigure}{1\textwidth}
			\centering
			\includegraphics[width=1.3cm,height=1.3cm]{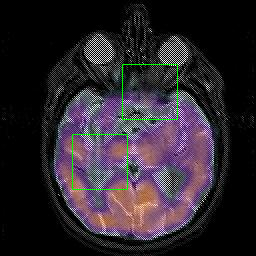}
		\end{subfigure} 
		\hfill\\
		\begin{subfigure}{1\textwidth}
			\centering
			\includegraphics[width=1.3cm,height=1.3cm]{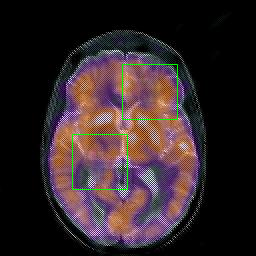}
		\end{subfigure} 
		\hfill\\
		\begin{subfigure}{1\textwidth}
			\centering
			\includegraphics[width=1.3cm,height=1.3cm]{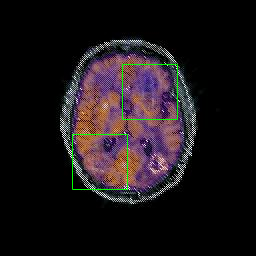}
		\end{subfigure} 
		\hfill\\
		\begin{subfigure}{1\textwidth}
			\centering
			\includegraphics[width=1.3cm,height=1.3cm]{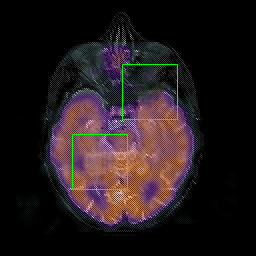}
		\end{subfigure} 
		\hfill\\
		\begin{subfigure}{1\textwidth}
			\centering
			\includegraphics[width=1.3cm,height=1.3cm]{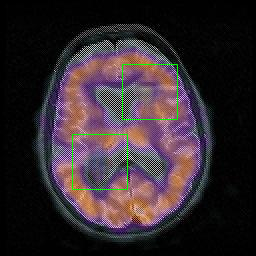}
		\end{subfigure} 
		\hfill\\
		\begin{subfigure}{1\textwidth}
			\centering
			\includegraphics[width=1.3cm,height=1.3cm]{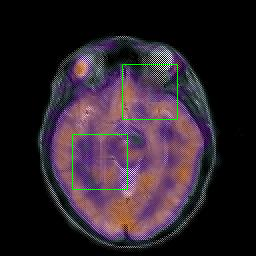}
		\end{subfigure} 
		\hfill\\
		\begin{subfigure}{1\textwidth}
			\centering
			\includegraphics[width=1.3cm,height=1.3cm]{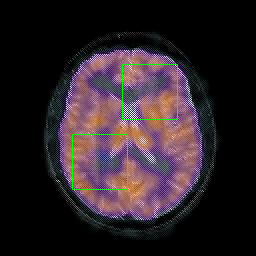}
		\end{subfigure} 
		\hfill\\
		\begin{subfigure}{1\textwidth}
			\centering
			\includegraphics[width=1.3cm,height=1.3cm]{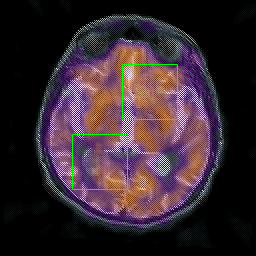}
		\end{subfigure} 
		\hfill\\
		\subcaption*{(f)}
	\end{minipage}
	\hfill
		\begin{minipage}[h]{0.06\textwidth}
		\centering
		\begin{subfigure}{1\textwidth}
			\centering
			\includegraphics[width=1.3cm,height=1.3cm]{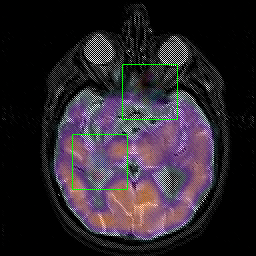}
		\end{subfigure} 
		\hfill\\
		\begin{subfigure}{1\textwidth}
			\centering
			\includegraphics[width=1.3cm,height=1.3cm]{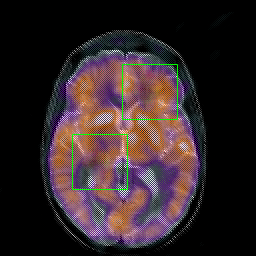}
		\end{subfigure} 
		\hfill\\
		\begin{subfigure}{1\textwidth}
			\centering
			\includegraphics[width=1.3cm,height=1.3cm]{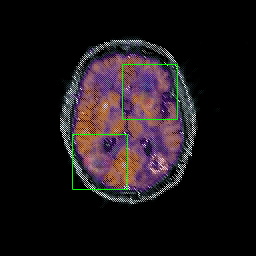}
		\end{subfigure} 
		\hfill\\
		\begin{subfigure}{1\textwidth}
			\centering
			\includegraphics[width=1.3cm,height=1.3cm]{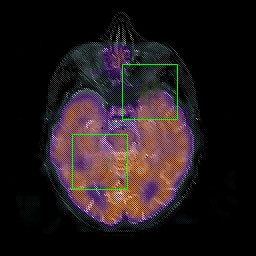}
		\end{subfigure} 
		\hfill\\
		\begin{subfigure}{1\textwidth}
			\centering
			\includegraphics[width=1.3cm,height=1.3cm]{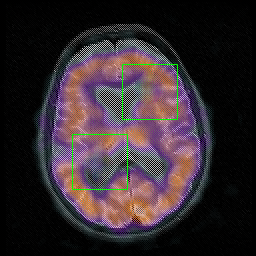}
		\end{subfigure} 
		\hfill\\
		\begin{subfigure}{1\textwidth}
			\centering
			\includegraphics[width=1.3cm,height=1.3cm]{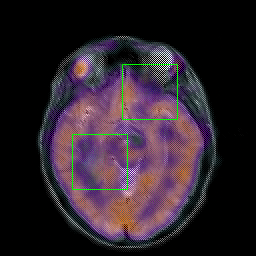}
		\end{subfigure} 
		\hfill\\
		\begin{subfigure}{1\textwidth}
			\centering
			\includegraphics[width=1.3cm,height=1.3cm]{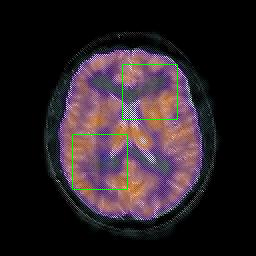}
		\end{subfigure} 
		\hfill\\
		\begin{subfigure}{1\textwidth}
			\centering
			\includegraphics[width=1.3cm,height=1.3cm]{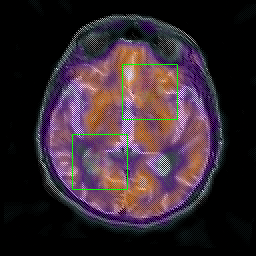}
		\end{subfigure} 
		\hfill\\
		\subcaption*{(g)}
	\end{minipage}
	\hfill
	\begin{minipage}[h]{0.06\textwidth}
		\centering
		\begin{subfigure}{1\textwidth}
			\centering
			\includegraphics[width=1.3cm,height=1.3cm]{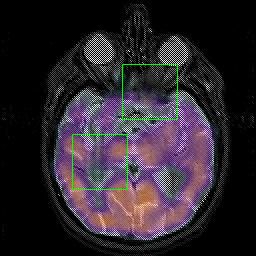}
		\end{subfigure} 
		\hfill\\
		\begin{subfigure}{1\textwidth}
			\centering
			\includegraphics[width=1.3cm,height=1.3cm]{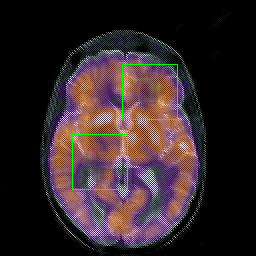}
		\end{subfigure} 
		\hfill\\
		\begin{subfigure}{1\textwidth}
			\centering
			\includegraphics[width=1.3cm,height=1.3cm]{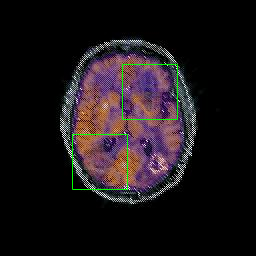}
		\end{subfigure} 
		\hfill\\
		\begin{subfigure}{1\textwidth}
			\centering
			\includegraphics[width=1.3cm,height=1.3cm]{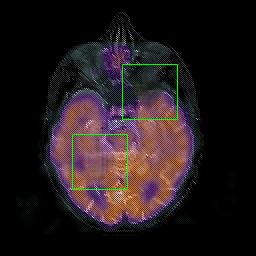}
		\end{subfigure} 
		\hfill\\
		\begin{subfigure}{1\textwidth}
			\centering
			\includegraphics[width=1.3cm,height=1.3cm]{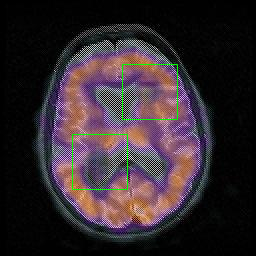}
		\end{subfigure} 
		\hfill\\
		\begin{subfigure}{1\textwidth}
			\centering
			\includegraphics[width=1.3cm,height=1.3cm]{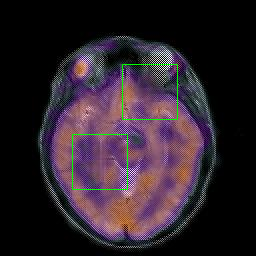}
		\end{subfigure} 
		\hfill\\
		\begin{subfigure}{1\textwidth}
			\centering
			\includegraphics[width=1.3cm,height=1.3cm]{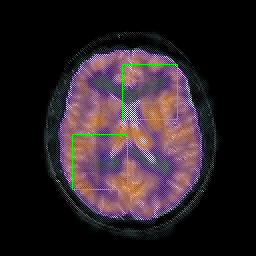}
		\end{subfigure} 
		\hfill\\
		\begin{subfigure}{1\textwidth}
			\centering
			\includegraphics[width=1.3cm,height=1.3cm]{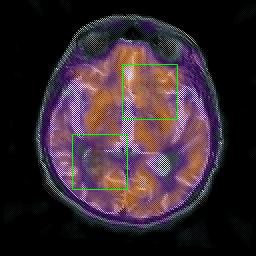}
		\end{subfigure} 
		\hfill\\
		\subcaption*{(h)}
	\end{minipage}
	\hfill
		\begin{minipage}[h]{0.06\textwidth}
		\centering
		\begin{subfigure}{1\textwidth}
			\centering
			\includegraphics[width=1.3cm,height=1.3cm]{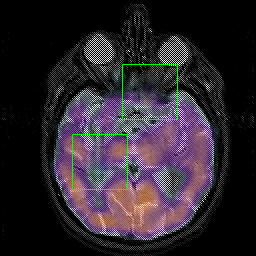}
		\end{subfigure} 
		\hfill\\
		\begin{subfigure}{1\textwidth}
			\centering
			\includegraphics[width=1.3cm,height=1.3cm]{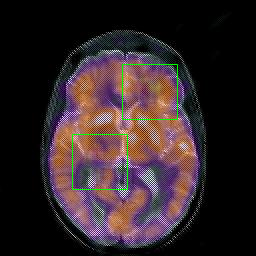}
		\end{subfigure} 
		\hfill\\
		\begin{subfigure}{1\textwidth}
			\centering
			\includegraphics[width=1.3cm,height=1.3cm]{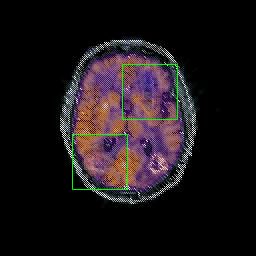}
		\end{subfigure} 
		\hfill\\
		\begin{subfigure}{1\textwidth}
			\centering
			\includegraphics[width=1.3cm,height=1.3cm]{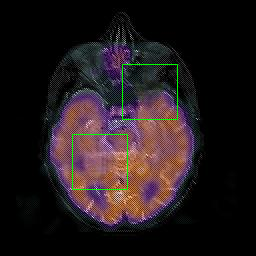}
		\end{subfigure} 
		\hfill\\
		\begin{subfigure}{1\textwidth}
			\centering
			\includegraphics[width=1.3cm,height=1.3cm]{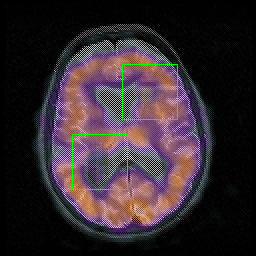}
		\end{subfigure} 
		\hfill\\
		\begin{subfigure}{1\textwidth}
			\centering
			\includegraphics[width=1.3cm,height=1.3cm]{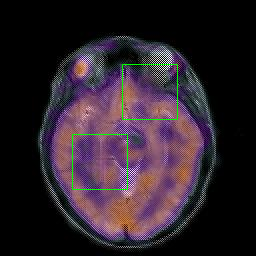}
		\end{subfigure} 
		\hfill\\
		\begin{subfigure}{1\textwidth}
			\centering
			\includegraphics[width=1.3cm,height=1.3cm]{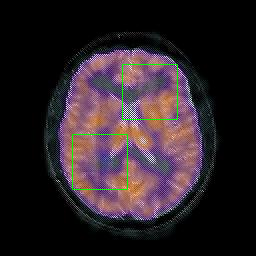}
		\end{subfigure} 
		\hfill\\
		\begin{subfigure}{1\textwidth}
			\centering
			\includegraphics[width=1.3cm,height=1.3cm]{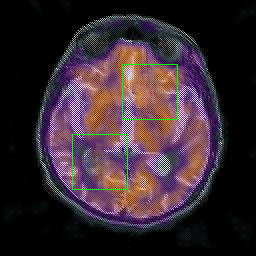}
		\end{subfigure} 
		\hfill\\
		\subcaption*{(i)}
	\end{minipage}
	\hfill
		\begin{minipage}[h]{0.06\textwidth}
		\centering
		\begin{subfigure}{1\textwidth}
			\centering
			\includegraphics[width=1.3cm,height=1.3cm]{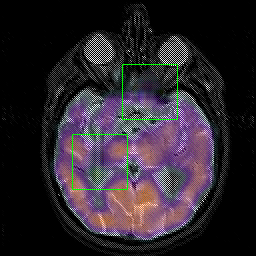}
		\end{subfigure} 
		\hfill\\
		\begin{subfigure}{1\textwidth}
			\centering
			\includegraphics[width=1.3cm,height=1.3cm]{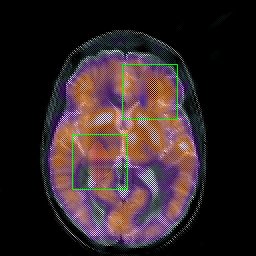}
		\end{subfigure} 
		\hfill\\
		\begin{subfigure}{1\textwidth}
			\centering
			\includegraphics[width=1.3cm,height=1.3cm]{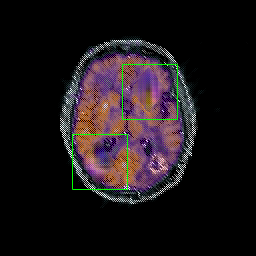}
		\end{subfigure} 
		\hfill\\
		\begin{subfigure}{1\textwidth}
			\centering
			\includegraphics[width=1.3cm,height=1.3cm]{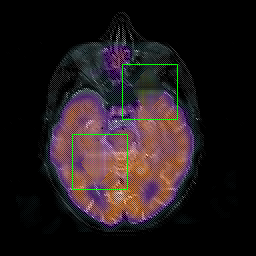}
		\end{subfigure} 
		\hfill\\
		\begin{subfigure}{1\textwidth}
			\centering
			\includegraphics[width=1.3cm,height=1.3cm]{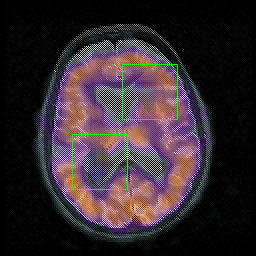}
		\end{subfigure} 
		\hfill\\
		\begin{subfigure}{1\textwidth}
			\centering
			\includegraphics[width=1.3cm,height=1.3cm]{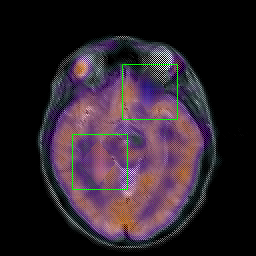}
		\end{subfigure} 
		\hfill\\
		\begin{subfigure}{1\textwidth}
			\centering
			\includegraphics[width=1.3cm,height=1.3cm]{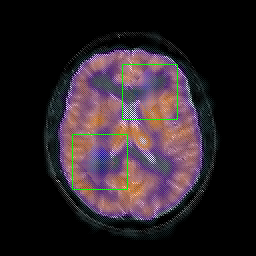}
		\end{subfigure} 
		\hfill\\
		\begin{subfigure}{1\textwidth}
			\centering
			\includegraphics[width=1.3cm,height=1.3cm]{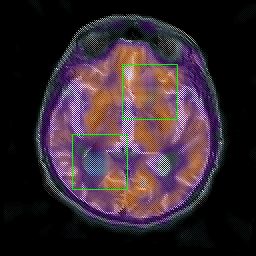}
		\end{subfigure} 
		\hfill\\
		\subcaption*{(j)}
	\end{minipage}
	\hfill
	\begin{minipage}[h]{0.06\textwidth}
		\centering
		\begin{subfigure}{1\textwidth}
			\centering
			\includegraphics[width=1.3cm,height=1.3cm]{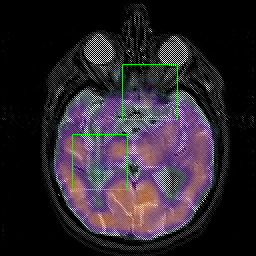}
		\end{subfigure} 
		\hfill\\
		\begin{subfigure}{1\textwidth}
			\centering
			\includegraphics[width=1.3cm,height=1.3cm]{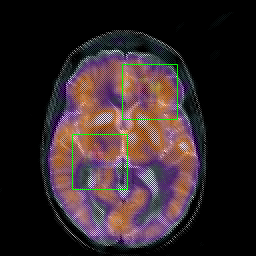}
		\end{subfigure} 
		\hfill\\
		\begin{subfigure}{1\textwidth}
			\centering
			\includegraphics[width=1.3cm,height=1.3cm]{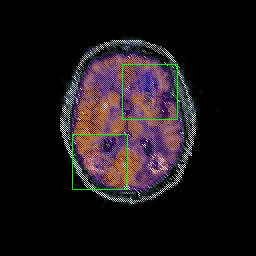}
		\end{subfigure} 
		\hfill\\
		\begin{subfigure}{1\textwidth}
			\centering
			\includegraphics[width=1.3cm,height=1.3cm]{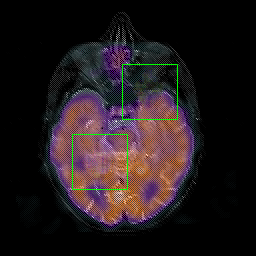}
		\end{subfigure} 
		\hfill\\
		\begin{subfigure}{1\textwidth}
			\centering
			\includegraphics[width=1.3cm,height=1.3cm]{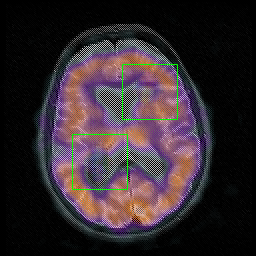}
		\end{subfigure} 
		\hfill\\
		\begin{subfigure}{1\textwidth}
			\centering
			\includegraphics[width=1.3cm,height=1.3cm]{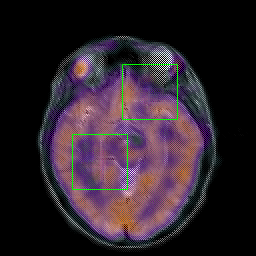}
		\end{subfigure} 
		\hfill\\
		\begin{subfigure}{1\textwidth}
			\centering
			\includegraphics[width=1.3cm,height=1.3cm]{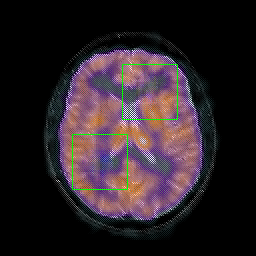}
		\end{subfigure} 
		\hfill\\
		\begin{subfigure}{1\textwidth}
			\centering
			\includegraphics[width=1.3cm,height=1.3cm]{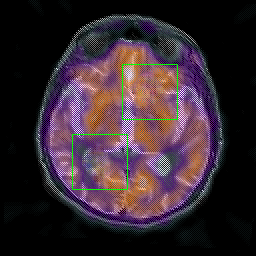}
		\end{subfigure} 
		\hfill\\
		\subcaption*{(k)}
	\end{minipage}
	\hfill
	\begin{minipage}[h]{0.06\textwidth}
		\centering
		\begin{subfigure}{1\textwidth}
			\centering
			\includegraphics[width=1.3cm,height=1.3cm]{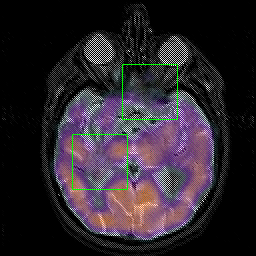}
		\end{subfigure} 
		\hfill\\
		\begin{subfigure}{1\textwidth}
			\centering
			\includegraphics[width=1.3cm,height=1.3cm]{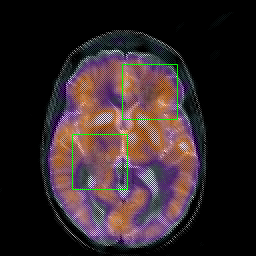}
		\end{subfigure} 
		\hfill\\
		\begin{subfigure}{1\textwidth}
			\centering
			\includegraphics[width=1.3cm,height=1.3cm]{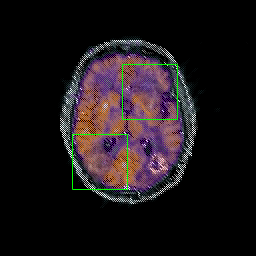}
		\end{subfigure} 
		\hfill\\
		\begin{subfigure}{1\textwidth}
			\centering
			\includegraphics[width=1.3cm,height=1.3cm]{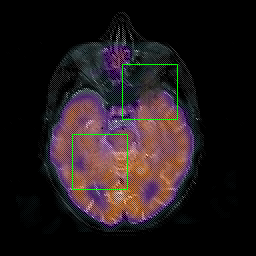}
		\end{subfigure} 
		\hfill\\
		\begin{subfigure}{1\textwidth}
			\centering
			\includegraphics[width=1.3cm,height=1.3cm]{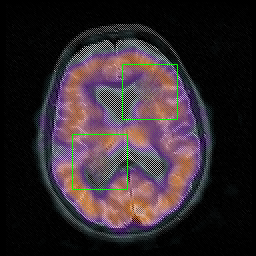}
		\end{subfigure} 
		\hfill\\
		\begin{subfigure}{1\textwidth}
			\centering
			\includegraphics[width=1.3cm,height=1.3cm]{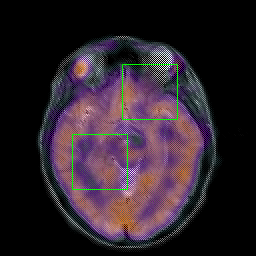}
		\end{subfigure} 
		\hfill\\
		\begin{subfigure}{1\textwidth}
			\centering
			\includegraphics[width=1.3cm,height=1.3cm]{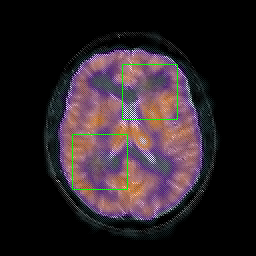}
		\end{subfigure} 
		\hfill\\
		\begin{subfigure}{1\textwidth}
			\centering
			\includegraphics[width=1.3cm,height=1.3cm]{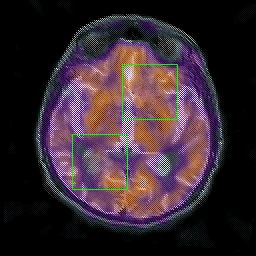}
		\end{subfigure} 
		\hfill\\
		\subcaption*{(l)}
	\end{minipage}
	\hfill\\
		\begin{minipage}{1\linewidth}
		\begin{table}[H]
				\centering
			\resizebox{15.5cm}{2cm}{
				\begin{tabular}{|c|c|c|c|c|c|c|c|c|c|c|}		
					\hline
					Methods:& IRLNM-QR  & WNNM  & MC-NC  & TNNR & TNN-SR &LRQMC & LRQA-G  & QLNF &TQLNA& QQR-QNN-SR  
					\\ \toprule
					\hline
					Images:  &\multicolumn{10}{c|}{Random block missing}\\
					\hline
					Image (9)&31.268/0.988	&27.331/0.982	&32.451/0.989	&33.237/0.836	&33.004/0.989	&33.049/0.990
					&33.283/0.990	&32.957/\textbf{0.991}	&32.656/0.963	&\textbf{33.458}/0.990\\
					Image (10)&32.552/0.992	&29.195/0.988	&32.420/\textbf{0.993}	&33.073/0.914	&33.485/0.992	&32.948/0.992	&33.075/\textbf{0.993}	&30.806/0.991   &32.545/0.956	&\textbf{33.658}/\textbf{0.993}\\
					Image (11)&30.511/0.987	&26.068/0.980	&31.107/0.988	&32.015/0.946	&32.888/0.990	&31.835/0.988	&31.613/0.988	&31.318/0.989	&31.075/0.985	&\textbf{33.522}/\textbf{0.991}\\
					Image (12)&33.247/0.990	&30.071/0.984	&34.193/0.991	&34.126/0.824	&35.141/\textbf{0.992}	&33.072/0.989	&33.181/0.989	&32.871/0.988	&33.054/0.969	&\textbf{35.173}/\textbf{0.992}\\
					Image (13)&30.523/0.990	&28.637/0.986	&31.355/0.990	&31.589/0.986	&32.563/\textbf{0.991}	&31.306/0.989	&31.702/0.990	&30.447/0.988	&31.682/0.988	&\textbf{32.636}/\textbf{0.991}\\
					Image (14)&35.224/0.993	&30.661/0.987	&35.171/0.993	&35.240/0.921	&35.691/\textbf{0.994}	&35.657/0.993	&35.660/\textbf{0.994}	&34.562/0.993	&35.263/0.982	&\textbf{35.749}/\textbf{0.994}\\
					Image (15)&34.387/0.993	&29.966/0.986	&35.247/\textbf{0.994}	&34.636/0.962	&35.542/0.992	&34.813/0.992	&34.820/0.993	&33.725/0.993	&34.756/0.993	&\textbf{35.775}/0.993\\
					Image (16)&29.627/0.991	&26.839/0.985	&29.759/0.990	&30.006/0.859	&31.514/\textbf{0.992}	&29.847/0.988	&30.158/0.989	&30.432/0.990	&30.194/0.980	&\textbf{31.907}/\textbf{0.992}\\
					\hline
					Aver. &32.241	&28.691			&32.792		&32.943		&33.867		&32.764		&32.867		&22.708		&32.668		&\textbf{34.107} \\ \toprule		
			\end{tabular}}
			\label{Table5}
		\end{table}	
	   \subcaption*{(m)}
       \end{minipage}
   \setcounter{figure}{8} 	
	\caption{\label{fig:blockmissing} (a) Original color medical images. (b) The observed color medical images (block missing). (c)-(l) are the completion results of IRLNM-QR, WNNM, MC-NC, TNNR, TNN-SR, LRQMC, LRQA-G, QLNF, TQLNA, and QQR-QNN-SR, respectively. (m) Quantitative quality indexes (PSNR/SSIM) of different methods on color medical images (2 random rhombus blocks missing).}
\end{figure}

\section{Conclusions}
 Using quaternion matrix analysis, we proposed a new Tri-Factorization of a quaternion matrix called CQSVD-QQR, which is based on the quaternion QR decomposition and can approximate the QSVD of a quaternion matrix. The coupling between color channels may be handled naturally and the color information of color images is better retained when color pixels are regarded as vector units rather than scalars in the quaternion representation for color images. Therefore, within the context of the new quaternion matrix trifactorization, we developed a novel method for color image completion in this study, using the tool of quaternion representation of color images. Additionally, when developing the model, we take into account both the low rank of the image and the sparse prior information about the image. The method is based on CQSVD-QQR, quaternion nuclear norm, and sparse regularizer, called QQR-QNN-SR. We solve this problem under the quaternion ADMM framework. Experimental results on color images and color medical images show that the proposed method exhibits superiority in both numerical results and visual effects compared to several state-of-the-art methods. 
 
\section*{Acknowledgments}
This work was supported by University of Macau (MYRG2019-00039-FST), Science and Technology Development Fund, Macao S.A.R (FDCT/0036/2021/AGJ), and Science and Technology Planning Project of Guangzhou City, China (Grant No. 201907010043).

\section*{Declaration of competing interest}

The authors declare that they have no known competing financial interests or personal relationships that could have appeared to influence the work reported in this paper.

\normalem
\bibliographystyle{unsrt}
\bibliography{sample}

\end{document}